%% file: main.tex
\documentclass[a4paper,12pt,dvipsnames]{article}
\usepackage[affil-sl]{authblk}
\usepackage[left=1in, bottom=1in, right=1in, top=1.0in]{geometry}
\usepackage{setspace}
\DeclareOldFontCommand{\bf}{\normalfont\bfseries}{\mathbf} %
\usepackage{pgfplots}

\usepackage{amsmath, amsfonts, amssymb}
\usepackage{amsthm}
\usepackage{pifont}
\usepackage[round,longnamesfirst]{natbib}
\usepackage{verbatim}
\usepackage{graphicx, subcaption}
\usepackage{mathtools}
\usepackage[mathscr]{euscript}
\usepackage[labelfont=bf,labelsep=colon]{caption}
\usepackage[capposition=top]{floatrow}
\definecolor{myLightGray}{RGB}{191,191,191}
\definecolor{myGray}{RGB}{160,160,160}
\definecolor{myDarkGray}{RGB}{144,144,144}
\definecolor{myDarkRed}{RGB}{167,114,115}
\definecolor{myRed}{RGB}{255,58,70}
\definecolor{myGreen}{RGB}{1,37,110}
\definecolor{myDarkCerulean}{RGB}{16,80,113}

\usepackage[colorlinks=True, citecolor=myDarkCerulean, linkcolor=black, urlcolor=blue, hypertexnames=false]{hyperref}
\usepackage[inline]{enumitem}
\usepackage{eurosym}
\usepackage{booktabs}
\usepackage[page]{appendix}
\usepackage{xpatch}
\usepackage{multirow}
\usepackage{dsfont}

\makeatletter
\xpatchcmd{\paragraph}{3.25ex \@plus1ex \@minus.2ex}{3pt plus 1pt minus 1pt}{\typeout{success!}}{\typeout{failure!}}
\makeatother

\newcommand{\underb}[1]{\underbar{\itshape #1}}
\setcitestyle{citesep={,}}

\providecommand{\sym}[1]{\ifmmode^{#1}\else\(^{#1}\)\fi}

\newtheorem{prop}{Proposition}[section]
\newtheorem{cor}{Corollary}[section]
\newtheorem{lem}{Lemma}[section]
\newtheorem{defi}{Definition}[section]
\newtheorem{ppert}{Property}[section]

\usepackage{tikz}
\usetikzlibrary{decorations.pathreplacing}
\usetikzlibrary{graphs}
\usetikzlibrary{chains,shapes.multipart}
\usetikzlibrary{shapes,calc,fit}
\usetikzlibrary{automata,positioning}

\usepackage{etoolbox,xstring,xspace}
\makeatletter

\AtBeginDocument{%
\let\origref\ref
\renewcommand*\ref[1]{%
  \origref{#1}\xlabel{#1}}
}
\newrobustcmd*\xlabel[1]{%
   \ifcsdef{siteref@doc@#1}{}{\csgdef{siteref@doc@#1}{,}}%
    \@bsphack%
    \begingroup
       \csxdef{siteref@doc@#1}{\csuse{siteref@doc@#1},\thepage}%
         \protected@write\@auxout{}%
        {\string\SiteRef{siteref@#1}{\csuse{siteref@doc@#1}}}%
     \endgroup
     \@esphack%
}

\newrobustcmd*\SiteRef[2]{\csgdef{#1}{#2}}

\newrobustcmd*\xref[1]{%
\ifcsundef{siteref@#1}{%
     \@latex@warning@no@line{Label `#1' not defined}
     }{%
    \begingroup
      \StrGobbleLeft{\csuse{siteref@#1}}{2}[\@tempa]\relax%
      \def\@tempb{}%
      \@tempcnta=0\relax%
      \@tempcntb=\@ne\relax%
      \def\do##1{\advance\@tempcnta\@ne}%
      \expandafter\docsvlist\expandafter{\@tempa}%
       \def\do##1{%
         \ifnum\@tempcntb=\@tempcnta\relax%
            \hyperpage{##1}%
         \else
            \hyperpage{##1},%
          \fi%
          \advance\@tempcntb\@ne
       }%
       [\expandafter\docsvlist\expandafter{\@tempa}]\xspace%
    \endgroup
   }%
}

\usepackage{scrwfile}
\TOCclone[\contentsname]{toc}{atoc}
\newcommand\StartAppendixEntries{}
\AfterTOCHead[toc]{%
  \renewcommand\StartAppendixEntries{\value{tocdepth}=-10000\relax}%
}
\AfterTOCHead[atoc]{%
  \edef\maintocdepth{\the\value{tocdepth}}%
  \value{tocdepth}=-10000\relax%
  \renewcommand\StartAppendixEntries{\value{tocdepth}=\maintocdepth\relax}%
}
\newcommand*\appendixwithtoc{%
  \appendix
  \addtocontents{toc}{\protect\StartAppendixEntries}
  \listofatoc
}

\usepackage{bibunits}

\usepackage{threeparttable}

\begin{document}

\begin{bibunit}
\input{draft_v9_body}

{
\singlespacing
\small
\input{main.bbl}
}
\end{bibunit}

\pagebreak

\setcounter{page}{1}
\setcounter{footnote}{0}

\numberwithin{equation}{section}
\numberwithin{figure}{section}
\numberwithin{table}{section}

\renewcommand{\theequation}{\thesection.\arabic{equation}}
\renewcommand{\thetable}{\thesection.\arabic{table}}
\renewcommand{\thefigure}{\thesection.\arabic{figure}}

\begin{center}
\section*{\large{Appendix for \\ ``The Collapse of Human Capital Ladders in Recessions''}}\label{appendix:online_appendix}
\end{center}
\smallskip 
\begin{center}
Edoardo Maria Acabbi, Andrea Alati and Luca Mazzone
\end{center}
\smallskip 
\begin{center}
January 2026
\end{center}

\appendixwithtoc
\begin{bibunit}

\input{draft_v9_FullAppendix}
{
\clearpage
\singlespacing
\small
\def\printappendixreferences{}
\input{main.bbl}
}
\end{bibunit}

\end{document}

%% file: draft_v9_body.tex
\title{\large{\textbf{
The Collapse of Human Capital Ladders in Recessions}}\thanks{The authors thank Job Boerma, Gideon Bornstein, Gabriel Chodorow Reich, Alessandro Dovis,  Simone Ferro,  Cecilia Garcia-Pe\~nalosa, Christoph Hedtrich, Kyle Herkenhoff, Marianna Kudlyak, Lien Laureys, Paolo Martellini, Antonio Mele, Silvia Miranda-Agrippino, Kurt Mitman, Simon Mongey, Facundo Piguillem, Victor Rios-Rull, Edouard Schaal, Silvia Vannutelli, Gianluca Violante, Liangjie Wu and the participants at 
numerous seminars and workshops for useful comments and conversations. Acabbi gratefully acknowledges the financial support from the Comunidad de Madrid (Programa Excelencia para el Profesorado Universitario, convenio con Universidad Carlos III de Madrid, V Plan Regional de Investigación Científica e Innovación Tecnológica), the Spanish Ministry of Science and Innovation (projects PID2020-114108GB-I00 and PID2022-137707NB-I00), and the Fundación Ramón Areces (projects CISP20S12348 and CISP22S15759). 
This study uses anonymous data from the Italian Social Security Institute (INPS). Data access was provided as part of the VisintINPS Scholars initiative. We are grateful to the staff of Direzione Centrale Studi e Ricerche at INPS. Findings and conclusions expressed are solely those of the authors and do not represent the views of INPS, the Bank of England, or the IMF. All errors are our own.
\newline Acabbi, corresponding author: edoardo.maria.acabbi@uni-mannheim.de \\ Alati: andrea.alati@bankofengland.co.uk\\ Mazzone: luca.mazzone@umontreal.ca} }
\author[1]{{\fontsize{12}{10.8}\selectfont Edoardo Maria Acabbi}}
\author[2]{{\fontsize{12}{10.8}\selectfont Andrea Alati}}
\author[3]{{\fontsize{12}{10.8}\selectfont Luca Mazzone}}
\affil[1]{University of Mannheim}
\affil[2]{Bank of England}
\affil[3]{University of Montr\'eal}

\date{\fontsize{12}{10.8}\selectfont \vspace{-1em} January 2026}

\frenchspacing
\tolerance=1
\emergencystretch=\maxdimen
\hyphenpenalty=10000
\hbadness=10000
\newgeometry{top=2em,bottom=2em,left=1.1in,right=1.1in}
\maketitle
 
\vspace{-1cm}
\begin{abstract}
\input{draft_v9_abstract}

\end{abstract}
\vspace{1em}
\thispagestyle{empty}
\restoregeometry

\section{Introduction}
\setcounter{page}{1}
\input{draft_v9_intro}

\section{Firms matter for human capital accumulation}\label{sect2}

We leverage unique longitudinal matched employer-employee data from Italy to shed new light on the role played by firms in shaping the speed of workers' human capital accumulation and career progression. 
We proxy human capital accumulation by looking at persistent returns to experience for workers' earnings. 
We first highlight that experience-earnings profiles are steeper for individuals employed in highly productive firms, and their effect persists after individuals have moved away from those firms, even after displacement from the previous job. 
Second, we show that past firm quality matters for re-employment wages. Workers finding a job after an unemployment spell earn more the greater the productivity of their previous firm, even after controlling for their past wage or the duration of their unemployment spell.
Lastly, we show that workers tend to match with firms with lower productivity in recession periods and that the job ladder collapses in recessions \citep{haltiwanger_cyclical_2025}.

We interpret these results as evidence that workers accumulate human capital at different rates based on the productivity of the firms they work for, and this effect varies over the cycle. As a consequence, in the rest of the paper we incorporate firm-dependent human capital accumulation in a labor market search model with heterogeneous workers and firms and aggregate risk, to evaluate the aggregate effect of distortions in sorting and human capital production along the business cycle. %

\subsection{Data} 
We rely on two main data sources provided by Italian Social Security Institute (INPS): i) data on the universe of private employment relationships in Italy from 1996 to 2018 (\textit{Uniemens} dataset, \citealt{inps_data}) and ii) balance sheet data for incorporated Italian firms from 1994 to 2019 (\textit{Cerved} dataset, \citealt{cerved}). On the workers' side, we have information on annual incomes, months worked, contract types, a coarse measure of qualifications, education levels and usual demographic information. We do not observe unemployment benefits in this dataset, and cannot distinguish between unemployment and non-employment when a worker is not observed in the data. On the firms' side, we have access to standard balance sheet and income statements with information on revenues, value added, payroll size, and sector of activity. We report more details on the data and sample construction in \textbf{Appendix~\ref{sect:app_data}}.

\subsection{Firm dependent returns to experience}
Measuring human capital accumulation is challenging. In this section, we proxy human capital accumulation by looking at potentially different returns to experience for workers' earnings across firms \citep{mion2022dream, arellano2022}. The key intuition is that, if workers accumulate more human capital in a certain firm type with respect to the other, this should translate into persistent changes in their lifetime earnings trajectory, regardless of unemployment events and job mobility.

We use administrative data on the universe of Italian workers and firms' employment, as well as firms' balance sheet data, spanning 1994 to 2019. We assign each firm, $f$, to a class according to its productivity. Firm classes are defined as quintiles of yearly value added per employee. Then, for each worker $i$ we construct a measure of experience in each firm-class and we estimate the following wage equation,
\begin{equation}\label{eqn:experience_profiles}
\log(w_{i,t}) = \alpha_i + \alpha_{j(i,t)} + \sum_{c = 1}^5 \beta_c e^c_{i,t} \times \mathbb{I}\{f(i,t) \in c \} + \sum_{c = 1}^5 \gamma_{c} e^c_{i,t} \times \mathbb{I}\{f(i,t) \notin c\} + \mathbf{X}'_{i,t} \mathbf{\theta} + \varepsilon_{i,t},
\end{equation}
where $w_{i,t}$ are monthly earnings\footnote{Throughout the section we use monthly earnings as the variable of reference for labor earnings. Results are qualitatively analogous when using annual earnings.}, $e_{i,t}^c$ are the number of years worker $i$ has worked in firms belonging to productivity class $c$ up to year $t$, $\mathbb{I}\{\cdot\}$ are indicator functions that take value one when the worker's current employer, $f(i,t)$, belongs to quintile $c$, and $\mathbf{X}_{i,t}$ are a set of controls that include age, sector-year and contract type (temporary versus open-ended, full-time versus part-time) fixed effects. In the spirit of the AKM literature \citep{Abowd1999a}, all regressions control for worker and firm fixed effects $\alpha_i$ and $\alpha_{j(i,t)}$. The firm fixed effect is estimated by clustering firms in one of 10 latent firm types obtained by applying the K-means measure reduction method on the data as in \cite{bonhomme2015}, using the earnings decile distribution for clustering similarity. We denote firm clusters by $j$.

The main coefficients of interest are $\{\beta_c,\gamma_c\}_{c=1}^5$. They represent the returns to experience when working in a firm in the $c^{th}$ quintile of the labor productivity distribution, with $\beta_c$ capturing returns to experience when the worker is \textit{currently} employed in firms belonging to a certain class, while $\gamma_c$ captures the returns of having worked in a given firm class at some point in the worker's career. Distinguishing between the two cases is useful to assess how portable is the experience gained in any firm-class.\footnote{In {Appendix~\ref{sect:empirics}} we also report the estimate from the same class of regressions without distinguishing between current and past firm classes' experience returns, splitting by high-school vs higher education, allowing for quadratic returns to experience. The latter results highlight, consistently with our model, that firm-specific marginal experience returns are diminishing \citep{gregory2023}.} 
If one year of experience is more valuable when accrued in a more productive firm, we would expect $ \{\beta_c,\gamma_c\}_{c=1}^5$ to be positive and increasing with firms' productivity. A positive relationship between firm productivity and wages could also be consistent with employers being forced to match outside offers, as in sequential bargaining models \`a la \cite{postel-vinay2002},  while workers are not becoming more productive over time.  In order to rule out this alternative explanation, we estimate \textbf{Equation~\eqref{eqn:experience_profiles}} also on a sample of exogenously displaced workers, only for the years following their displacement/mass layoff event. In order to identify mass layoff events we follow the logic of the classical event study in \cite{jacobson1993} and the criteria of \cite{bertheau2023} (see \textbf{Appendix~\ref{sect:displ}} for a precise definition).%
 \;By virtue of the displacement, the bargaining position of workers is equalized to their unemployment value and the returns to experience should better capture the premia associated with an increase in skills.\footnote{Greater experience-earnings profiles in highly productive firms could also be due to the fact that high productivity firms are better at screening low quality workers. We look at the separation rates across worker qualities by firm productivity--\textbf{Figure~\ref{fig:seprates}}--and we do not observe a differential gradient in the separation rates of low quality workers in high quality firms.} 

\paragraph{Results.} \textbf{Figure~\ref{fig:experience_profiles}} plots the returns to experience for the full sample of workers, Panel~(a), and for the subsample of displaced workers, Panel~(b). Given the sample size, coefficients are estimated with an extreme degree of precision. Coefficients are reported in \textbf{Table~\ref{tab:controlling_unemp}} in \textbf{Appendix \ref{sect:empirics}}. In both cases, the returns to experience are increasing in firm classes. In particular, one year of experience in a firm belonging to the top 20\% of the productivity distribution is associated with an increase in earnings more than twice as large than the one implied by working an extra year for a firm in the bottom 20\%, 1.6\% versus 0.6\%. While the level of returns is lower once we remove returns to bargaining positions, the positive gradient in the wage experience profiles is preserved also in the sample of displaced workers, reinforcing the interpretation that experience in firms belonging to higher productivity classes benefits workers' earnings beyond the potential availability of better outside offers.%
The results are robust to using quintiles of firm-cluster fixed effects estimated using a rolling-AKM algorithm \citep{lachowska_firm_2023}, in which we address the limited mobility bias by clustering firms ex-ante on their labor earnings distribution as in \cite{bonhomme2015, bonhomme2019}. The estimated $\beta$ and $\gamma$ coefficients are of similar magnitude to each other across both samples, with returns to experience outside of the firm class of reference ($\gamma$) amounting to between $70\%$ and $90\%$ of those within firm class ($\beta$). Results are displayed in \textbf{Appendix~\ref{sect:empirics}}, where we also report robustness results including the duration in months of the unemployment spell between displacement and re-employment as an additional control. %
In sum, evidence suggests that experience returns are highly portable across firms, which is entirely consistent with the persistent accumulation of human capital on the job for workers.%

We take the results as evidence that workers' returns to being employed in different firms are highly heterogeneous and depend on firms' productivity. Firm characteristics and workers' employment histories therefore can have significant effects for the long-term dynamics of their careers. %

\begin{figure}
\centering
\caption{Experience profiles\label{fig:experience_profiles}}
\subcaptionbox{Full sample \label{fig:experience_profiles_full}}{\includegraphics[width=0.5\textwidth]{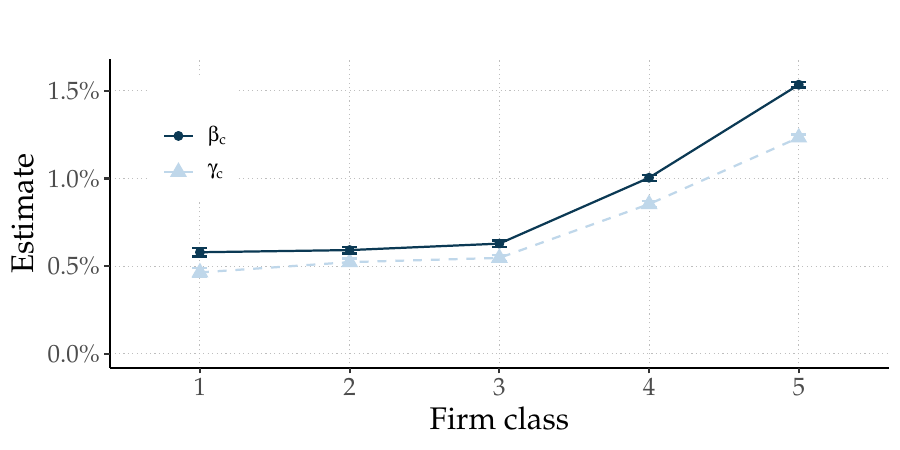}}%
\subcaptionbox{Displaced workers\label{fig:experience_profiles_disp}}{\includegraphics[width=0.5\textwidth]{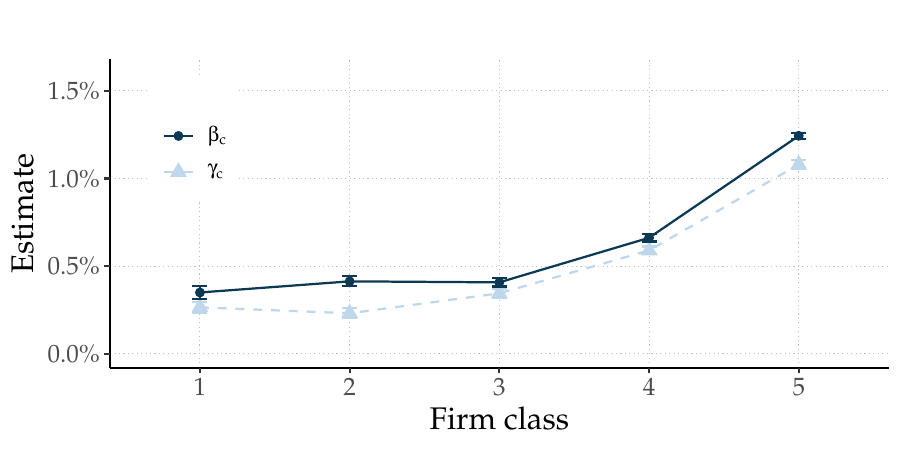}}
\vspace{-1em}
\floatfoot{\textbf{Note:} The figure reports the returns to experience estimated from Equation~\eqref{eqn:experience_profiles}. The full sample includes all workers employed in private firms in Italy and includes approximately 98M observations (total worker-year). Displaced workers are identified from mass layoffs and the sample includes approximately 5M observations. Error bars are 99.9\% confidence intervals with standard errors clustered at the worker level. When estimating the returns for the sample of displaced workers we include the worker and firm-class fixed effects estimated on the full sample as regressors. We report more details on the sample construction in Appendix~\ref{sect:app_data}.}
\end{figure}

\subsection{Persistent effects of firm quality on workers' earnings}

The results presented in \textbf{Figure~\ref{fig:experience_profiles}} highlight how experience in differently productive firms persistently influences labor earnings. We now provide complementary evidence that this effect is present even when workers move to new employers from unemployment. While this context allow us to control for the effect of differences in bargaining across workers, it is still exposed to separate issues.
Conditioning on unemployment spell duration  might induce selection on workers' quality out of unemployment. Not controlling for it, on the other hand, might confound the effect with duration dependence dynamics. Workers with lower skills tend to experience greater unemployment duration and to work in lower quality firms. As their human capital might deteriorate in unemployment, we might spuriously measure lower wage effects for low productive firms only because the workers themselves have lost human capital.\footnote{Recent evidence on labor displacement from \cite{bertheau2023} shows that for the Italian labor market the loss in earnings is mostly driven by lower employment probabilities.}  

We thus use two different subsamples of the data to address selection and negative duration dependence \citep{kroft2013, schmieder2016a}.  First, we use the subsample of workers undergoing a displacement event, which is a standard way in the labor literature to identify plausibly exogenous separations. Second, we identify a subsample of workers undergoing an unemployment spell between two periods of \textit{stable} employment at \textit{different} firms: an employment-unemployment-employment (E-U-E) transition.\footnote{We follow \cite{herkenhoff2024} for the definition of E-U-E events. Workers within this subsample must be employed the whole year in years $t$ and $t+2$ in this small panel, and undergo at least one quarter of unemployment in year $t+1$. We exclude quits when identifying unemployment spells.} 

We estimate how the origin firm's quality predicts labor earnings after  an unemployment spell when the worker is re-employed:
\begin{equation}\label{eq:eue_reg}
\log(w_{i,t+\tau}) = \alpha\log(w_{i,t}) + \sum_{c=1}^5 \beta_c \mathbb{I}\{f(i,t)\in c\} + \mathbf{X}'_{i,t} \mathbf{\theta} + \varepsilon_{i,t},
\end{equation}
where $w_{i,t+\tau}$ are labor earnings for worker $i$, $\tau$-periods after the E-U-E transition, $w_{i,t}$ are earnings at the time of separation and $\mathbb{I}\{f(i,t)\in c\}$ is an indicator function that takes value one if the employer at the time of separation belonged to productivity quintile $c$. $\mathbf{X}_{i,t}$ are a set of controls that include sex, contract type, occupation and sector-time fixed effects. 
The main coefficients of interest are $\{\beta_c\}_{c=2}^5$, $c=1$ being the excluded category. They indicate how the quality of past employers influences future earnings at a specific horizon $\tau$ after unemployment spells. Controlling for past earnings and job characteristics allows us to indirectly control for the possibility that workers in unemployment still retain different bargaining positions by virtue of potentially different generosity of unemployment benefits, a feature absent in the empirical specification of the previous section.

\paragraph{Results.} Figure~\ref{fig:EUE_empirical} reports the main coefficients of interest for two horizons $\tau$: one and five years after a displacement event (left panel), or a short unemployment spell through stable jobs as defined in the previous section, (right panel).\footnote{In \textbf{Appendix~\ref{sect:empirics}} we report robustness results using firm cluster AKM fixed effects quintiles as measures of firm quality, and the tables of coefficient for all regressions.}

For both samples, two features are worth highlighting. First, the coefficients are monotonically increasing. This indicates that workers previously employed in relatively better firms, \textit{conditional} on their past earnings, obtain higher earnings after an unemployment spell. Second, even five years after unemployment, past firms have a significant effect on workers' earnings. The impact of past firm quality is very persistent across all firm types, and similar across subsamples. The stability and significance of our estimated results over time and across subsamples provides reassuring evidence that selection and duration-dependence cannot explain our effects.\footnote{An alternative mechanism that would explain these positive effects in past firm productivity is that unemployed workers' coming from the best firms just have greater reservation earnings. This would however plausibly imply that they wait longer to find new employment. We find that, controlling for the pre-E-U-E earnings, the quality of past employers is negatively associated with unemployment duration in our sample, as shown in \textbf{Table~\ref{tab:EUE_hor_wage}}. The evidence shows that workers coming from the most productive firms find a new job \textit{faster} and are paid \textit{more} regardless of the previous earnings.} 

Overall, we conclude that the quality of past employers bears a long-lasting effect on workers' earnings even when workers undergo involuntary unemployment spells, consistently with persistent on-the-job human capital accumulation. 

\begin{figure}
\centering
\caption{Persistent effect of past firms for workers\label{fig:EUE_empirical}}
\subcaptionbox{Displaced workers \label{fig:EUE_displ}}{\includegraphics[width=0.5\textwidth]{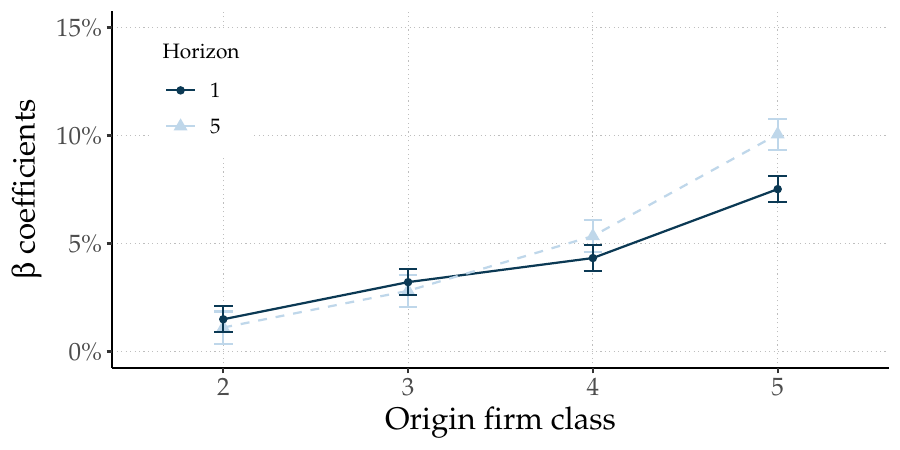}}%
\subcaptionbox{Short unemployment spells\label{fig:EUE_full}}{\includegraphics[width=0.5\textwidth]{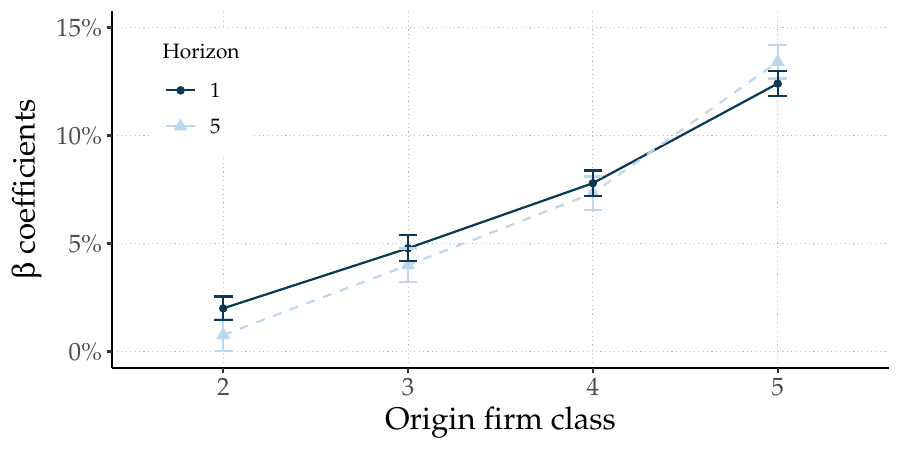}}
\vspace{-1em}
\floatfoot{\textbf{Note:} The figure reports the main coefficient of interest for the regression in \textbf{Equation~\ref{eq:eue_reg}}, error bands are 99\% confidence intervals. The displaced workers sample includes workers that make E-U-E transitions following mass layoffs, while the short unemployment spell sample identifies all E-U-E transitions in which workers are employed for at least 4 quarters after an unemployment spell of less than one year. Firm quality is defined as (yearly) quintiles of firm value added per employee. %
}
\end{figure}

\subsection{Worker-Firm Sorting and the Business Cycle}

Workers accumulate more human capital when employed in more productive firms. This creates a job ladder in firm quality, which workers would want to climb to improve their human capital: a ``human capital ladder''. \textbf{Figure \ref{fig:avg_sorting}} shows that such a ladder exists in the data. Each quarter, workers are grouped into quintiles of 6 years centered rolling-windows AKM person effects,  while firm quality is again measured using value-added per employee. We then calculate the quarterly job-to-job transitions. The figure shows the average firm class destination for every class of worker across our sample: more than 25\% of workers in the top quintile move towards firms in the top quintile, while less than 10\% of workers in the median quintile do. Conversely, more than 25\% of workers in the bottom quintile find jobs in the bottom quintile of firms.

\begin{figure}
\caption{Worker and firm characteristics for new matches and over the cycle} 
\centering
\subcaptionbox{Sorting for new matches\label{fig:avg_sorting}}{\includegraphics[width=0.4\textwidth]{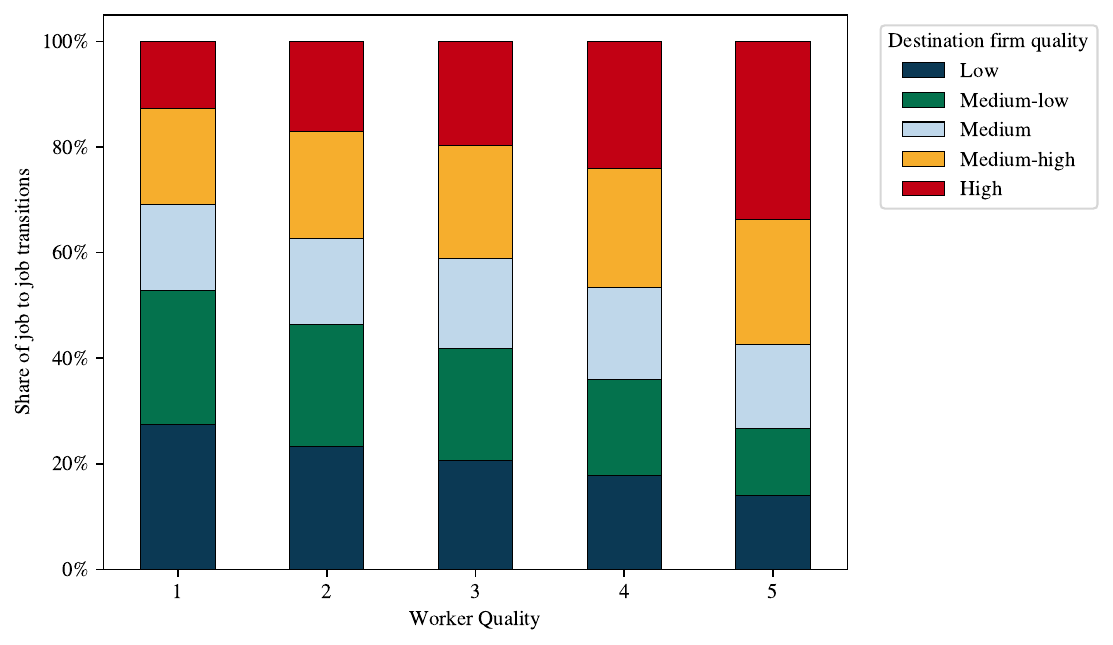}}%
\subcaptionbox{\footnotesize{New matches: Destination}\label{fig:sorting_cycle}}{\includegraphics[width=0.3\textwidth]{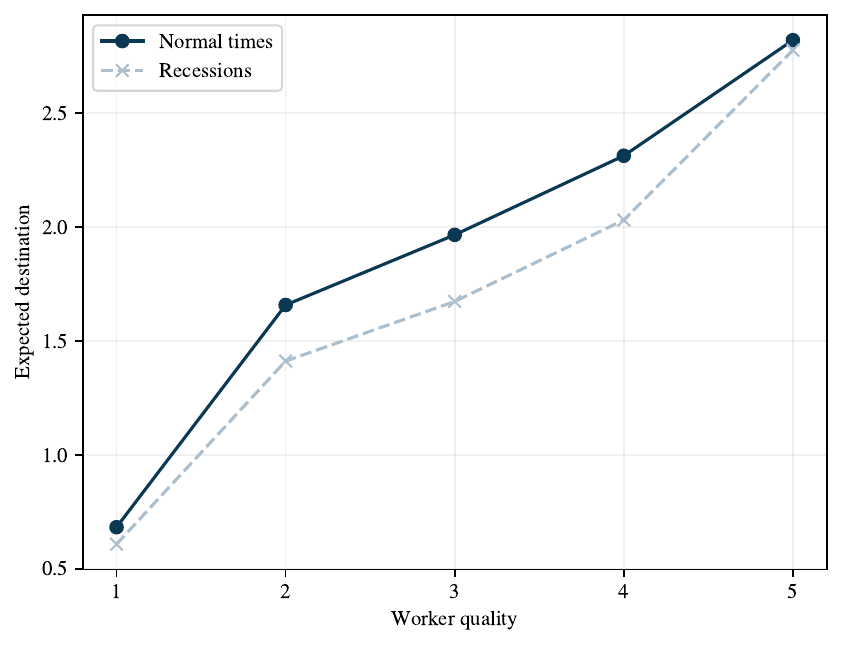}}%
\subcaptionbox{\footnotesize{New matches: Origin}\label{fig:sorting_cycle_origin}}{\includegraphics[width=0.3\textwidth]{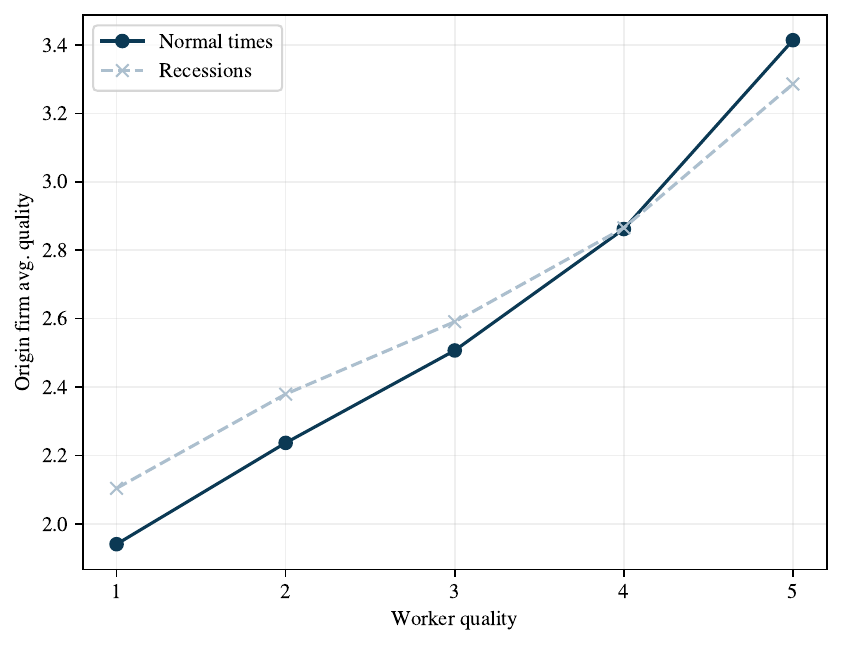}}%

\floatfoot{\textbf{Note:} The figure presents how job-to-job flows are directed towards different quality firms in the data (Panel a) and how workers expected firm quality of destination changes over the cycle (Panel b). Workers are divided into (quarterly) quintiles of own AKM fixed effect, while the measure of firm quality is (quarterly) quintile of value added per worker. In Panel (a): each stacked bar reports the shares of job-to-job transitions ending into each firm quality bin. In Panel (b): the average firm quality of destination by worker quality, distinguishing between normal times and recession periods. In Panel (c): the average origin firm quality of job to job and unemployment transitions, distinguishing between normal times and recessions periods. %
The figure uses quarterly rates from all pooled years, distinguishing quarters of OECD recession from normal times, in the matched employer-employee dataset.}
\end{figure}

While in general better workers will tend to move towards better firms during their career, the job ladder for each worker might be flatter in recessions. Because of the role of firms as learning environments, recessions would thus impact not only the allocations of workers to firms, but also the production of human capital in the aggregate.

 \textbf{Figure \ref{fig:sorting_cycle}} shows that this is indeed what we see in the data: workers have a harder time matching with good firms in recessions.\footnote{Recessions are identified according to the OECD definition. See \textbf{Figure~\ref{fig:ita_u_rate_q}}.}
 In the figure we consider all workers who are moving out of their current match - because of a job to job transition, a voluntary separation, a firing, or a dismissal. The expected firm class the worker is going to be employed at in their next employment spell is then displayed for different moments in the business cycle. During recessions firms ``upskill'' \citep{hershbein2018, modestino2020}, increase their skill requirements for a given position. This forces workers who are moving to new jobs to move to comparatively worse firms. In light of the results in the previous sections, recessions thus impact human capital accumulation, as getting a job of the same quality as in good times is comparatively harder for workers. While these patterns clearly reflect sullying forces, recessions may also weed out unproductive matches, generating a cleansing effect. To assess this mechanism, \textbf{Figure~\ref{fig:sorting_cycle_origin}} reports the average productivity of workers' \textit{origin} firms for (i) job-to-job transitions and (ii) transitions into unemployment, across the business cycle. The final panel paints a mixed picture. On the one hand, high quality workers (top quintiles of the worker fixed-effect distribution) disproportionately separate from less productive firms in downturns---consistent with reduced mismatch (cleansing). On the other hand, low quality workers are, on average, separating from more productive firms during downturns---consistent with a collapse of the job ladder. Taken together with the destination evidence, these origin patterns imply that the distribution of realized match quality shifts left in downturns.

A deeper understanding of the job ladder's cyclical dynamics comes from examining how poaching flows and hires from non-employment differ between firms at the top and bottom of the productivity distribution \citep{haltiwanger_cyclical_2018,haltiwanger_cyclical_2025}. A collapse of the job ladder should manifest itself as weaker net poaching of high- relative to low-productivity firms, where poaching is defined as hires from employment at other firms minus separations to other firms, alongside to a greater reliance on hires from non-employment. \textbf{Table~\ref{tab:HHMS}} reports estimates of how these high-minus-low differentials respond to labor-market slack.\footnote{We describe the empirical approach in detail in \textbf{Appendix~\ref{app:job_ladder_cyc}}.} Consistent with job-ladder compression, increases in unemployment are associated with net poaching shifting away from high-productivity firms toward low-productivity firms (negative coefficients in Column~1). Cleansing forces manifest themselves through a relative decline in hires from non-employment among low-productivity firms (positive coefficients in Column~2). Overall, downturns are associated with a significant compression of the job ladder in the Italian economy, with both cleansing at the bottom of the productivity distribution and sullying going up to the top.

Putting these results together, we highlight long-run costs of business cycles that operate through a collapse of the job ladder, which distorts worker--firm sorting and, in turn, human-capital accumulation. The next section introduces a theoretical framework to interpret these forces, and \textbf{Sections~\ref{sect:estimation_body}}--\textbf{\ref{sect:anatomy}} quantify their aggregate implications.

\begin{table}
\caption{Job ladder collapses in recessions\label{tab:HHMS}}
\include{figures/PaperFigures/Empirics/HHMS_reg.tex}
\vspace{-1.0em}
\floatfoot{\textbf{Note:} The table reports the cyclicality of the difference in Net poaching and Net nonemployment flows as in  \cite{haltiwanger_cyclical_2025} and \cite{haltiwanger_cyclical_2018}. Each cell presents results from a separate regression. The dependent variable is the differential worker flow rate between high- and low-productivity firms. The independent variables in each regression include a cyclical indicator as well as a linear and quadratic time trend and a constant, which are not reported. The sample is quarterly from 1996 to 2018.}
\end{table}

\section{Model \label{sect:model}}

We start by presenting the environment and the preference structure of workers. We discuss the features of the frictional labor market with directed search and finally we characterize the worker problem, the optimal contract and the equilibrium. 

\subsection{Environment}

Time is discrete, runs forever and is indexed by $t \in \mathbb{N}$. We denote future values in recursive expressions by adding a $^\prime$ to them, or index elements by $t$ in non-recursive ones. The economy is populated by a unit mass of $T \geq 2$ overlapping generations of finitely lived, hand-to-mouth, risk-averse workers and a continuum of risk-neutral entrepreneurs. %
All agents in the economy share the same discount factor $\beta \in (0,1)$.
Workers live for $T$ periods, with age $\tau \in \mathcal{T} \equiv \{1,2,3,\dots,T \}$, and %
maximize lifetime flow-utility from non-durable consumption:

\[
    \mathbb{E}_{t_0} \left( \sum_{\tau=1}^T \beta^\tau u(c_{\tau,t_0+\tau})\right) \text{, }
\]

\noindent where $t_0$ characterizes the labor market entry year, and $\tau$ characterizes the age of the agent. Consequently, $c_{\tau,t_0+\tau}$ refers to the consumption of workers of age $\tau$ in time $t_0 + \tau$.

Workers are characterized by heterogeneous human capital levels $h$, with $h \in \mathcal{H} \equiv [\underb{h},\overline{h}]$, and are heterogeneous also with respect to their formal education level $\iota \in \mathcal{I} \equiv \{\text{g,s}\}$, which indicates college and high school education, respectively. Both types enter the labor market with a baseline level of human capital drawn from type-specific exogenous continuous distributions. Upon entry in the labor market, $\mathbb{E}[h | \text{g}] > \mathbb{E}[h|\text{s}]$. To account for the different number of education years in the data, graduate workers entry to the labor market is delayed accordingly. Workers exit the labor force when their employment prospects deteriorate below a certain utility threshold, effectively distinguishing long-term unemployment from non-participation.

Firms are characterized by different levels of quality $ y \in \mathcal{Y} \equiv [\underb{y}, \overline{y}]$.  We model the human capital accumulation process assuming that it depends on the quality of the firm workers are matched with, and on their own initial level of human capital, $h$. 
Workers accumulate human capital only while employed and according to a law of motion that is match-specific: $h^\prime = \phi(h,y,\iota,\psi) = g_{\iota}(h,y)+\psi$, $ \phi: \mathcal{H} \times \mathcal{Y} \times \mathcal{I} \times{\Psi} \rightarrow \mathcal{H}$, where $g_{\iota}$ is the deterministic component of the human capital accumulation dynamics, and $\psi \in \Psi$ constitutes a stochastic component. The function $g_{\iota}$ is concave in both its arguments.\footnote{The deterministic component of human capital accumulation is akin to a ``catching-up" of the firm's quality \citep{lise2020}, up to a point when the worker will not be able to learn any more from the match. The only difference across education levels is the speed of the ``catching-up'', with college-educated workers catching up faster. Workers who match with a low-quality firm will see their ability deteriorating with the same $g_{\iota}$ function. We treat human capital as representing workers' production capacity, entering the production function as effective units of labor, but we remain agnostic as to the specific microfoundations underlying the primitive sources of human capital: it could reflect any factor that enhances a worker's contribution to production.}

Human capital accumulation is risky: at any period any employed worker is subject to the idiosyncratic human capital shock $\psi$, which enters additively with respect to the deterministic component.\footnote{The additive nature of the shock keeps the properties of monotonicity and uniqueness of workers' search strategies unaltered, which is essential for tractability.} The shock affects workers' ability and can amplify, shrink, or even reverse human capital accumulation. We further allow for the possibility that human capital deteriorates while workers are unemployed, according to an arbitrary process $g_u$.\footnote{This process might be without loss of generality deterministic or stochastic, and might or might not depend on current human capital $h$.} In order to take into account the presence of firm-specific human capital, we assume that a share $\tau_{eu}$ of human capital  depreciates upon job loss. Similarly, we allow the share of human capital to be retained after a job transition, $\tau_{ee}$, to be below 1. %

Firms are modeled as one worker--one job matches. Each job is characterized by a promised utility to the worker $V \in \mathcal{V}\equiv[\underb{v},\overline{v}]$. %
We group worker-specific characteristics in a tuple $\chi \in \mathcal{X} \equiv \{ \mathcal{H} \times \mathcal{T} \times \mathcal{I} \}$. The aggregate state of the economy $\Omega$ is characterized by the productivity level   $a \in \mathcal{A} \subseteq \mathbf{R}^+_0$ and by the distribution of agents across states $\mu \in \mathcal{M}$ : $\{W,U\} \times \mathcal{Y} \times \mathcal{X} \times \mathcal{V} \rightarrow [0,1]$.  Let $\Omega = (a,\mu) \in \mathcal{A} \times \mathcal{M} $  represent the aggregate state of the economy and let $ \mathcal{M} $ represent the set of distributions $\mu$ over the states of the economy. Then $ \mu^\prime = \Phi(\Omega, a^\prime)$ is the law of motion of the distribution.  Aggregate productivity is a stationary monotone increasing Markov process, namely $ a^\prime \sim F(a^\prime | a): \mathcal{A}\rightarrow\mathcal{A}$, with the Feller property.

\subsection{Labor markets}

Search is directed. Each labor market is organized as a continuum of submarkets indexed by the expected lifetime utility offered by firms of type $y$, $v_y \in \mathcal{V}$. %
Starting a firm amounts to posting a vacancy at a quality-specific cost $c(y)$, and will be described in \textbf{Section \ref{sect:free_entry}}. %

The matching function $ M(u,\nu) $ for each submarket has constant return to scale and is twice continuously differentiable. The tightness of each submarket in $\mathcal{X} \times \mathcal{V}$  is defined as $ \theta = \nu/u $, with $\theta(\cdot): \mathcal{X} \times \mathcal{V} \times \mathcal{A} \times \mathcal{M} \rightarrow \mathbf{R}^+_0 $. Job finding rates are defined as $p(\theta(\cdot))=M(u,\nu)/{u}$, where $p(\cdot): \mathbf{R}^+_0 \rightarrow [0,1]$ is a twice continuously differentiable, strictly increasing, and strictly concave function with $p(0)=0$, $\underset{\theta\rightarrow +\infty}{\lim\:p(\theta)} = 1$ and $p^\prime(0)<\infty$. The vacancy-filling probability is defined as $q(\theta(\cdot))=M(u,\nu)/\nu$ , where $q(\cdot): \mathbf{R}^+_0 \rightarrow [0,1]$ is twice continuously differentiable, strictly decreasing, and strictly convex, with $q(0)=1$, $\underset{\theta\rightarrow +\infty}{\lim\:q(\theta)}=0$ and $q^\prime(0)<0$, such that $q(\theta) = p(\theta)/\theta$, and $p(q^{-1}(\cdot))$ is  concave.

Upon matching, workers produce according to the twice-continuous increasing and concave production function $f(h,y;a) + x(a) : \mathcal{A} \times \mathcal{H} \times \mathcal{Y} \rightarrow \mathbf{R}^+_0$. The $x(a)$ component of the production function is a cost which can depend on the aggregate productivity realization.\footnote{This is a reduced-form way of incorporating financial frictions in the model, which make fixed costs loom larger over flow-production in downturns.} Workers' compensation is determined by dynamic contracts through which firms deliver a promised lifetime utility, as described in \textbf{Section \ref{sect:contract}}.

Workers search on the job with probability $\lambda_e$. Matches are destroyed each period at an exogenous rate, possibly varying by age, $ \delta_{\tau} $. Matches separate also if the worker moves to another firm (poachings), or if the firm unilaterally walks away from the match (firings). Lastly, unemployed workers whose expected value of re-employment falls below a threshold $\underb{p}$ are assumed to permanently exit the labor force.\footnote{Workers could also voluntarily decide to quit to unemployment. We experimented with allowing workers quits to unemployment and found that in the model, given plausible parameterizations and in absence of additional shocks, quits to unemployment are never optimal. For this reason and for simplicity we remove the possibility from the model presented in this paper.}  

Timing is represented in \textbf{Figure \ref{fig:timeline}}. At the beginning of each period an aggregate productivity shock is drawn; entrepreneurs open vacancies across submarkets and post their offers; workers search from unemployment or on-the-job, and move to a new job if the search is successful; production takes place; workers accumulate human capital depending on their employment status and idiosyncratic shock realization; an exogenous share of matches breaks down, while some firms endogenously exit. %

\begin{figure}
    \centering
        \begin{tikzpicture}[%
    every node/.style={
        font=\scriptsize,
        text height=1.0ex,
        text depth=.5ex,
        align = center
    },
    scale = 2.0,
]

\tikzset{label/.style={draw=gray, ultra thin, rounded corners=.25ex, fill=gray!20,text width=4.75cm, text badly centered,  inner sep=.5ex, above = 2em, anchor=west,rotate=45}}

\tikzset{tick/.style={below=3pt}}

\tikzset{lubel/.style={draw=gray, ultra thin, rounded corners=.25ex, fill=gray!20,text width=3.5cm, text badly centered,  inner sep=.5ex, above = 2em, anchor=west,rotate=45}}

\draw [dashed] (0,0) -- (0.5,0);
\draw [->] (0.5,0) -- (7.0,0);

\node[anchor=north] at (0,-0.25) {\textbf{t}};

\foreach \x in {0.5, 1.5, ...,3.5}{
    \draw (\x cm,3pt) -- (\x cm, 0pt);
}

\draw (5.25 cm,12pt) -- (5.25 cm, -6.0pt);

\draw (0.5,0) node(A1) [tick] {contracting} node (B1) [label]  {  firms  $\{ h, \tau, \iota, y\}$ offer  $\{ w_\tau (s^\tau )\}^T_{\tau = t}$};
\draw[blue] (B1.west) -- ++(0,-0.4);

\draw (1.5,0) node(A2) [tick] {search} node (B2) [label]  {workers $\{ h, \tau, \iota \}$ search in $\{ h, \tau, \iota, v\}$ };
\draw[blue] (B2.west) -- ++(0,-0.4);

\draw (2.5,0) node(A2) [tick] {production} node (B3) [label]  {  $\{ h, y \}$ matches produce $f(h,y;a)$};
\draw[blue] (B3.west) -- ++(0,-0.4);

\draw (3.5,0) node(A2) [tick] {learning} node (B4) [label]  { $\psi$ shocks realized, $h^\prime$ revealed};
\draw[blue] (B4.west) -- ++(0,-0.4);

\draw (4.5,0) node(A2) [tick] {separation} node (B5) [label]  { exogenous EU, firm exit};
\draw[blue] (B5.west) -- ++(0,-0.4);

\node[anchor=north] at (5.25,-0.25) {\textbf{t+1}};

\draw (6.0,0) node(A2) [tick] {shock resolution} node (B6) [lubel]  { agents learn about $\Omega$};
\draw[red] (B6.west) -- ++(0,-0.4);

\end{tikzpicture}

    \caption{Timeline of Worker--Firm Match}
    \label{fig:timeline}
\end{figure}
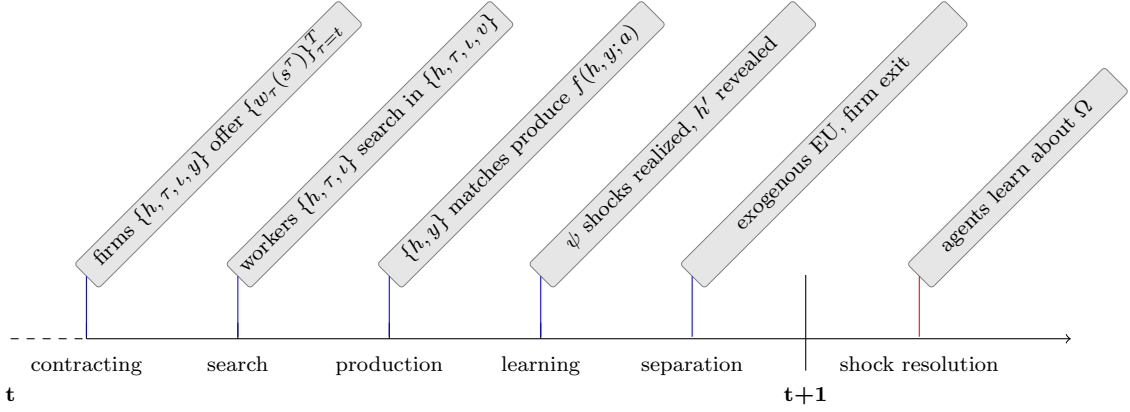

\subsection{Informational and contractual structure \label{sect:infospace}}

Firms post fully state-contingent contracts. Each contract prescribes an action for each realization of the history of the worker--firm match. The tuple defining productivity and worker characteristics in a match with any firm $y$  at time $t$ is defined by $s_{t} = (h_{t},\tau_{t},\iota, a^{t}, \mu^{t}) \in \mathcal{S}^t = \mathcal{X} \times \Omega^t = \mathcal{H} \times \mathcal{T} \times{I} \times \Omega^t$, that is the worker skill, age, formal education, the history of aggregate productivity shocks, and workers' distributions across their employment history. A given history of realizations between $t$ and $k$ periods ahead is thus $s^{t+k} = (s_t,s_{t+1}...., s_{t+k})$. The contract defines a transfer of utility from the risk-neutral firm to the risk-averse worker within the match for all future possible histories of shocks. We define $\tau_{t_0}$ as the age at which the worker is hired and $T$ is the retirement age. The history of realizations between $t_0$, the time of hiring of the worker, and $t_0 + (T-\tau_{t_0})$, the time of maximum duration of the match with the worker before retirement, is thus $s^{t_0 + (T-\tau_{t_0})}$.

Histories of workers and productivity shocks are common knowledge, and the future realizations of shocks are fully contractible. While the contract is state-contingent, workers' actions are private knowledge in the search stage, so firms are  unable to  counter outside offers. The contracts offered by firms are then  defined as:
\begin{equation}
    \mathcal{C}^{\tau_{t_0}} \coloneqq (\mathbf{w},\mathbf{\zeta})\text{ with } \mathbf{w}\coloneqq\{w_t(s^{\tau_t-\tau_{t_0}+t_0})\}_{t=t_0}^{t_0 + (T - \tau_{t_0})}\text{, and }\mathbf{\zeta}\coloneqq\{v_t(s^{\tau_t-\tau_{t_0}+t_0})\}_{t=t_0}^{t_0 + (T - \tau_{t_0})} \label{eqn:contract_set} .
\end{equation}
Firms promise a series of state-contingent wages defined by the series of utility values $v_t$ sought at each node of the history.\footnote{Similarly to \cite{menzio2010,tsuyuhara2016}, and \cite{balke2022},  to guarantee that the problem is well behaved and the firm profit function is concave, the contract will require a two-point lottery, which specifies probabilities over the actions prescribed. We omit it here for conciseness.} $\mathbf{\zeta}$ is the action suggested by the contract, which is bound to be incentive compatible for the worker. %
The resulting relationship between workers and firms is characterized by a contract with forward-looking constraints. The state space of the worker problem can be expressed in terms of their current lifetime utility, as in \cite{Spear1987}, so as to avoid having to keep track of all past histories $s^t$ at each period. The relevant state space is then $\mathcal{X} \times \mathcal{V} \times \Omega$.

\subsection{Worker problem\label{sect:workers}}

Given current lifetime utility $V$, job seekers with characteristics $\chi$ have to decide in which submarket to direct their search. Submarkets are indexed by worker type $\chi$ and by offered utility $v$ associated to firms' posted vacancies. As discussed in \textbf{Section \ref{sect:free_entry}}, the choice over $v$ will also indirectly determine which firm $y$ the worker matches with, and thus the implied human capital accumulation path. For now, let us assume this (conditional) mapping exists. This amounts to assuming that the function $v(y;\chi,V)$ is a bijective function $f_v: \mathcal{X}  \times \mathcal{V} \times \mathcal{Y} \times  \mathcal{A} \times \mathcal{M} \rightarrow\mathcal{V}$. Upon observing a job offer with utility $v$, a worker $\chi$ with current utility $V$ will be able to infer which firm type $y$ is posting the offer. %

A worker of type $(\chi,V)$  that enters the search stage has lifetime utility $V + \lambda_i\max\{0,R(\chi,V;\Omega\}$, where the second component of the expression embeds the option value of the search, with $R$ being the search value function, and $\lambda_i =1$ if the worker is unemployed or $\lambda_i = \lambda_e$ if they are employed. $R$ is defined as:
\begin{equation} \label{eqn:wrk_option}
R(\chi,V;\Omega) = \underset{v}{\sup}\;\Big[ p(\theta(\chi,v;\Omega))\big[v-V]\Big] \text{.}
\end{equation}
We denote the solution of the search problem as $v^* = v^*(\chi,V;\Omega)$, and $p^*(\chi,v^*;\Omega) = p(\theta(\chi, v^*;\Omega))$ as the associated optimal job-finding probability. %
The lifetime utility of an unemployed worker at the beginning of the production stage can be defined as

\begin{align}
U(h,\tau,\iota;\Omega)&=\;u(b(h,\tau))+\beta \mathbb{E}_{\Omega,\psi}\bigg(U(h^\prime,\tau+1, \iota;\Omega^\prime) \nonumber\\
&+  \max\{0,R(h^\prime,\tau+1,\iota,U(h^\prime,\tau+1,\iota;\Omega^\prime);\Omega^\prime)\} \bigg)\label{eqn:unempl_prob} \text{,}
\end{align}

\noindent where $b(h,\tau)$ is a skill and age dependent unemployment benefit.
Given finite workers' lives, $U(h,\tau, \iota; \Omega) = 0 \,\, \forall (\chi; \Omega) \in \mathcal{X} \times \mathcal{A} \times \mathcal{M}$ whenever $\tau>T$. The corresponding lifetime utility of a worker employed at firm $y$, with human capital $h$, age $\tau$, education $\iota$ and promised utility $V$ at the beginning of production stage can be expressed as:

\begin{align}
V(h,\tau, \iota;\Omega)&=\; u(w) +\beta \mathbb{E}_{\Omega,\psi}\bigg(\delta U( h^\prime,\tau+1, \iota;\Omega^\prime)
+\; (1-\delta)\Big[
V(h^\prime,\tau + 1, \iota;\Omega^\prime) \nonumber\\
&+\;  \lambda_{e} \max\{0,R(h^\prime,\tau+1,\iota, V(h^\prime,\tau + 1, \iota;\Omega^\prime);\Omega^\prime)\}\Big]\bigg) \text{,}  \label{eqn:empl_prob}
\end{align}

\noindent where $w$ is the promised wage, $\delta$ is the separation probability, and $V(h^\prime,\tau + 1, \iota;\Omega^\prime)$ is next period's state-contingent promised utility of remaining in the current firm, which becomes the outside option in the search problem.%
 Again,  $V(h,\tau, \iota; \Omega) = 0 \,\, \forall (\chi; \Omega) \in \mathcal{X} \times \mathcal{A} \times \mathcal{M}$ whenever $\tau>T$. %
Firms internalize incentives embedded in workers' strategies and post wages and utility offers to maximize profits by optimizing retention. This way, future promised utilities incorporate both future wages or option values of search. %
The policy functions are uniquely defined and identify target $y$ as long as a bijective mapping between the offered utility $v$ and $y$ given $\chi$ exists.%

\begin{defi}[Optimal retention probability and utility return]
Let $v^*(\chi,V;\Omega)\in\mathcal V$ denote the worker's optimal choice and 
$p^*(\chi,v;\Omega)\in[0,1]$. The solution to the worker's problem defines a retention function
$\widetilde{p}:\mathcal{X}\times\mathcal{V}\times\Omega \rightarrow [(1-\delta)(1-\lambda_e),1-\delta]$
and a utility return $\widetilde{r}:\mathcal{X}\times\mathcal{V}\times\Omega \rightarrow \mathbb{R}$:
\begin{align}
\widetilde{p}(\chi,V;\Omega)
&\equiv (1-\delta)\Big(1-\lambda_{e}\, p^*(\chi,v^{*}(\chi,V;\Omega);\Omega)\Big)\label{eqn:p_ret},\\
\widetilde{r}(\chi,V;\Omega)
&\equiv \delta U(\chi;\Omega)+(1-\delta)\Big[V + \lambda_{e}\max\{0,R(\chi,V;\Omega)\}\Big].
\end{align}
\end{defi}

\subsection{Vacancy creation and free entry \label{sect:free_entry}}

The economy is populated by a continuum of risk-neutral entrepreneurs. Each entrepreneur can invest to reach the desired level of firm quality $y$. The start-up costs of the firm are priced in terms of the consumption good and they coincide with vacancy posting costs in the frictional labor market. The cost of each vacancy is positively related to the quality of the firm being created through the cost function $c(y)$, which is increasing and strictly convex.\footnote{We assume that entrepreneurs can borrow from risk-neutral, deep-pocketed financiers to finance the vacancy. As in \cite{herkenhoff2019impact} this assumption implies that the cost of credit for entrepreneurs coincides with the risk-free rate.}

At a generic time $t$ each entrepreneur chooses in which submarket to post the vacancy by offering utility $W \in \mathcal{V}$. Each submarket is characterized by worker characteristics and current utility $(\chi,V) \in \mathcal{X} \times \mathcal{V}$, and we prove later in \textbf{Section \ref{sect:contract}} that the firm choice over $W$ uniquely maps upon vacancy posting into firms' qualities $y \in \mathcal{Y}$ conditional on the submarket's characteristics of choice.

We define $J(h,\tau,\iota,W,y;\Omega) \in \mathcal{X} \times \mathcal{V} \times \mathcal{Y} \times \Omega$ as the value function of a firm, which capitalizes all future profits from the match. As entrepreneurs choose the submarkets in which to open a vacancy, they face the following problem:
\begin{equation} \label{eqn:vacancy_open}
\Pi(h,\tau,\iota,W,y;\Omega) = \underset{h,\tau,\iota,W}{\sup}\:-c(y) +q(\theta(h,\tau,\iota,W;\Omega))[J(h,\tau,\iota,W,y;\Omega)]
\end{equation}
Given perfect competition, free entry and the possibility for all entrepreneurs to choose \emph{any} possible firm kind $y$, the expected profits from creating a vacancy are driven down to 0 in submarkets that actually open.%
This translates into a free entry condition:
\begin{align} \label{eqn:free_entry}
\Pi(h,\tau,\iota,W,y;\Omega) & \leq0\:\;\text{for}\:\:\forall\{h,\tau,\iota,W,y;\Omega\}\in\{ \mathcal{X} \times \mathcal{V}\times\mathcal{Y} \times \Omega \}
\end{align}
The equilibrium tightness in each open submarket is:
\begin{equation}
\theta (h,\tau,\iota,W; \Omega) = q^{-1} \left( \frac{c(y)}{J(h,\tau,\iota,W,y;\Omega)} \right). \label{eqn:eq_tight}
\end{equation}

\subsection{Firm problem \label{sect:contract}} 
The firm contract design plays out as a sequential equilibrium game with leader-follower dynamics, in which firms play as the principal/leader and workers are the agents/followers. Workers' limited commitment implies they will search for new jobs whenever they have the possibility to do so. Firms cannot observe poaching offers and thus cannot counter them. The sequence of past histories $s^{t}$ is common knowledge, and while the firm cannot observe any of actions of its workers, it has enough information to internalize their optimal search policy decisions.
As an additional constraint, wages are downward rigid for continuously employed workers: firms may adjust the growth path of wages within an ongoing match (e.g., slow or freeze wage increases), but they cannot implement nominal wage cuts as long as the worker remains employed at the firm. In absence of the last constraint, the wage would be fully state-contingent, with workers facing wage decreases instead of displacements. This wage-setting restriction is essential for the model to reproduce the cross sectional properties or the cyclical behavior of match separations, as discussed in \textbf{Section~\ref{validation}}.

The firm will decide whether to continue in each period according to the exit policy:
\begin{defi}[Exit policy] The following indicator takes a value of one if the firm decides to exit :
$$
 \eta(h,\tau,\iota,W,y;\Omega) = \left\{
 \begin{array}{l}
 0 \;\;\text{if firm continues} \\ %
 1 \;\;\text{if firm exits}
 \end{array}
 \right.
$$
\end{defi}

Define also $\widetilde{p}(\chi,V;\Omega)$ as the optimal retention function and $\widetilde{r}(\chi,V;\Omega)$ as the optimal continuation utility for workers from workers' solution to the on-the-job search problem.%
The firm chooses the wage(s) to be offered in the current period $w_i$, the utility promises $W_i^\prime$ and the probability $\pi_i$ in a two-point lottery. The value function of an incumbent firm $y$ in state $(h,\tau,\iota,W;\Omega)$ can be written recursively using the promised utilities as additional state variables, so that:
\begin{align}
J(h,&\tau,\iota,W,y; \Omega)= \nonumber \\ &\underset{\pi_i, \{\eta^\prime_i,w_i,W_{i}^\prime\}}{\sup} \sum_{i=1,2} \pi_i  \Bigg( f(y,h;a)-w_i \nonumber \\
+ & \beta \mathbb{E}_{\psi}\bigg[ (1 - \eta^\prime_i) \cdot \max\bigg\{0,\mathbb{E}_{\Omega} \bigg[\widetilde{p}(h^\prime,\tau+1,\iota,W_{i}^\prime;\Omega^\prime) J(h^\prime,\tau+1,\iota,W_{i}^\prime,y;\Omega^\prime)\bigg]\bigg\}\bigg] \Bigg)\label{eqn:firm_prob}%
\end{align}
subject to
\begin{align}
&W =  \sum_{i=1,2} \pi_i \Bigg( u(w_i)+ \beta \mathbb{E}_{\Omega,\psi} ( (1 - \eta^\prime_i)\widetilde{r}(h^\prime,\tau+1,\iota,W_{i}^\prime;\Omega^\prime)
+ \eta^\prime_i U(h^\prime,\tau+1,\iota;\Omega^\prime)  ) \Bigg),\label{eqn:pkeep}\\
&J(h^\prime,\tau+1,\iota,W_{i}^\prime,y;\Omega^\prime) < 0 \implies \eta^\prime = 1 \text{, otherwise } \eta^\prime = 0 ,\label{eqn:firmpc}\\
  & (w - w^\prime)(1- \eta^\prime) \leq 0 \label{eqn:wagecon},\\
&\sum_{i=1,2} \pi_i = 1,
\end{align}
where \textbf{Equation \eqref{eqn:pkeep}} is the promise keeping constraint ensuring that the current value of the contract is based on the current wage and future utility promises. \textbf{Equation \eqref{eqn:firmpc}} captures the firm side of the limited commitment friction: at any point the entrepreneur can walk away and open a new vacancy, so the outside option against which the continuation of the match is compared is the \textit{ex ante} value of opening a new vacancy, which is $0$ because of free entry. Finally,  \textbf{Equation \eqref{eqn:wagecon}} limits wage adjustments by requiring wage growth to be (weakly) positive.

Firms commit to deliver a utility value to workers, but exit when the present value of future profits becomes negative, i.e. when the constraint \eqref{eqn:firmpc} binds - then  $\eta^\prime = 1$.  Incumbent firms make their exit decisions before the realization of aggregate productivity but after the next period realization of idiosyncratic human capital shocks, as in \cite{gomes_financing_2001} and \cite{Xiaolan2014}.%
At the beginning of a period both firms and workers know the contingencies under which the firm will shut down. Exit is therefore state-dependent (see Appendix for details).

\subsection{Equilibrium definition}

\paragraph{Recursive Equilibrium.} Let $\Theta=\mathcal{A} \times \mathcal{M} \times \mathcal{H} \times \mathcal{T} \times \mathcal{I}$. A recursive equilibrium in this economy consists of a market tightness $\theta: \Theta \times \mathcal{V} \rightarrow \mathbb{R}_+$, a search value function $R:\Theta \times \mathcal{V} \rightarrow \mathbb{R}$, a search policy function $v^*:\Theta \times \mathcal{V} \rightarrow \mathcal{V}$, an unemployment value function $U:\Theta \rightarrow \mathbb{R}$, a firm value function, $J: \Theta \times \mathcal{V} \times \mathcal{Y} \rightarrow \mathbb{R}$, a series of contract policy functions $\{c_\tau\}_{\tau=1}^T: \mathcal{S}^\tau \times \mathcal{Y} \rightarrow \mathcal{C}^{\tau}$, a bijective mapping between firm qualities and promised utilities at hiring $f_v:\Theta \times \mathcal{V} \times \mathcal{Y}\rightarrow \mathcal V$, an exit threshold for aggregate productivity $a^*:\Theta \times \mathcal{V} \times \mathcal{Y} \rightarrow \mathcal{A}$, and a law of motion for the aggregate state of the economy $\Phi_{\Omega,a}: \mathcal{A} \times \mathcal{M} \rightarrow \mathcal{A} \times \mathcal{M}$ such that:
\begin{enumerate}
    \item Given the mapping $f_v$, market tightness satisfies \textbf{Equation~\eqref{eqn:eq_tight}}.
    \item The unemployment value function solves \textbf{Equation~\eqref{eqn:unempl_prob}}.
    \item Search value functions solve the search problem in \textbf{Equation~\eqref{eqn:wrk_option}} and $v^*$ is the associated policy function.
    \item Firm value functions and associated contract and exit policy functions solve \textbf{Equation~\eqref{eqn:firm_prob}} for each $\tau \leq T$.%
    \item The law of motion for the aggregate state of the economy respects the search and contract policy functions and the exogenous process of aggregate productivity.
\end{enumerate}

\begin{defi}[Block Recursive Equilibrium]\label{defi:bre_defi}
 A Block Recursive Equilibrium (BRE) is a recursive equilibrium such that the value and policy functions depend on the aggregate state only through aggregate productivity, $a \in \mathcal{A}$ and not through the distribution of agents across states $\mu \in \mathcal{M}$.
\end{defi}

\begin{ppert}[Existence of a BRE]\label{prop:BRE}
A block recursive equilibrium exists for the model.
\begin{proof}
See \textbf{Appendix Section~\ref{Appendix:BRE_existence}}.
\end{proof}
\end{ppert}

\subsection{Model Properties: Discussion  \label{sect:discussion_body}}

The objective of our model of dynamic sorting is to understand the properties of jobs creation and worker search in a setting with two-sided heterogeneity. The following properties guarantee a high degree of tractability.\footnote{We refer the reader to \textbf{Appendices~\ref{Appendix:wrk_problem}} and \textbf{\ref{Appendix:optimal_contract}} for a more in-depth discussion of the theoretical properties of the model, proofs of the uniqueness of workers' policy functions and all the other proofs of this section.}

\begin{ppert}[Unique Bijective Mapping]\label{ppert:mapping}
Upon matching, firm quality $y$ and utility promises in vacancy postings $v$ are related by a bijective mapping conditional on the aggregate state of the economy, $\Omega$, and workers characteristics $(\chi,V)$.
\end{ppert}

The previous proposition establishes that workers' directed search toward promised values is equivalent to directed search toward firms' types. We then focus on the properties of the search strategy to get a complete view of how sorting works in equilibrium. 
\begin{ppert}[Search Monotonicity and Uniqueness]\label{ppert:monoton}
The optimal search strategy when unemployed, conditional on age $\tau$, formal education $\iota$ and the aggregate state $\Omega$, is unique and weakly increasing in workers' human capital $h$. The optimal search strategy when employed, conditional on age $\tau$, formal education $\iota$ and the aggregate state $\Omega$, is unique and weakly increasing in workers' human capital $h$ and current level of lifetime utility $V$. 
\end{ppert}

\textbf{Property \ref{ppert:monoton}} guarantees that, abstracting from idiosyncratic as well as aggregate shocks, workers sort positively with respect to their human capital. \textbf{Property \ref{ppert:mapping}}, in turn,  guarantees that workers with same observable characteristics agree on firms' relative ranking. Firms are thus vertically differentiated, and there is a separating equilibrium whereby workers with different characteristics optimally search in distinct firms. 

Because we are interested in how aggregate fluctuations shape the distribution of matches, we now turn to considering how they affect search strategies. 

\begin{ppert}[Search in Good and Bad Times]\label{ppert:aggr}
The optimal search strategy is increasing in the aggregate productivity level, $a$. 
\end{ppert}

\textbf{Property \ref{ppert:aggr}} points to a key mechanism in the model: aggregate fluctuations modify sorting in the labor market. The value of vacancies posted by each firm in equilibrium changes with the business cycle, as submarkets become less productive and less tight in bad times. This also takes place due to the presence of increasing and convex vacancy costs, which do not vary with the cycle. Faced with a lower probability of successfully matching with the firm they would aim to match with in good times, risk-averse workers adjust their search downward. In turn, firms will adjust downwards their utility offers given the lower expected values of matches across the board.

Firms' offers will optimally respond to workers' incentives for on-the-job search. 

\begin{ppert}[Optimal Retention]\label{ppert:retention} Retention probabilities, $\widetilde{p}(h,\tau,\iota,W;\Omega)$ are: 
\begin{enumerate}[label=(\roman*)]
    \item  increasing in the value of promised utilities, $W$.
    \item  decreasing in aggregate productivity, $a$.
\end{enumerate}
\end{ppert}

Despite continuation values within the match being procyclical and workers searching more ambitiously in good times, matches will separate more often in expansions as workers transition to new jobs. This is consistent with employment-to-employment transitions being strongly pro-cyclical in the data. 

\textbf{Property \ref{ppert:retention}} highlights another aspect of the incentives that shape the contract designed by firms: retention grows in continuation values $W$. To close the model, we need a rule for surplus sharing between firms and workers, that is, a wage protocol for firms to deliver lifetime utility promises to workers. We focus here on the wage protocol for a firm for which $\eta^\prime = 0$, so that there is a wage to be paid in the future period.

\begin{ppert}[Wage Protocol]\label{ppert:wage_protocol} The optimal contract for the continuing firm delivers a wage growth rule that satisfies: \begin{equation} \label{eqn:euler}
\frac{\partial \log\,\widetilde{p}(\chi^\prime,W^\prime_{i};\Omega^\prime)}{\partial W_{i}^\prime}J(\chi^\prime,W^\prime_{i},y;\Omega^\prime) =  \frac{1}{u'(w^\prime_{i})} - \frac{1}{u'(w_{i})} \text{,}
\end{equation}
with $\chi^\prime \equiv (\phi(h,y,\iota,\psi),\tau+1,\iota)$ being the definition of individual characteristics and $w^\prime_{i}$ being the wage paid in the future state, conditional on realizations of idiosyncratic risk $\psi$ and aggregate risk $a^\prime$.
\end{ppert}

This result extends the wage equation in \cite{balke2022} to an environment with two-sided heterogeneity. Wage growth is  proportional to the residual continuation value of the match, $J$ and the semi-elasticity of the worker's retention probability to future value promised. Limited liability provides the rationale for inefficient separations. At the same time, it also gives rise to wage rigidity, as it ensures that both elements in Equation \ref{eqn:euler} are weakly positive if the firm does not close down.\footnote{Notice that, in the presence of risky human capital accumulation, $J$ will fluctuate together with the human capital levels of the worker even in the absence of aggregate fluctuations. However, because the contract provides insurance to workers, changes in their human capital will have asymmetric effects on wage growth, thus weakly increasing the labor share over time \citep{ai2021}.}

\begin{ppert}[Countercyclical Separations]\label{ppert:seps}
Conditional on the existing contract and on worker and firm types, there exists an aggregate state $a^*$ below which firms will not continue to operate. The threshold $a^*$ is, all things being equal, increasing in the value promised to workers, and decreasing in worker and firm types.
\end{ppert}  

A clear implication of \textbf{Property \ref{ppert:seps}} is that, at the onset of recessions, firms are significantly more likely to lay off workers. In addition, lower-skilled workers and low-productivity firms are more likely to separate in recessions. The counter-cyclicality of separations is a common feature in the data, together with the lower job security enjoyed by workers who are younger, less productive, or employed by less productive firms.\footnote{These relationships are observed in the data and explicitly modeled in \cite{jarosch2023}, but emerge endogenously in our framework, without the need of specifying job security as a contract characteristic. Abstracting from labor market institutions is not limiting the ability of our model to partly capture the observed duality of the labor market.} 

\subsection{Sorting in equilibrium}

The theory discussed in this section predicts that workers' search is monotonic in individual characteristics and in the aggregate state (see \textbf{Proposition \ref{ppert:monoton}}). 

\begin{figure}[t]
    \centering
    \caption{Search Policies}\label{fig:search}
   \includegraphics[width=0.7\textwidth]{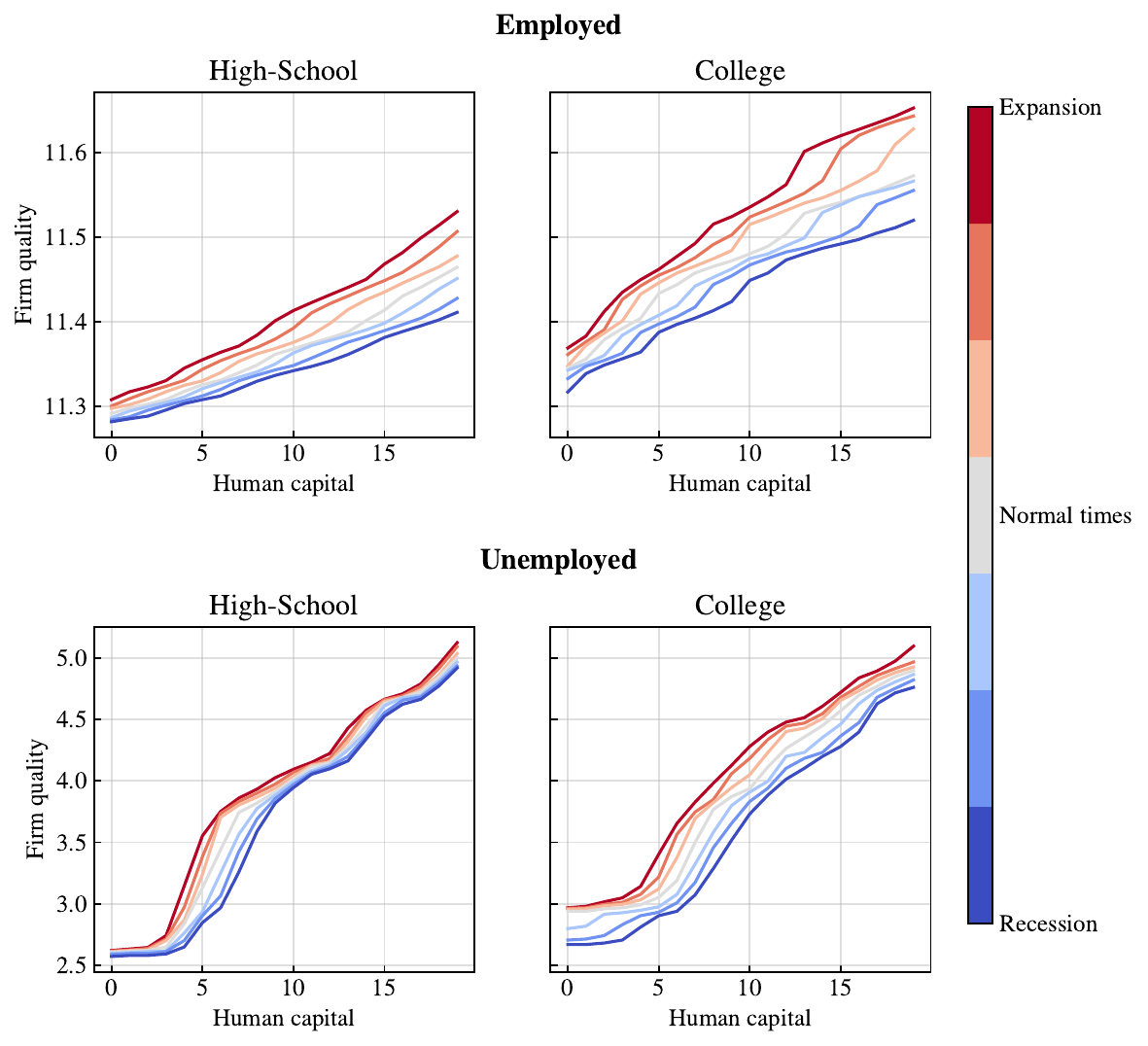}%
     \vspace{-1.0em}
    \floatfoot{\textbf{Note:} Search policy function by human capital level and aggregate state. For the employed, policies are averaged across labor market experience and wage promises.}
\end{figure}
In \textbf{Figure \ref{fig:search}} we plot the equilibrium mapping between workers' human capital and search behavior resulting from our estimation, for different realizations of the aggregate state. Search is monotonic in both dimensions. Stark differences in the set of submarkets targeted by employed versus unemployed workers are a standard feature of directed-search models. Employed workers target more productive but harder-to-access submarkets (lower job-finding probabilities), whereas unemployed workers concentrate on less productive submarkets where job finding is more likely. Search strategies of college-educated workers are more sensitive to shifts in aggregate conditions, because expected human capital accumulation has greater long term bearing on their careers. This also implies that more educated workers will adjust their search strongly in bad times, especially if unemployed, in order to minimize their unemployment risk and the consequent human capital deterioration. Vacancy creation thus tends to be of lower quality in recessions, which in turn increases misallocation and leads to a deterioration in sorting throughout the recovery. Our model endogenously rationalizes the empirical findings in \cite{haltiwanger_cyclical_2025}, who observe that recessions apparently feature an initial ``cleansing'' phase followed by a more prolonged ``sullying'' dynamics as the labor market re-builds its job ladder.

\section{Bringing the model to the data \label{sect:estimation_body}}

The model features internally and externally calibrated parameters. To estimate the first group of parameters, we target moments from Italian administrative data, provided by the Italian Social Security Institute (INPS), for all years between 1996 and 2018.%
To obtain model moments, we simulate a population of overlapping generations working for 45 years (180 quarters, from 18 to 63 years old, the legal retirement age for most years in our period of analysis). We then use a simulated method of moments (SMM) approach. This section will first present the quantitative setup of the model, then calibration choices for the parameters that are set externally, and finally our estimation results.\footnote{In the discussion of the model estimation and results we use earnings and wages interchangeably, as the model abstracts from hours variation.}   %

\subsection{Calibration and estimation}\label{calib_est}
 
\paragraph{Quantitative Setup.} \textbf{Table \ref{tab:FunForms}} collects all the functional form choices. We assume a Cobb-Douglas production function and allow for potentially cyclical maintenance costs, captured by the parameter $x$. We follow \cite{schaal2017} and \cite{menzio2010} in picking a CES function in market tightness. Vacancy creation imposes increasing costs in firm's quality $y$ and the type of labor market flow $j$ (EE or UE), according to the convex function $c_j(y)$. Workers are risk-averse with constant-relative-risk-aversion (CRRA) utility. The human capital production technology is concave in firm quality, $y$, which is scaled by a parameter $\xi$, and in the existing stock of human capital, $h$. Future human capital is also subject to additive i.i.d. shocks, $\psi \sim \mathcal{N}(0, \sigma_\psi)$. Home production is increasing in the stock of human capital according to the parameter $\xi_b$. Finally, we fix the ratio in separation rate between old and young workers following the data and we then estimate internally the separation rate for young workers, allowing the separation rate to be mildly age-dependent to capture age specific aspects of worker quality that are unrelated to business cycles but still empirically relevant.\footnote{The actual separation probability from the workers' perspective will thus be $ \delta = \max \{ \delta_{\tau,\iota}, \eta(\cdot) \}$.} 

\begin{table}[h!]
\centering
\caption{Functional Forms}\label{tab:FunForms}
\resizebox{0.6\textwidth}{!}{
    \begin{tabular}{ll}
        \toprule
        \textbf{Functions} &  \\
        \midrule
        Production function & $f(y,h;a) = A y^\alpha h^{1-\alpha} - x(A-1)$\\
        Job-finding probability & $p(\theta) = \theta(1 + \theta^\gamma)^{-\frac{1}{\gamma}} $\\
        Vacancy creation cost by flow & $c_j(y) =  \frac{y^\kappa}{\kappa}c_j$, for $j=\{EE, UE\}$\\
        Utility function & $U(c) = \frac{c^{1-\nu}}{1 - \nu}$\\
        Human capital accumulation & $g_{\iota}(h,y) = (\xi y)^{\phi_{\iota}} h^{1-\phi_{\iota}}$\\
        Home production & $b(h,\tau) = b + \xi_b h$\\
        \bottomrule
    \end{tabular}}
\end{table}

\paragraph{Calibration.} The model is characterized by 37 parameters: 15 externally calibrated and 22 jointly estimated.\footnote{Appendix Section~\ref{sect:app_computation} provides more details on the model solution and estimation procedure.} Preference parameters (discount factor $\beta$, and agents' risk aversion $\nu$), and the annualized risk-free rate $r_f$ are set in line with the literature. We calibrate the persistence and volatility of the aggregate shock, ($\rho_a$, $\sigma_a$) by estimating an AR(1) on the detrended series of Italian real total factor productivity (TFP). %
In addition, workers draw their innate ability and human capital upon entry into the market from two initial distributions depending on their educational attainment. %
For each type in the model, high--school or college, we pin down the initial distribution by fitting a weighted sum of Beta distributions on the distribution of initial earnings for workers without (high--school) and with (college) tertiary degrees in our sample.\footnote{Specifically, we target the [0,1] rescaled mean, p10, median, p90 and skeweness of the distribution to estimate $\mathcal{B}(\varrho_\iota,\{a_{j,\iota},b_{j,\iota}\}_{j=1,2})=\varrho_\iota\mathcal{B}(a_{1,\iota},b_{1,\iota})+(1-\varrho_\iota)\mathcal{B}(a_{2,\iota},b_{2,\iota})$, for each type $\iota$.} 
\begin{table}[t]
    \centering
    \caption{Parameter Values}\label{tab:calib}
    {\small
    \resizebox{0.8\textwidth}{!}{\input{figures/PaperFigures/CalibrationTable_oldExitRates.tex}
    }}
\end{table}
\paragraph{Estimation and Identification.} We estimate the remaining 22 parameters via SMM, targeting a set of standard labor market moments and to properly account for the empirical age distribution, we weigh each moment from simulated data according to the distribution of the Italian working-age population.\footnote{Age weights are constructed following the age distribution of the 2010 census from the website of the Italian National Institute of Statistics (ISTAT).} \textbf{Table~\ref{tab:calib}} reports estimated parameter values, while \textbf{Figure~\ref{fig:moments}} compares the simulated and empirical moments. 

As shown in \textbf{Figure~\ref{fig:moments}}, the model is able to fit well the profiles of transitions and separations for workers of different age and with different educational attainment. The model is also able to capture the differences in the cyclicality of labor market flows, capturing both the procyclicality of E-E transitions for High-School and College workers and the countercyclicality of E-U transitions for High-School workers, while matching the essentially acyclical separations for College graduates. Similarly, the model fits the earnings growth at the beginning and at the end of workers careers very well. The model undershoots the separation rate of young workers for both High school and College workers. A possible explanation for the undershoot is due to the fact that part of the observed separations in the data is linked to institutional factors (e.g. fixed-term contracts) and peculiarities of the Italian economy (such as the presence of sectors that strongly rely on seasonal workers) that the model is not meant to capture.\footnote{For example, the average share of temporary workers in Accommodation and food services and Agriculture between 2008 and 2018 is 31\% (\textit{Source}: Eurostat, LFS).}

To examine the role of each moment in identifying our parameters, we use two complementary approaches. First, we consider the series generated by our global solver--in which we draw 1,500 points from a Sobol sequence spanning the parameter space--and project each parameter on the simulated moments using a LASSO regression with a 0.8 penalty. For each parameter, non-zero coefficients indicate which moments (or groups thereof) exhibit variation most strongly associated with that parameter. We report the coefficients from this exercise in \textbf{Figure~\ref{fig:app_IDfigLasso}}. Second, we compute an average elasticity matrix from a local perturbation of our parametrization, and examine the moments-parameter linkages with the highest elasticities, we report this matrix in \textbf{Figure~\ref{fig:IDelas}}. 

In both cases, the results reveal intuitive patterns linking moments to parameters. The matching function elasticity ($\gamma$) and the scale of vacancy costs ($c_{ue}$, $c_{ee}$) are pinned down by the cyclicality of labor market flows and unemployment rates. The elasticity of vacancy costs ($\kappa$) and human capital retention upon job-to-job moves ($\tau_{ee}$) are identified by employment-to-employment (EE) transition rates, as these parameters determine the incentives for workers to climb the job ladder. Employment-to-unemployment (EU) and EE flows--particularly their cyclicality for high-school and college workers--identify the exogenous separation rates ($\delta$, $\delta_g$) and the probability of on-the-job search ($\lambda_e$). 
The variance of idiosyncratic human capital shocks ($\sigma_\psi$) and the unemployment benefit ($b$) are informed by the share of long-term unemployment and the cyclicality of EU transitions. Earnings growth profiles by education and experience identify the human capital accumulation process, specifically the scaling effect of firm quality ($\xi$), the accumulation rates ($\phi$, $\phi_g$), and the scaling in the upper bound of initial human capital ($\vartheta$). Cross-sectional moments, including worker-firm sorting and the value-added distribution, influence the production elasticity ($\alpha$), the lower bound of firm quality ($\underline{y}$), and the cyclical cost component ($x$). Finally, human capital losses during unemployment ($l_s$, $l_g$, $\tau_{eu}$) are identified by the persistence of earnings losses and unemployment moments, while the labor force exit threshold ($\underline{p}$) and the link between human capital and unemployment benefits ($\xi_b$) are governed by the observed unemployment durations and the cyclicality of labor market transitions for graduates.
\begin{figure}[!th]
\caption{Empirical and simulated moments \label{fig:moments}}
\includegraphics[width=\textwidth,height=0.4\textheight]{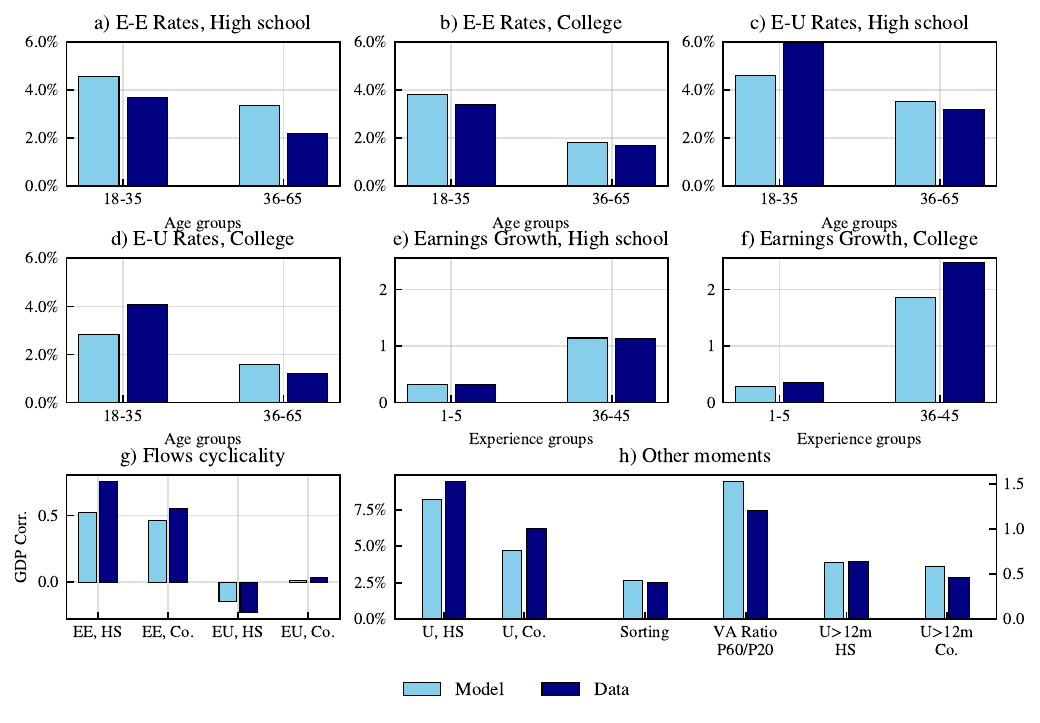}
\vspace{-2em}
\floatfoot{\textbf{Note:} The figure plots the target moments in the data and in model simulations.}
\end{figure}

\section{Model properties and validation}\label{validation}

To assess estimation quality, we evaluate the model on untargeted moments across the cross section and over time, at both the micro and macro levels. To isolate the role of the human-capital ladder, we also re-estimate an alternative specification in which human capital is independent of firm quality and evolves linearly, as $h_{\iota,t+1} = \phi_{\iota} + h_{\iota, t} + \epsilon_{\iota,t}$, on the same outcomes. We will refer to this alternative specification as ``\textit{the linear model}''.\footnote{We report comparisons with two additional alternative models, one with flexible wages and one with no endogenous separations, in \textbf{Appendix~\ref{sect:app_additional_tables_andfig}}.}  %

\subsection{Experience returns and firm quality} 
As discussed in \textbf{Section~\ref{sect2}}, the shape of experience returns and the significant association between previous employers and wages after E-U-E transitions are features that are consistent with workers accumulating human capital on-the-job. \textbf{Figure~\ref{fig:ABS_EUE_modelData}} provides a comprehensive validation of our baseline model along these two dimensions. 

From a qualitative standpoint, our baseline model accurately captures both the shape of experience returns across firm quality quintiles, shown in \textbf{Figure~\ref{fig:ABS_modelData}}, and the persistent effect of origin firms on post-unemployment wages, plotted in \textbf{Figure~\ref{fig:EUE_modelData}}. The linear model is not able to replicate both the shape of experience returns and the effect of previous employers on post-unemployment wages. 

\begin{figure}[h!]
\centering
\caption{Experience returns and persistent effect of origin firms: Model vs Data\label{fig:ABS_EUE_modelData}}
\subcaptionbox{Experience returns\label{fig:ABS_modelData}}{
\includegraphics[width=0.47\textwidth]{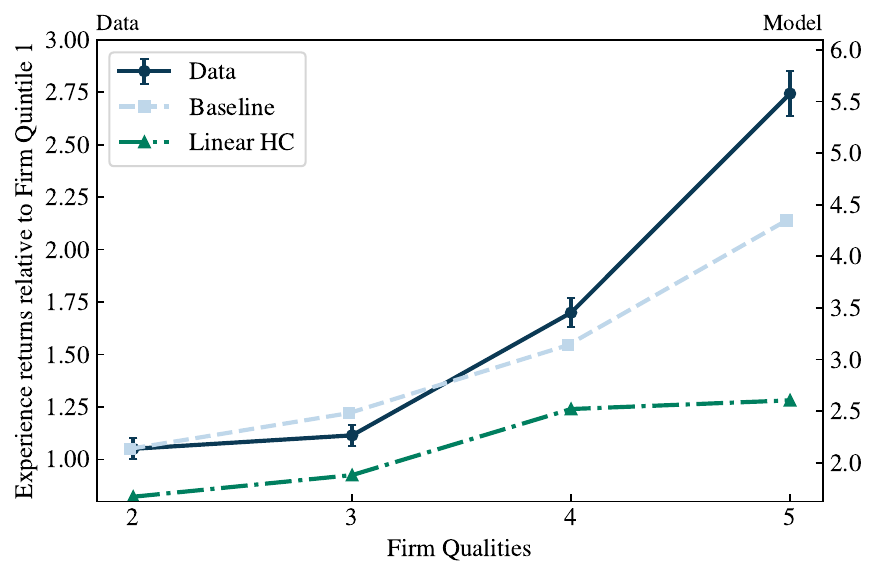}}%
\subcaptionbox{Persistent effect of the firm of origin\label{fig:EUE_modelData}}{\includegraphics[width=0.47\textwidth]{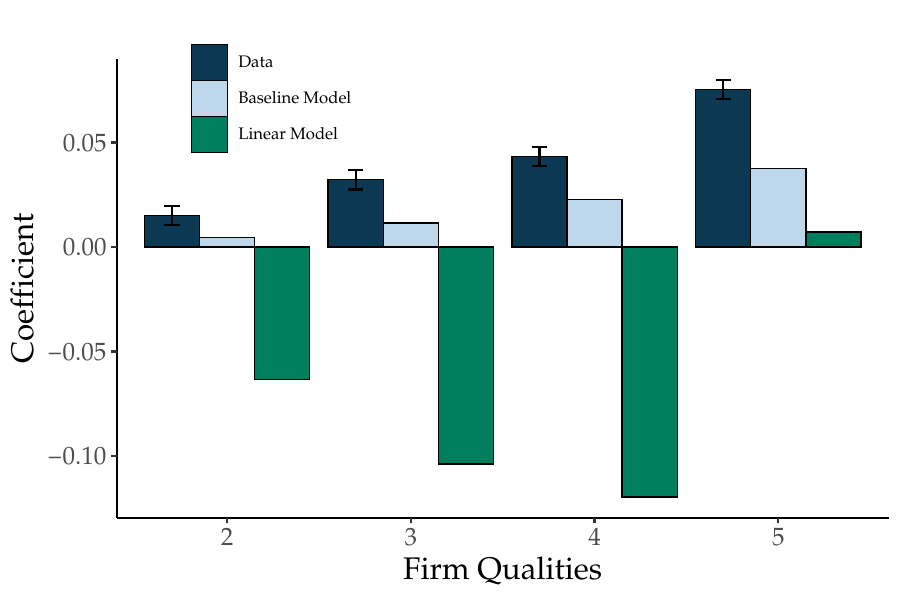}}
\floatfoot{\textbf{Note:} Panel (a) plots the yearly returns to experience for firms at different quintilies of the productivity distribution, \textit{Data} on the left-hand side axis and \textit{Models} on the right-hand side, scaled to align with the Data for the \textit{Baseline}. We run a version of \textbf{Equation~\ref{eqn:experience_profiles}} on our baseline model and two alternatives: one in which human capital follows a linear process, independent from firm quality and another one in which firms are able to freely adjust wages downward. In all, we do not distinguish experience accrued within and outside the firm class and report the average experience returns for all models and the estimates in our sample. Panel (b) reports a comparison of the specification in \textbf{Equation~\ref{eq:eue_reg}} on a sample of E-U-E transitions, at the 1 year horizon. We again compare the empirical estimates with model outcomes.}
\end{figure}

Returns to experience are strongly convex in the data: top-quintile firm experience yields roughly triple the returns of bottom-quintile experience. The baseline model reproduces this convexity, whereas the linear model delivers a concave pattern, since wage growth becomes tenure-driven when human capital gains do not depend on firm quality. 

In \textbf{Section~\ref{sect2}}, we show a monotonic relationship between origin-firm quality and post-E - U--E reemployment wages: workers from higher-productivity firms are rehired at higher wages. Only the baseline model captures this gradient, while the linear model does not show any origin-firm dependence. In the linear model, firm quality does affect wages contemporaneously through the match, but does not independently impact human capital accumulation. For this reason, the wage premium from working at a high-productivity firm disappears at separation. After displacement, high-productivity workers have strong incentives to exit unemployment quickly, and because a uniform human-capital technology reduces the gains from targeting top firms, they concentrate search on lower-productivity firms. Conditional on the past wage, this effect determines a negative (or insignificant, high up the ladder) relationship between the quality of past employment and the re-employment firm quality.

\subsection{Cross-sectional properties} While in the estimation we target the average unemployment rate, the model exhibits a good fit also for the age profile of the unemployment rate (see  \textbf{Figure~\ref{fig:unemployment}}). 
In addition, \textbf{Figure~\ref{fig:ue_model_data}} compares the implied unemployment-to-employment rates by age in the baseline model and in the data, highlighting how the baseline model is able to capture the age profiles of transitions out of unemployment well. The model with linear human capital accumulation struggles to match the unemployment rates, which leads to an overshoot of U-E rates. 
\begin{figure}[h!]
    \caption{Cross-Sectional Features, Model and Data}\label{fig:bigvalunemp}
    \subcaptionbox{Unemployment rates\label{fig:unemployment}}{\includegraphics[width=0.425\textwidth]{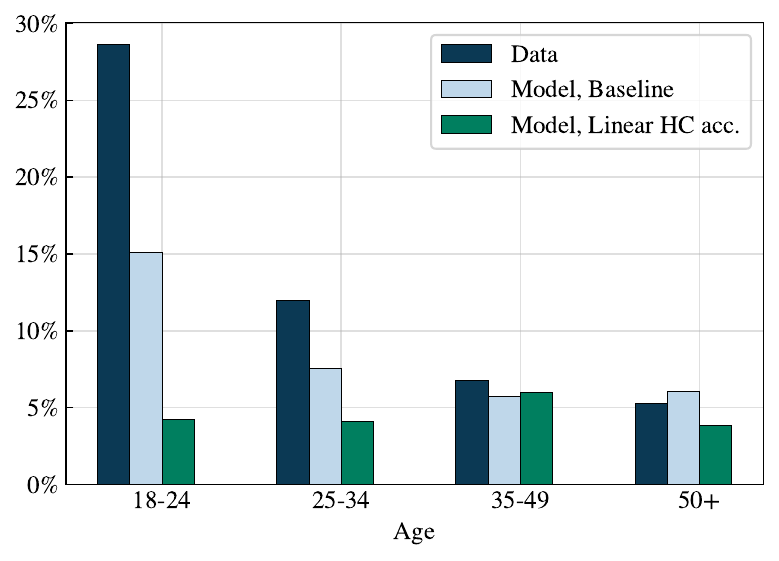}}
    \subcaptionbox{Unemployment-to-Employment rates\label{fig:ue_model_data}}
    {\includegraphics[width=0.425\textwidth]{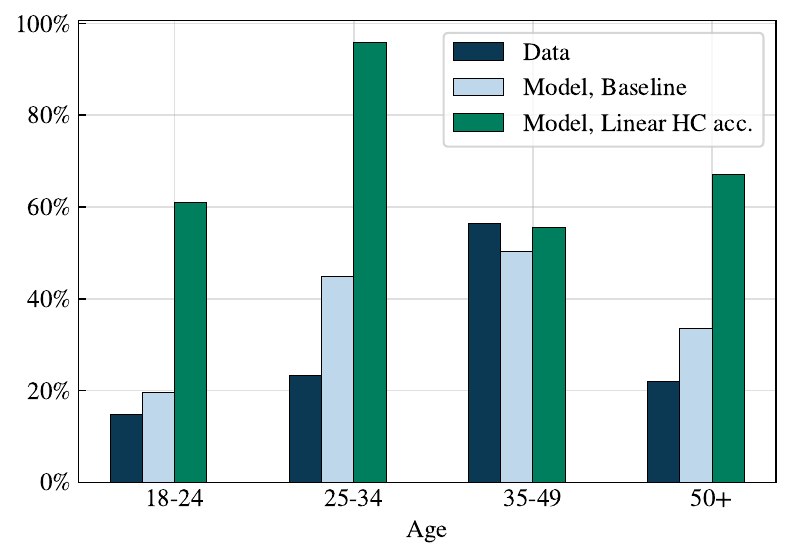}}
    \vspace{-1em}
    \floatfoot{\textbf{Note:}  Panel~(\textbf{a}) plots the unemployment rate by age groups in model simulations and in the data. 
    \textit{Sources:} unemployment rates are taken from the Italian National Statistical Agency (ISTAT). Panel~(\textbf{b}) reports the average UE rates by worker age.}
\end{figure}

Our baseline model is also able to qualitatively reproduce the distribution of earnings. \textbf{Figure~\ref{fig:wage_dist}} in the appendix, displays the cross-sectional distribution of earnings in the data and in the model. The empirical wage distribution is centered at slightly below \euro $2,000 $ and skewed to the left, with most observations below \euro $ 4,000 $. What the baseline and linear human capital model fail to generate is the long right tail of wages in the data, which corresponds mainly to managerial figures whose earnings command premia that our mechanism is not meant to capture.

\subsection{Business cycle properties} \paragraph{Aggregate dynamics.} Replicating aggregate time-series properties of the data provides additional validation of the channels in the model. We project the detrended quarterly series of Italian TFP on a discrete grid to simulate a series of discretized aggregate shocks. The simulated model tracks the empirical series of GDP, capturing peaks and troughs as well as the overall behavior of the empirical series (see appendix \textbf{Figure~\ref{fig:gdp_data_model}}). As shown in \textbf{Table~\ref{tab:validation_table}}, besides approximating well the volatility of output (the standard deviation of detrended log-output is close to 3\% both in the model and in the data) the model is able to generate also some volatility for the unemployment rate, which is 0.9\% in the model versus 1.4\% in the data.\footnote{To compute the volatility of unemployment we replicate the exercise in \cite{Shimer2005}. Specifically, we remove a slow moving trend from the time-series of unemployment, by filtering the quarterly series with an Hodrick-Prescott filter with smoothing equal to $10^5$.}

{
\begin{table}[h!]
\centering
\caption{Co-movements at business cycle frequency \label{tab:validation_table}}
\resizebox{0.6\textwidth}{!}{\input{figures/PaperFigures/ValidationTable_BaselineLinear.txt}}
\floatfoot{\textbf{Note:} \footnotesize{We report the correlation between model simulated series and their data counterparts (Panel~A), the correlations with aggregate output in the model and in the data (Panel~B) and the standard deviations of output and unemployment (Panel~C). Simulations are obtained by feeding the model with a TFP series  that matches the Italian TFP from 2000Q1 to 2019Q4. All series have been  detrended.  Panel~A and~B: Hamilton filter (4 lags and 8 leads); Panel~C: Hodrick-Prescott filter (smoothing equal to $10^5$, as in \cite{Shimer2005}). We report the correlations for our baseline model and for the alternative, in which human capital accumulation does not depend on firm quality and is given by $h_{\iota,t+1} = \phi_{\iota} + h_{\iota, t} + \epsilon_{\iota,t}$ (\textit{Linear Human Capital}). The alternative model has been re-estimated as described in \textbf{Section \ref{calib_est}}.}}
\end{table}
}

Further validation of our framework is given by the ability to replicate the long-run effects of business cycles on workers' career outcomes. In particular, we adapt the reduced-form models proposed in the literature on the effects of recessions on labor market entrants (\citealt{kahn2010}, \citealt{schwandt2019}) and we run it on both the Italian administrative data and on a model-simulated panel.\footnote{In these empirical specifications we control for age, period, and cohort effects, proxied by the cyclical realization of GDP at cohort entry. We report the empirical estimates in \textbf{Appendix~\ref{sect:app_additional_tables_andfig}}.} %
Entering the labor market in a downturn is associated with persistent losses in earnings in our framework. As shown in  \textbf{Figure~\ref{fig:scarring_model_data}}, our baseline model is able to generate scarring effects that are in line with those observed in the data. Baseline simulations match both the magnitude and the dynamics of the empirically observed scarring effects, while this ability is lost by the linear model, highlighting the importance of sorting into different early-career employers for future earnings trajectories.
The model also matches well the concentration of separation rates along the productivity distribution--as shown in \textbf{Figure~\ref{fig:separationsShares}}. As in the data, separations are concentrated  in firms on the lower rungs of the productivity ladder, with this pattern more pronounced in recessions.

\paragraph{Job ladder cyclicality.} \textbf{Table~\ref{tab:differentialnetjob}} reports results from the main specification following \cite{haltiwanger_cyclical_2018,haltiwanger_cyclical_2025} for our baseline model and for the linear model. The baseline model captures both the order of magnitude and sign of the correlation between job flows and changes in the unemployment rate (Panel B). Results for the linear model are, on the other hand, highly counterfactual with respect to what we see in the data (Panel C). Recessions do not seem to generate significant changes in reallocation patterns of labor along the productivity ladder. Absent firm-dependent learning, differential net poaching and flows from unemployment across the productivity distribution are statistically uncorrelated with the business cycle.%

\begin{table}[h!]
\centering
\caption{Cyclical job ladder collapse \label{tab:differentialnetjob}} {\resizebox{0.65\textwidth}{!}{
\input{figures/PaperFigures/HHMS_comparisonsSmall.tex} }}
\vspace{-1.0em}
\floatfoot{\textbf{Note:} The table reports the cyclicality of the difference in Net poaching and Net nonemployment flows as in  \cite{haltiwanger_cyclical_2025} and \cite{haltiwanger_cyclical_2018}. Each row presents results from a separate regression. The dependent variable is the differential worker flow rate between high- and low-productivity firms. The independent variables in each regression include a cyclical indicator, a linear and quadratic time trend, and a constant, which are not reported. We use quarterly data, from 1996 to 2018.}
\end{table}

\section{Anatomy of recessions}\label{sect:anatomy}

The model developed in \textbf{Section~\ref{sect:model}} can be used to analyze 1) how aggregate shocks propagate through the economy through changes in firm and worker sorting, and 2) what channels are responsible for their persistent effects on aggregate output.  

To illustrate these points we provide a a decomposition of the economy's cumulative impulse responses to a TFP shock. We do this by first computing generalized impulse response functions. We compare a series of simulations of the model with aggregate shocks to otherwise identical counterparts, for which the only difference is that, at some exogenous time chosen to be the same for all simulations, the economy experiences three consecutive negative realizations of the TFP process.\footnote{The cumulative drop in TFP is approximately $15\%$ (equally split in each quarter) and is consistent in magnitude to the TFP drop in the three quarters following the 2011 sovereign debt crisis.} %
We then focus on labor market outcomes of affected cohorts and the response of aggregate GDP throughout these specific recessions. We do this for our baseline model and the model where human capital accumulation is independent of firm quality. This allows to gauge how misspecifying the human capital accumulation process can lead to a misdiagnosis of the transmission channels of recessions. We illustrate how the shock propagates in \textbf{Figure~\ref{fig:recession_experiment}} and we then provide a decomposition of the response of aggregate output into sorting, human capital and displacement channels in \textbf{Figures~\ref{fig:decomp_irf_output}}~to~\textbf{\ref{fig:decompByHC_irf}}.

\begin{figure}[t]
    \centering
    \caption{Recession Experiment: Baseline and Linear Model}
    \label{fig:recession_experiment}
    \subcaptionbox{Firm quality\label{fig:rec_exp_fq}}{\includegraphics[width=0.33\textwidth]{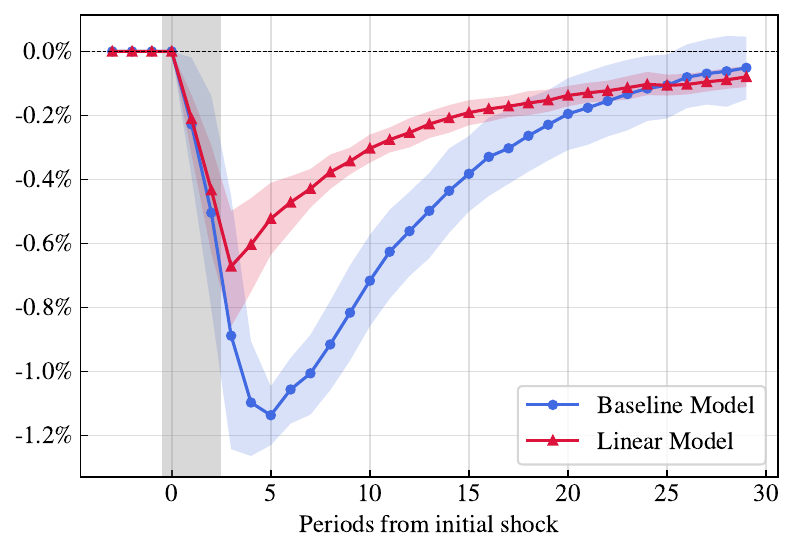} }%
    \subcaptionbox{Human capital\label{fig:rec_exp_hc}}{\includegraphics[width=0.33\textwidth]{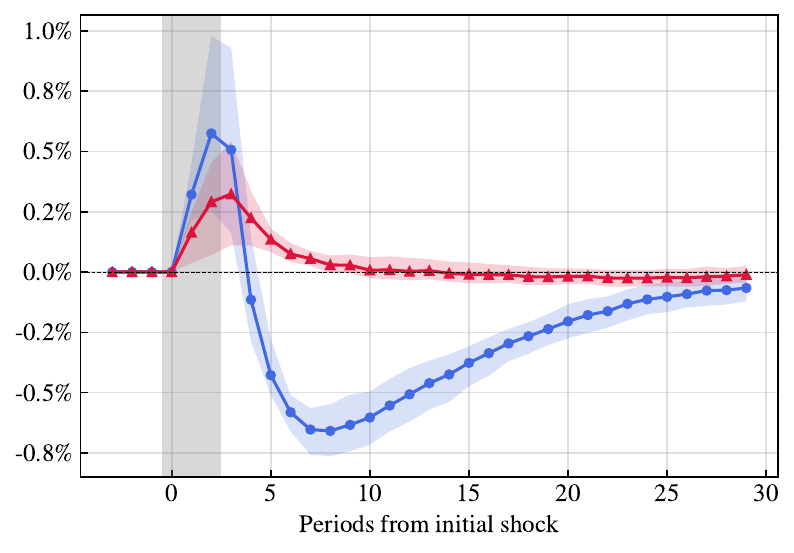} }%
    \subcaptionbox{Sorting \label{fig:rec_exp_sorting}}{\includegraphics[width=0.32\textwidth]{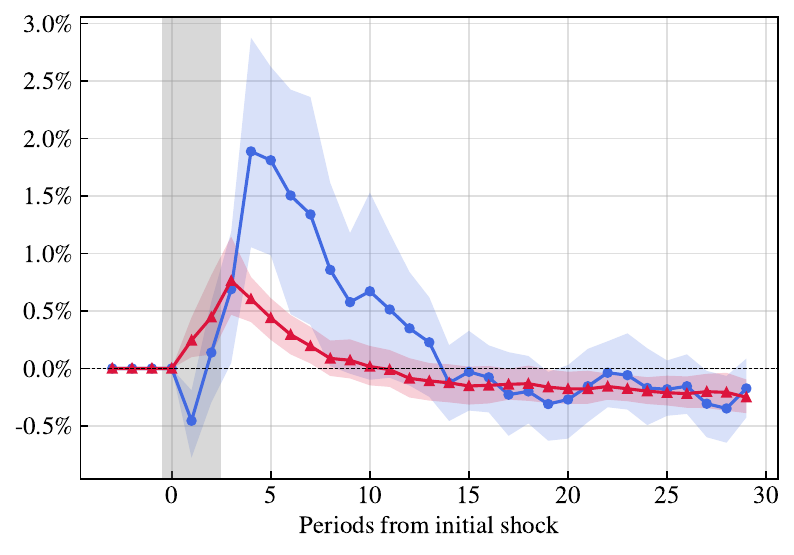}}%

    \subcaptionbox{Average wage\label{fig:rec_exp_Wageirf}}{\includegraphics[width=0.33\textwidth]{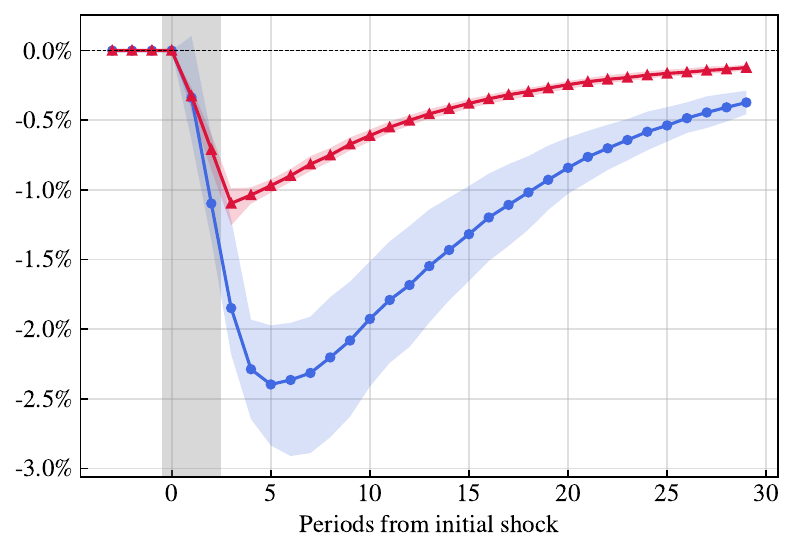}}%
    \subcaptionbox{Aggregate labor share\label{fig:rec_exp_LSirf}}{\includegraphics[width=0.33\textwidth]{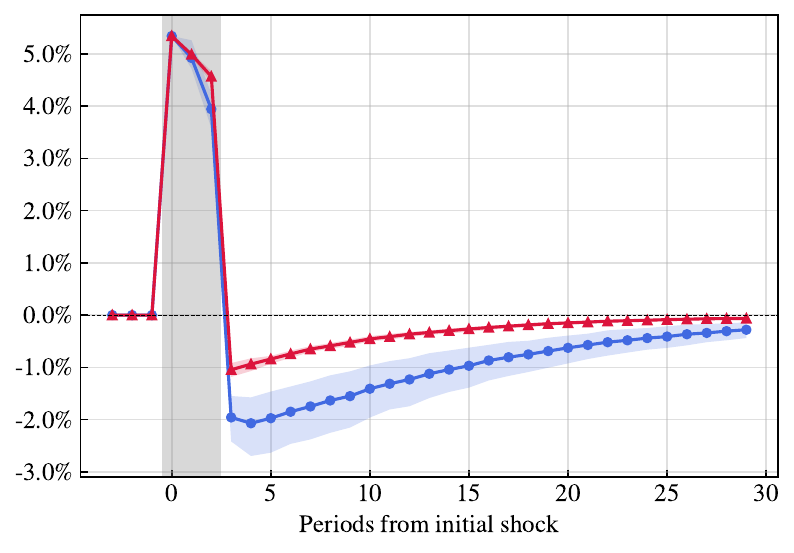}}%
    \subcaptionbox{Aggregate output\label{fig:rec_exp_Yirf}}{\includegraphics[width=0.33\textwidth]{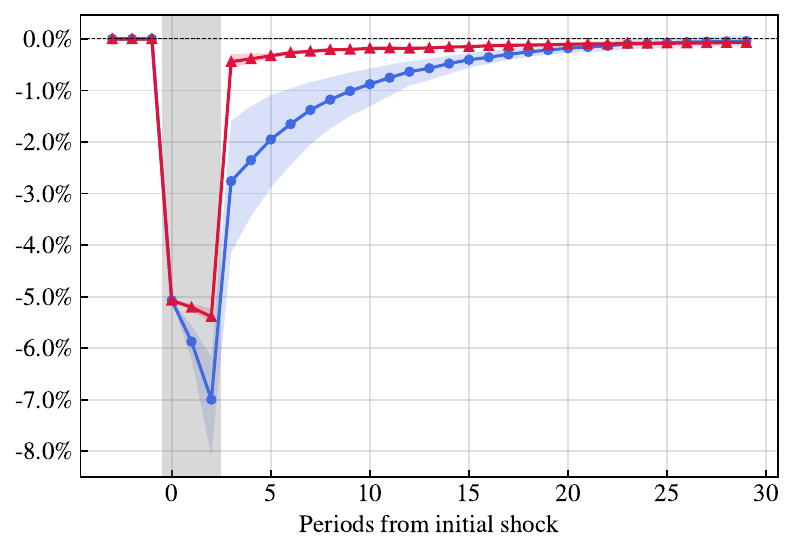}}%
    \vspace{-1em}
    \floatfoot{\textbf{Note:} The panels in the figure plot the ratios of the aggregate variables between an economy in which we impose a three-quarter negative TFP realization and an economy without aggregate shocks, that serves as a benchmark. The gray shaded area are the quarters in which TFP is below trend, while the shaded areas are the 90-10 quantile ranges across one hundred model simulations.
    }
\end{figure}

\paragraph{Shock propagation.} In \textbf{Figure~\ref{fig:rec_exp_fq}} and \textbf{Figure~\ref{fig:rec_exp_hc}} we show the generalized impulse response functions of average firm quality and human capital. The onset of a recession is characterized by a sort of ``Schumpeterian'' dynamics, as implied by the initial marginal increase in the levels of human capital. \textbf{Figure~\ref{fig:rec_exp_sorting}} shows our measure of sorting, the correlation of firms and workers' quality, relative to the baseline simulation. After a short-lived improvement, firm-worker sorting is reverts back to its pre-recession levels, albeit at very different levels of worker and firm quality as shown by the dynamics of average human capital and average firm quality. 
 A prolonged drop in the quality of factors of production, in fact, increases the persistence of the initial shock on output beyond the original duration of the recession. \textbf{Figure~\ref{fig:rec_exp_Yirf}} shows that aggregate output is still about {\color{black} 0.8\%} below its counterfactual level two years from the start of the recession. %

The recession has a strong effect on both the average wage and the average labor share in the economy, as shown in \textbf{Figure~\ref{fig:rec_exp_Wageirf}}~and~\textbf{\ref{fig:rec_exp_LSirf}} respectively. Due to the downward wage rigidty embedded in the contracts, average wages drop less than the average firm output. This allows for a brief spike in the labor share that quickly declines and remains below the counterfactual economy for a long time. The decline occurs for two reasons. First, the matches that form in the recovery period are subject to the sullying in firm and human capital qualities, lowering the average levels of starting wages. Second, due to the back-loading of compensation embedded in contracts, new contracts feature initially lower labor shares and greater profit shares for firms. %

The triangle-red lines in \textbf{Figure~\ref{fig:recession_experiment}} show the model with linear human capital accumulation. Its response to a TFP shock differs sharply from the baseline: the initial cleansing effect is re-absorbed as less productive workers who were initially displaced  exit unemployment. Employment losses in the linear model are less than one-third of those in the baseline, leaving too few unemployed workers requiring re-employment to mechanically push average human capital below the baseline. As a result, there is no aggregate loss in human capital in the transition. These muted human capital dynamics are mirrored in smaller declines in firm quality and substantially smaller losses in average wages and aggregate output. 

{\color{black} Why is human capital production is more sensitive to temporary shocks in the baseline model? Even though human capital grows faster at more productive firms, and those firms are less exposed to temporary revenue shocks, low-productivity firms remain central: young, low--human capital workers sort into below-median productivity, more volatile employers, which account for over half of aggregate learning. Top firms, instead, mostly employ seniors with little remaining accumulation. Changes in sorting -- in particular, changes in the pattern of job-to-job transitions -- within the volatile, lower-productivity firms drive business cycle human capital dynamics. This is consistent with evidence in \cite{gregory2023}, who finds that most human capital accumulation takes place in less productive firms.\footnote{\citeauthor{gregory2023}'s \citeyearpar{gregory2023} identification relies on the assumption that learning on the job is absent for older workers, an assumption validated by our results.} The linear specification masks this mechanism: among continuously employed workers, human capital accumulation is invariant to recessions, implying no sullying effects operating through job-to-job reallocation.}  
 
Finally, the baseline economy exhibits significant state dependence on the labor share---a 1pp higher average labor share is associated with a $\sim 2\text{pp}$ larger output loss in the first year (\textbf{Table~\ref{app_tab:state_dependence}} in the appendix)---whereas the linear specification displays much weaker state dependence, as reflected in the narrow swathes around the IRFs.

\subsection{Decomposing recessions} 

What explains the amplification and persistence of recessionary shocks? Different competing channels are at play. The first, which we call the human capital channel, captures the human capital accumulation that does not take place because of the recessionary event, keeping human capital stock as if the recession did not happen. The second, the sorting channel, amounts to the firm component of the different joint worker-firm distributions that emerge in the periods following the shock,  because of  different match formation along the cycle. The third, the search channel, is the impact of different search probabilities due to changes in aggregate and individual states. Finally, the standard displacement channel captures the job destruction that takes place because of the negative shock and its spillovers. We decompose the effects of recessions in the model economy by shutting down each channel at a time and then comparing the resulting dynamics to the one of a baseline recession.\footnote{Specifically, for the search component we run model simulations in which the post-recession job-finding probabilities are those associated with search in the baseline simulation. Similarly, for the firm quality component of the sorting channel we keep employed workers in the same firm quality they have in their baseline simulation. For the human capital channel we erase human capital losses by forcing workers' human capital in our counterfactual simulation to be the same as in the baseline one. The displacement channel is then obtained residually.}

\begin{figure}[h!]
    \centering 
    \caption{Decomposition of cumulative impulse response function of aggregate output\label{fig:decomp_irf_output}}
    \subcaptionbox{Baseline model\label{fig:decomp_irf_output_baseline}}{\includegraphics[width=0.45\textwidth]{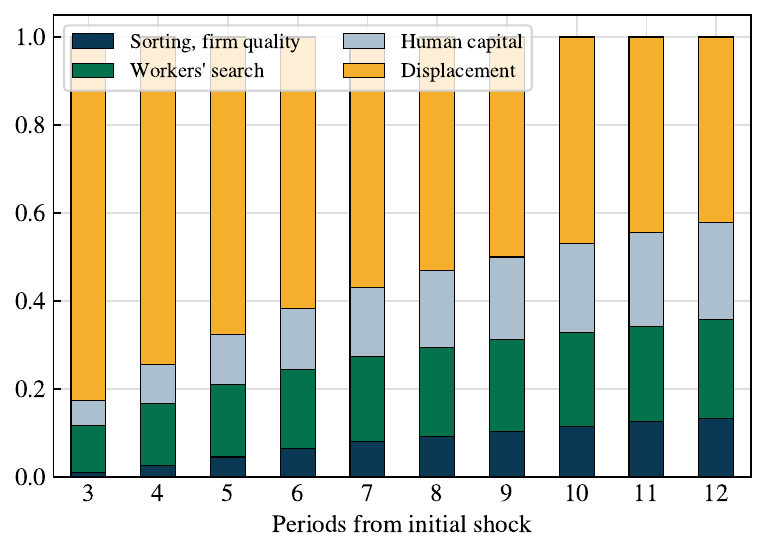}}%
    \subcaptionbox{Linear human capital model\label{fig:decomp_irf_output_linear}}{\includegraphics[width=0.45\textwidth]{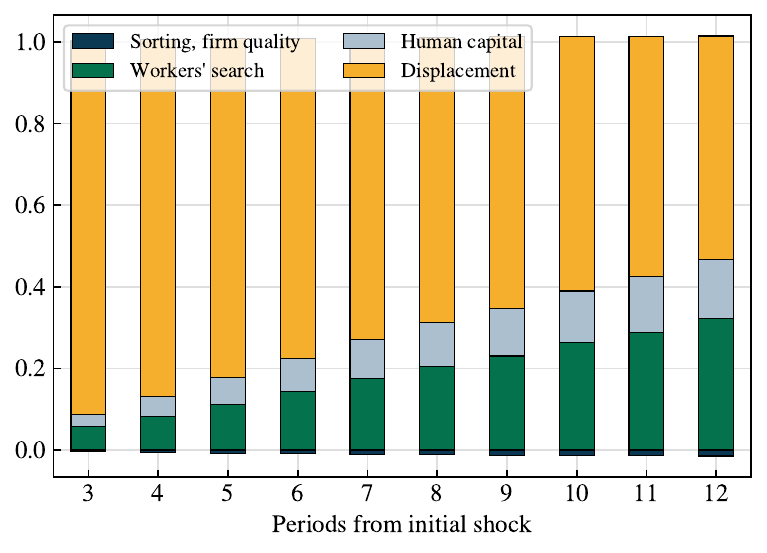}}%
    \floatfoot{\textbf{Note:} The figure shows the relative importance of each transmission channel compared to the baseline recession for the cumulative response of GDP in the three years after the onset of the recession.}
\end{figure}

\textbf{Figure \ref{fig:decomp_irf_output}} decomposes contributions to output after a shock into our four channels. The displacement channel is the main driver of the recession on impact, explaining the majority of output losses in the first year after the shock. During the recovery, the loss in output can be explained via a combination of a lower search quality from workers as well as a relatively lower firm quality for those that re-match during the transition. Recovery from displacement, while not immediate as search is frictional, is relatively fast, in part because unemployed workers have lower reservation wages. 

Taken together, the sorting and human capital channels explain about 35\% of cumulative output losses three years after the shock in the baseline model (\textbf{Figure~\ref{fig:decomp_irf_output_baseline}}). While negligible in the short-run, the importance of lower firm quality builds up over time, amounting to approximately 13\% after three years from the onset of the negative TFP shock. Human capital losses, instead, account for approximately 9\% in the short-run and up to 22\% at the three-years horizon. Worsened workers' search, finally, accounts for approximately 14\% of output losses in the first year of the recovery and up to 22\% three years after the start of the recession.

In the linear model the transmission of the negative TFP shock is significantly different (\textbf{Figure~\ref{fig:decomp_irf_output_linear}}) . The sorting and human capital channels explain roughly 14\% of cumulative output losses three years after the shock. Worsened firm quality is essentially unimportant (even marginally negative) for the dynamics of aggregate output after a negative TFP shock, with the largest effects coming from worsened workers' search--accounting for approximately 32\% after three years--and lower human capital accumulation amounting to approximately 15\% in the medium-run. 

\begin{figure}[h!]
    \centering
    \caption{Decomposing output's losses across the human capital distribution\label{fig:decompByHC_irf}}
    \subcaptionbox{Baseline model \label{fig:decompByHC_irf_output_baseline}}{
    \includegraphics[width=0.8\textwidth]{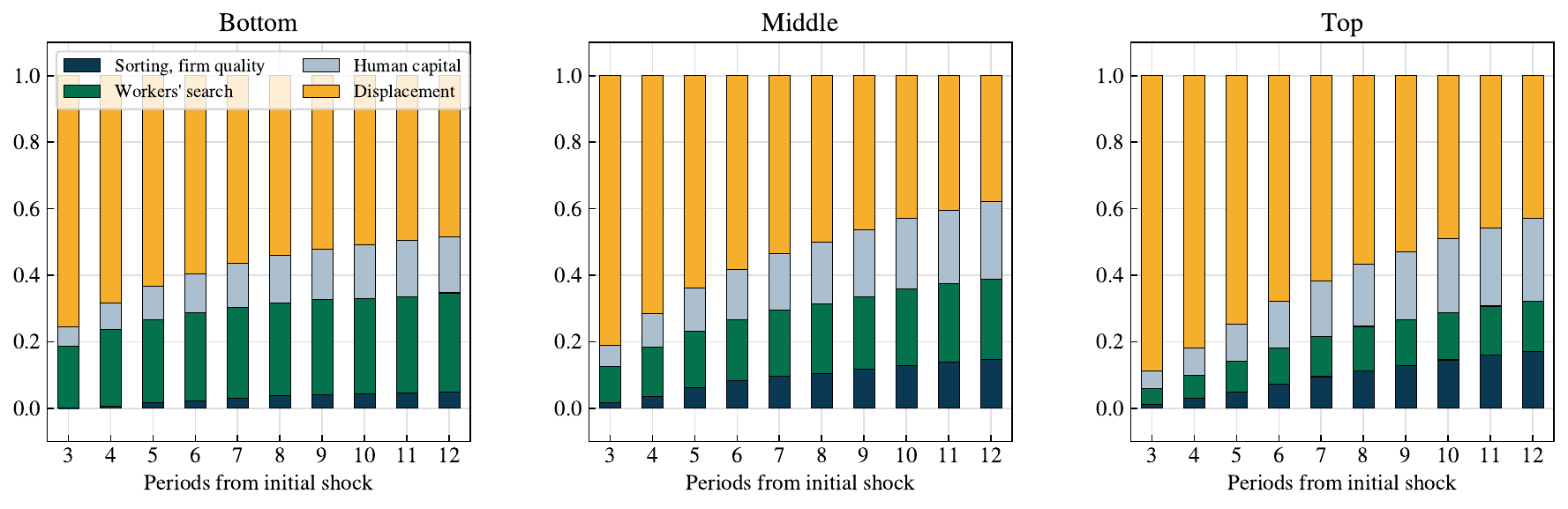}}
    
    \subcaptionbox{Linear human capital model\label{fig:decompByHC_irf_output_linear}}{\includegraphics[width=0.8\textwidth]{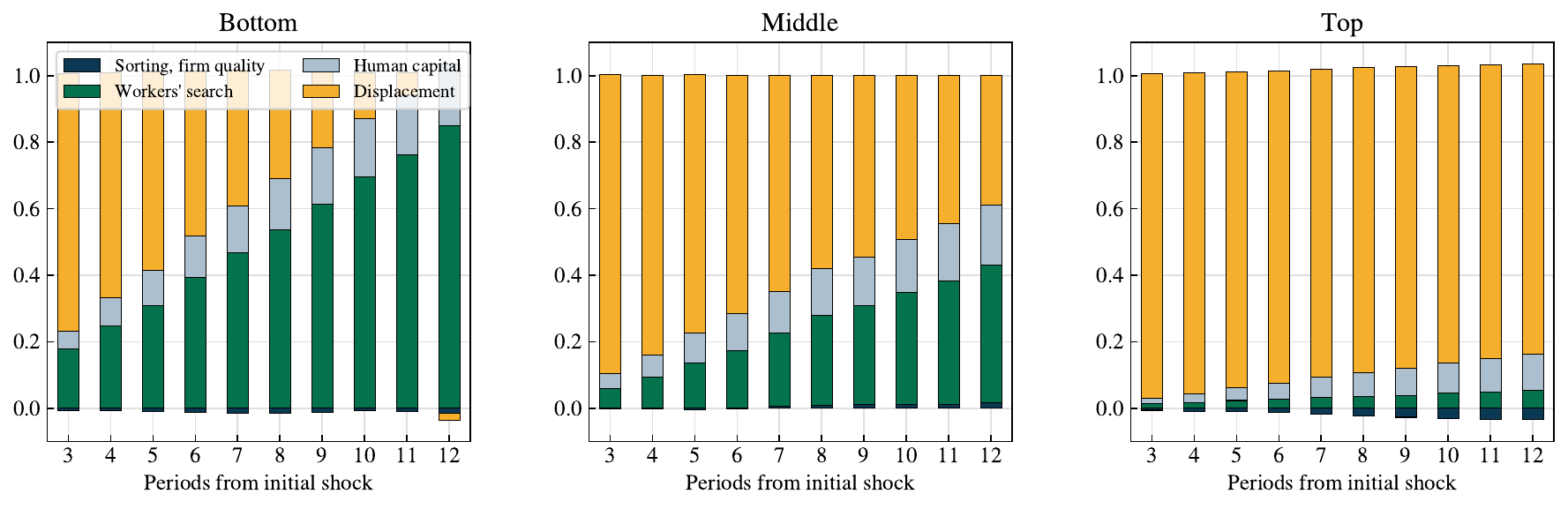}}
    \vspace{-1.0em}
    \floatfoot{\textbf{Note:} For each human capital group, Panel~(\textbf{a}) shows the relative importance of each transmission channel compared to the baseline recession for the cumulative impulse response of GDP across different age groups in the two years after the onset of the shock. Panel~(\textbf{b}) plots the same decomposition for the model with linear human capital accumulation.}
\end{figure}

\paragraph{The role of human capital heterogeneity.} These four channels vary in importance across workers'  human capital distributions, as shown in \textbf{Figure~\ref{fig:decompByHC_irf}}.\footnote{We report the heterogeneous decomposition by age in \textbf{Appendix~\ref{sect:app_additional_tables_andfig}}.} 
The two models diverge sharply across the human capital distribution. In the baseline model, different channels dominate at different rungs: search behavior drives losses for workers at the bottom, while human capital and firm quality channels dominate at the top. This pattern reflects that College workers---who populate higher human capital rungs---accumulate human capital faster in the model. When recessions compress the job ladder, these workers suffer disproportionately from reduced opportunities to climb the human capital ladder. This mechanism is absent in the linear model, where human capital grows independently of employer quality. There, non-employment is the only relevant channel for high human capital workers, while those at lower rungs suffer primarily from prolonged unemployment that worsens their human capital accumulation and subsequent job prospects.

These differences in how recessions propagate---and which workers are most affected through which channel---carry important policy implications. The linear model suggests interventions targeting search frictions, while the baseline model highlights the importance of preserving job quality and match formation. Misspecifying the source of human capital accumulation obscures these distinctions.%

\subsection{The costs of recessions}

\paragraph{Costs of Business Cycles.}  \cite{lucas1987} argues that welfare gains from reducing business cycle volatility are negligible, and quantifies them in less than $0.05$ percentage points of aggregate consumption. By doing an analogous calculation, we estimate the cost of business cycles to be, on average, close to $4$ percentage points,  quantitatively close to the one in \cite{barlevy2004}, who first theorized the potentially sullying effects of recessions. %
\begin{figure}[h!]
    \centering
    \caption{Business Cycle Costs}
    \subcaptionbox{Baseline model\label{fig:cost_businesscycles}}{\includegraphics[width=0.35\textwidth]{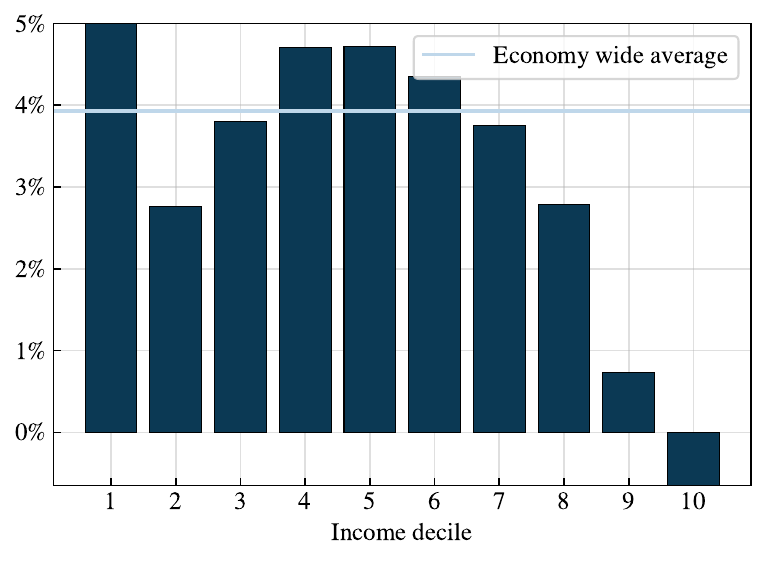}}%
    \subcaptionbox{Linear human capital model\label{fig:cost_businesscycle_age}}{\includegraphics[width=0.35\textwidth]{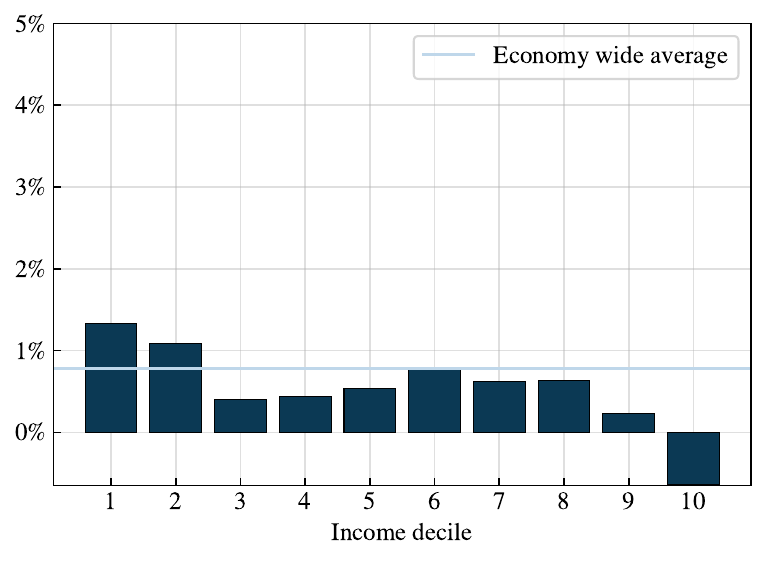}}
    \vspace{-1.0em}
    \floatfoot{\textbf{Note:} Panel~(\textbf{a}) and (\textbf{b}) plot the reduction in consumption-equivalent utility due to aggregate fluctuations by income deciles in the baseline model and in the model with linear human capital accumulation. Workers are assigned to income deciles in the simulations with aggregate fluctuations.}
\end{figure}
The rich heterogeneity in our model allows us to estimate how the welfare costs of business cycles vary along the income distribution, and analyze how they interact with income inequality: \textbf{Figure \ref{fig:cost_businesscycles}} plots the cost of business cycles for different income deciles in our baseline model.\footnote{ In \textbf{Appendix~\ref{sect:app_additional_tables_andfig}} we also show, consistently with \cite{heathcote2020}, that income inequality increases throughout recessionary episodes in our model simulations.}  Highest deciles bear virtually no costs from business cycle fluctuations. This is due to the fact that the very top of the distribution comprises the highest skill workers, who have no cyclical unemployment risk and have reached the maximum human capital attainable given firm opening costs and search frictions. The welfare costs of business cycles are, however, much larger for all other deciles. Business cycles costs are high for workers at the bottom of the income distribution, as those are workers that suffered the most aggregate fluctuations, via job losses and impaired human capital accumulation. Interestingly, the cost of business cycles remains particularly high also for workers up to the $8^{th}$ decile of the distribution, indicating that only the very top workers are sufficiently insulated by aggregate fluctuations.
In \textbf{Figure \ref{fig:cost_businesscycle_age}} instead, we report welfare costs in the model with linear human capital accumulation. Welfare cost of business cycles are significantly smaller and more uniform than those in the baseline model. Once human capital accumulation, and consequently workers' career prospects, are independent from the quality of employers workers along the entire income distribution are more insulated from aggregate fluctuations and the cyclical quality of job creation.

\section{Conclusions \label{sect:conclusion}}
We develop a novel framework of on-the-job search and human capital accumulation with two-sided heterogeneity to study the dynamics of human ladders. Search frictions and aggregate uncertainty prevent efficient worker--firm allocation, and because human capital accumulation depends on match quality, recessions induce persistent distortions in sorting. These distortions are slow to unwind, contribute to sluggish recoveries, and generate long-run changes in the joint distribution of firms and workers. Quantitatively, they account for about 35\% of the cumulative output loss three years after a recession.

The model is rich yet tractable and naturally suited for policy analysis. It implies that scarring and welfare effects of recessions are heterogeneous across education, skill, and age, reflecting differences in exposure to unemployment risk and human capital deterioration. The framework can be used to evaluate the aggregate and distributional effects of policies such as countercyclical unemployment benefits, short-time work, or employment retention schemes, which we leave for future research.

 \clearpage

%% file: draft_v9_abstract.tex
Using administrative data, we document that workers acquire more human capital at more productive firms. Recessions distort workers-firm sorting, flatten the job ladder and impact human capital accumulation, as workers match on average to worse firms. To quantify the aggregate relevance of these effects, we build a directed search model with aggregate risk and worker-firm heterogeneity, in which human capital accumulation depends on firm quality. We estimate the model and show that recessions have persistent negative effects on the productivity of worker-firm matches, with distortions in sorting and human capital accumulation accounting for approximately 35\% of cumulative output losses.
\bigskip
{ 
\newline \textbf{Keywords}: Human Capital Accumulation; Sorting; Scarring; 
\newline \textbf{JEL Codes}: J24; E32; E24; J63
}

%% file: draft_v9_intro.tex
Extensive research demonstrates that sorting -- the degree to which more productive workers are matched with more productive firms -- is highly cyclical. We introduce novel empirical evidence illustrating how changes in sorting can impact output dynamics across the business cycle. First, using matched employer-employee data from Italy, we establish a novel empirical finding: workers employed in more productive firms accumulate more human capital. Specifically, we show that experience-earnings profiles are steeper for individuals employed in more productive firms, and returns persist long after individuals have moved on from those firms. Over their careers, workers climb these ``human capital ladders''. Second, we show in the data that human capital ladders collapse in recessions: workers who start a new job are on average going to do so in a relatively less productive firm, and will be less able to move towards higher rungs of the firm productivity ladder. Cyclical changes in sorting can then have long-term effects on the stock of human capital in the economy.

The implications of these empirical findings for the allocative effects of recessions are not straightforward. While recessions affect workers by forcing displacements and impacting their lifetime human capital accumulation, they also lead to labor reallocation across firms. This reallocation can potentially result in efficiency gains for the overall economy if workers are induced to move towards better jobs by the cleansing of less productive firms. It is then natural to ask whether recessions induce sullying effects \citep{barlevy2002}, meaning persistent negative impacts on the productivity of workers and firms, or whether they promote long-term, efficiency-enhancing cleansing effects.

Evaluating the relative importance of sullying versus cleansing effects requires tracking the heterogeneity of workers in skill and career progression, their matching with firms of varying productivity, and how these distributions jointly evolve with business cycles. To this end, we develop a directed search model with heterogeneity among firms and workers, an overlapping generations structure, on-the-job search, and aggregate risk. We model human capital ladders: consistent with our empirical evidence, the influence of firms on workers' earnings profiles is captured by match-dependent human capital accumulation. We thus introduce a new channel that links on-the-job learning with labor market sorting and human capital accumulation.

We estimate the model and find that the persistent negative effects of recessions on the productivity of worker-firm matches outweigh any cleansing effects. Our findings are consistent with empirical evidence of initial cleansing effects, followed by persistent losses in productivity \citep{haltiwanger_cyclical_2025}. We find that recessions impair the sorting process and distort the paths of human capital accumulation: nearly one-third of the cumulative post-recession output losses are due to the disruption in the matching of workers to firms, a ``flattening'' of the job ladder, and the consequent erosion of aggregate human capital. 

While prior work has examined displacement costs and cyclical shifts in sorting (e.g., \citealt{lise2017}, \citealt{baley2022}), we provide the first analysis of how firm-dependent human capital ladders shape aggregate sorting over the business cycle. Models focused on job-loss costs (\citealt{jarosch2023}, \citealt{huckfeldt2022}) typically abstract from aggregate sorting implications. By bringing these elements together, we show that the sorting of workers along the productivity ladder varies importantly over the cycle, and that omitting firm-dependent human capital accumulation leads to a substantial understatement of long-run sullying effects.

A key feature of our theoretical framework hinges on the fact that human capital production is linked to the match quality between workers and firms. In our model workers employed by more productive firms experience faster human capital accumulation.\footnote{An expanding body of literature highlights the role of firms for workers' human capital accumulation and skill development (see \citealt{herkenhoff2024}; \citealt{jarosch2021}; \cite{mion2022dream}; \citealt{arellano2022}; \citealt{gregory2023}).} This modeling choice is supported by our novel findings regarding differential experience returns across firms of varying productivity. We establish that, controlling for worker quality, the wage increase attributed to one year of experience at a top productivity firm is over twice that at a firm at the lower end of the productivity spectrum. Furthermore, we find that the experience returns are persistent, maintaining their value through periods of unemployment or displacement, and that this effect is robust to various specifications and definitions of firm quality.\footnote{Changes in the firm class of employment imply earnings losses compatible with high but incomplete human capital portability across firms, or a limited role of firm-specific human capital \citep{kambourov2009}, a feature that we reproduce in the model.}

Workers' pay is shaped by more than human-capital accumulation. Profit-maximizing firms strategically design contracts to attract and retain productive workers. Because employed workers cannot credibly commit to stop searching on the job, optimal contract design must internalize how outside options evolve over the life of a match, especially as workers accumulate human capital and their search incentives change. Moreover, with risk-averse workers, there is a demand for insurance against fluctuations in match value.

The resulting state-contingent contracts front-load firms' profits and back-load workers' compensation, strengthening retention incentives while providing partial insurance against both idiosyncratic and aggregate shocks. Younger workers, who are generally less productive but highly mobile, receive steep wage growth as firms try to boost retention. As human capital accumulates, workers can target better opportunities, raising incentives to move up the job ladder. Job-to-job transitions therefore become a key engine of wage progression \citep{topel1992}, with workers moving from less to more productive firms that offer higher-paying contracts. Overall, life-cycle earnings dynamics reflect the interaction of human-capital accumulation, tenure effects, and ladder climbing.

Capturing cyclical displacements and endogenous layoffs, together with their heterogeneity across workers and firms, requires additional structure. With firm-provided insurance, the wage path specified in the contract is weakly increasing as long as the match has positive surplus, which effectively imposes a wage floor and limits downward adjustment. We also posit that wages are downward rigid, meaning that adverse shocks could lead to separations that are inefficient despite both parties being fully informed about the contract's terms in advance.\footnote{A significant body of research highlights strong resistance to cuts in nominal or real wages, suggesting that layoffs are a more probable reaction to negative productivity shocks (\citealt{altonji1999extent}; \citealt{agell2003survey}; \citealt{grigsby2021aggregate};  \citealt{bertheau2023}). Our analysis further shows that a model without this feature cannot accurately reproduce the observed joint dynamics of output and unemployment or the cyclical nature of firm layoffs.} With these assumptions, the model successfully matches the empirical patterns of separations due to firm layoffs, considering variations across age, worker, and firm productivity levels. Displacements occur more frequently in less productive firms, typically employing younger and less skilled workers, and these firms are also more likely to enact layoffs during recessions. The simulated output and unemployment volatilities in the model match their empirical counterparts, a challenging result for standard search and matching models. %

Two main features ensure the tractability of our framework. First, we prove that workers can map the productivity of each firm to the value of the posted vacancy. This ensures that the optimal search strategy is unique, and that target firm productivity is monotonously increasing in workers' human capital level, reducing the number of states of the problem and excluding the presence of multiple optima in search strategies. Second, the choice of modeling search as directed and the assumption of free entry allow us to prove that the model features a block recursive equilibrium \citep{menzio2010}. The model can then be solved without needing to track the distribution of agents as state variables in the optimization problem.

The model is used as a laboratory to study the interactions between the sorting of workers and firms and the business cycle. The key innovation of our framework is the existence of human capital ladders, that is the full integration of human capital accumulation as a consequence of, and motivating factor for, job mobility. This innovation has important consequences for the understanding of individual and aggregate dynamics of recessions.

Following a recessionary shock, the scarring effects on workers' careers stem from a decline in their outside options and an endogenous reduction in the economy's overall productive capacity. Multiple channels contribute to this outcome. First, many workers are displaced throughout a recession, and they need to restart climbing the job ladder from the bottom.  Second, wage growth decreases for all surviving matches, as workers' outside options deteriorate. Third, the ability of workers moving across jobs and up the job ladder if employed, or towards any job if unemployed, diminishes.

The driving force behind the contraction in job mobility is the flattening of the job ladder for \emph{all} workers, given a decrease in job openings.\footnote{The pro-cyclicality of high-quality jobs has been documented, among others, by \cite{moscarini2016did}.} This effect is stronger at higher rungs of the ladder, where vacancies are more expensive to create and exhibit greater sensitivity to economic cycles. Not only labor demand falls during recessions, but also firms increase their skill requirements for any vacancy type \citep{hershbein2018, modestino2020}, which makes it harder for any worker to climb the job ladder. A flatter job ladder impairs workers' productivity growth and their further potential upward mobility.  Workers match up the ladder with firms of lower quality, form less productive matches and thus accumulate less human capital. Therefore, changes in worker-firm matching amplify and prolong the impact of recessions on earnings and aggregate output by altering their ability to climb their human capital ladder, leading to sullying of the productivity distribution.

The relative importance of distortions in sorting and human capital accumulation can be recovered by decomposing the cumulative output loss of recessions into four channels. Three years after a negative productivity shock, 35\% of losses stems from a worsened sorting between workers and firms and human capital losses: 13\% is caused by the decrease in firm quality in matches formed throughout the recessionary and recovery periods (\emph{firm quality} sorting channel), and 22\% arises from the decrease in human capital accumulation (\emph{human capital} channel). In addition, distortions in workers' job finding probabilities with respect to normal times (\emph{search} channel) account for approximately 22\%. The remaining portion is ascribed to direct displacement effects (the \emph{displacement} channel). This exercise implies that recessions have a sullying effect. 

Incorporating firm-dependent human capital accumulation alters the interpretation of cleansing effects during the initial stages of recessions. While the exit of less productive firms is conventionally seen as cleansing, it predominantly affects younger workers in less productive matches, leading to displacements that have enduring effects (see \citealt{huckfeldt2022, jarosch2023}). Cleansing in the short term may actually hinder long-term growth, as displacements and lower quality employment-to-employment transitions impede human capital accumulation, trapping certain workers at lower levels of the job ladder. 

Because human-capital ladders are central to our results, we validate their role by showing that the model reproduces the untargeted empirical patterns we document, and by benchmarking these outcomes against a standard alternative with stochastic, firm-independent on-the-job human-capital accumulation. Workers have steeper experience-earnings profiles at more productive firms, and these returns rise monotonically along the productivity distribution. The gains persist after any unemployment spell in both model and data, and pre-unemployment firm productivity predicts re-employment wages beyond what pre-unemployment wages imply. The model also reproduces the cyclical dynamics of net poaching and hiring from unemployment flows across the productivity distribution, as in \cite{haltiwanger_cyclical_2025}. The model without firm-dependent human-capital accumulation misses the experience-return gradient, implies counterfactual effects of pre-unemployment productivity on re-employment wages, and fails to capture the cyclical regularities about cleansing and scarring reallocation in \cite{haltiwanger_cyclical_2025}.

Omitting human-capital ladder dynamics therefore understates the importance of matching along the firm-productivity distribution. As expected, a negative productivity shock in an economy without the ladder yields much smaller output effects and less persistent changes in productivity and human-capital distributions. Consequently, any decomposition of post-recession output losses downweights the role of human capital and firm quality, especially for high-skill, highly educated workers.

Finally, we estimate the welfare cost of business cycles. %
With respect to the canonical \cite{lucas1987} estimate (0.05\%), our estimate of the costs is almost two full orders of magnitude larger on average (4\%). The incidence of these costs varies markedly across the income distribution. We find them to be greatest for low and middle-income workers, and not monotonically decreasing over the income distribution. For low--skill workers, which tend to have lower wages, costs are driven by displacement and lower job-finding probabilities. For higher--skill workers, on the other hand, the distortions in firm quality matching and human capital accumulation matter more than the immediate displacement effect. Once again, a model without the human capital ladders delivers way smaller, and more homogeneous values. \footnote{These results for welfare echo and extend the results for earnings found by \cite{heathcote2020} and \cite{doniger2023}, who find that recessions amplify income inequality and tend to have immediate stronger effects for lower income workers.}%

\paragraph{Related literature.} Our work relates to three strands of empirical research. 

First, we contribute to the existing body of empirical evidence on the importance of worker-firm matching in explaining earnings dispersion. Our novel contribution is to show that workers accumulate more human capital when matched in more productive firms. Our empirical exercise is close in spirit to \cite{arellano2022}, who provide empirical methods to classify firms based on firm-dependent earnings profiles. Other studies leverage specific firms characteristic determining differential human capital accumulation within firms: co-workers quality (\citealt{herkenhoff2024}, \citealt{jarosch2021}, \citealt{akerman2025}), international outreach \citep{mion2022dream}, location within cities \citep{roca2017learning},  or specific returns to formal education  (\citealt{engbom2017returns}, \citealt{deming2023wages}). Our paper provides extensive evidence of heterogeneous firm dependent human capital accumulation, increasing in firm productivity, and estimates persistent firm-specific returns to experience on the universe of Italian private employment over 25 years of data and various business cycle conditions. We show that despite this finding, and consistently with what \cite{gregory2023} finds for Germany, in our model most of the human capital accumulation of workers takes place in relatively low productivity firms. This finding emerges endogenously as an equilibrium outcome of the model, absent any ad-hoc restriction on age-dependent learning ability. %

Second, we relate to the empirical literature analyzing the impact of business cycles on workers' earnings, exploring how recessions impact workers both via worse career prospects and by forcing displacements (\citealt{kahn2010}; \citealt{oreopoulos2012}; \citealt{arellano-bover2022}; \citealt{schmieder2020} ).\footnote{Our results also reproduce, as untargeted moments, earning and employment dynamics of workers following displacement events \citep{jacobson1993,schwandt2019,schmieder2020,bertheau2023}.} Our model accurately characterizes the so-called ``scarring effects'' of recession for workers, and provides an evaluation of their aggregate effects for output and productivity. %

Third, we relate to the literature analyzing the effects of recessions on productivity, debating whether recessions tend to have long-lasting negative effects on the productivity distribution (\citealt{barlevy2002}; \citealt{haltiwanger_cyclical_2025}; \citealt{crane2023}; \citealt{bajo2025}) or, by freeing up resources for more efficient uses they generate cleansing effects (\citealt{davis1996}). Our findings align with empirical evidence showing initial cleansing effects, overshadowed by a sustained loss in productivity in the medium run (\citealt{haltiwanger_cyclical_2025}). We use the insights from this empirical literature to calibrate our model and benchmark our untargeted results in the estimation, and then use the model to analyze the productivity long run dynamics after a recessionary period.
 
Our model of the labor market builds on the directed search models developed in \cite{menzio2010, menzio2016}. We adopt a standard wage setting protocol from the literature on dynamic contracts as in \cite{Thomas1988} and, more recently, \cite{balke2022}. In both cases, we extend their framework to explicitly account for two-sided heterogeneity and on-the-job human capital accumulation to study how business cycles interact with workers' decisions, firms' optimal retention policies, and influence labor market sorting.

Related work shows that business cycles reshape sorting and human-capital dynamics through job ladders and on-the-job search (\citealt{lise2017, baley2022, herkenhoff2019, gulyas2023, adda2023, carrillo2023, carrillo2023cyclical, blanco2023}). By linking human-capital accumulation to firm quality, we extend this literature by quantifying feedbacks between sorting and aggregate output. In \cite{carrillo2023} and \cite{carrillo2023cyclical}, workers build occupation-specific human capital and the cyclical pattern of mobility interacts with unemployment duration and careers; downturns slow upward moves, lowering returns to occupational upgrading and making ``sullying'' persistent. \cite{baley2022} study cyclical worker--occupation/task mismatch in a directed-search model with skill learning, where career switching is central. They show that recessions combine cleansing (mismatch falls as underqualified workers separate) with sullying (mismatch among new hires rises via overqualification). Their mechanism emphasizes occupational reallocation, while we quantify a sullying/hysteresis channel in which downturn-induced downgrading along the firm ladder reduces human-capital growth and leaves persistent scars on productivity and output. The approaches are complementary for understanding cleansing and sullying over the cycle.

\cite{jarosch2023} studies a search model where workers climb a ``slippery'' job ladder to gain job security while accumulating human capital on the job; persistence arises because jobs differ exogenously in productivity and layoff risk. In our framework, life-cycle unemployment risk is instead an endogenous by-product of sorting between workers and firms of different quality.\footnote{Given the model's treatment of pay dynamics, our research is also related to a strand of literature in labor and finance analyzing firms' management of liquidity and compensation dynamics (\citealt{Xiaolan2014}; \citealt{favilukis2020}; \citealt{acabbi2022b}; \citealt{acabbi2023}).} Ignoring our empirical findings would mischaracterize scarring effects and output and unemployment dynamics. In particular, without match-dependent human-capital accumulation, the model would miss output and unemployment volatilities, separation patterns, life-cycle experience returns, and cyclical reallocation dynamics. Our results highlight that within-spell pay dynamics---shaped by retention incentives and human capital accumulation---are crucial for matching layoffs and worker mobility.

%% file: figures/PaperFigures/Empirics/HHMS_reg.tex
\resizebox{0.8\textwidth}{!}{\begin{tabular}{l*{2}{c}}
\toprule
                    &\multicolumn{1}{c}{(1)}&\multicolumn{1}{c}{(2)}\\
                    &\multicolumn{1}{c}{Net poaching}&\multicolumn{1}{c}{Net nonemployment}\\
\midrule
$\Delta$ unempl. rate &      -1.366 &       1.530\\
                    &     (0.557)         &     (0.540)         \\
Unempl. rate (HP filtered)&      -0.492 &       0.540 \\
                    &     (0.215)         &     (0.259)         \\[1ex]
\midrule
Observations        &          92         &          92         \\
\bottomrule
\multicolumn{3}{p{0.8\linewidth}}{\footnotesize Robust standard errors in parentheses.}
\end{tabular}}

%% file: figures/PaperFigures/CalibrationTable_oldExitRates.tex
\begin{tabular}{llc}
						\toprule
						\textbf{Parameter} & \textbf{Description} & \textbf{Value} \\
						\midrule
														& \textbf{Externally Calibrated} &  \\
						\midrule
						$\nu$ & Risk aversion 	&2.000\\
						$\beta$ & Discounting 	&0.990\\
						$r_f$ & Real interest rate 	&0.011\\
						$(\varrho,\{\mu,\sigma\}_{1,2})$ & Betas' share and shapes of initial human capital dist., High-School 	&(0.57,3.47,7.86,8.89,0.11)\\
						$(\varrho_g, \{\mu_g,\sigma_g\}_{1,2})$ & Betas' share and shapes of initial human capital dist., College    	&(0.52,4.62,8.72,9.24,0.08)\\
						$(\rho_A,\sigma_A)$ & Mean and std of TFP process & (0.94, 0.009)\\
						\midrule
														& \textbf{Jointly Estimated} &  \\
						\midrule
						$\alpha$ & Production function elasticity to firm quality 	&0.411\\
						$\gamma$ & Matching function 	&1.067\\
						$\phi$ & Human capital adjustment rate, High School 	&0.154\\
						$\phi_g$ & Human capital adjustment rate, College 	&0.287\\
						$b$ & Unemployment benefit 	&0.489\\
						$\delta_b$ &  Exogenous separation prob., High School - old 	&0.028\\
						$\delta_{g,b}$ & Exogenous separation prob., College - old	&0.013\\
						$\kappa$ & Vacancy cost, elasticity 	&2.710\\
						$c_{ee}$ & Vacancy cost, scale for E-E flows 	&0.556\\
						$c_{ue}$ & Vacancy cost, scale for U-E flows	&0.703\\
						$\lambda_e$ & On-the-job-search prob. 	&0.722\\
						$\xi$ & Scaling factor in human capital accumulation 	&0.411\\
						$\xi_b$ & UB dependence on human capital 	&0.082\\
						$l$ & Linear loss of human capital while unemployed, High School 	&0.114\\
						$l_g$ & Linear loss of human capital while unemployed, College 	&0.259\\
						$\tau_{ee}$ & Human capital retention after EE 	& 0.862\\
						$\tau_{eu}$ & Human capital retention after EU 	& 0.874\\
						$x$ & Cyclical component of cost function 	& -4.337\\
						$\underbar{p}$ & Out of labor force threshold & 0.071\\
						$\sigma_{\psi}$ & Std of idiosyncratic human capital shock & 0.668\\
						$\vartheta$ & Upper bound on initial human capital & 0.980\\
						$\underbar{y}$ & Lower bound of firm quality & 2.451\\
						\bottomrule
						\end{tabular}

%% file: figures/PaperFigures/ValidationTable_BaselineLinear.txt
\begin{tabular}{lccc}
  \multicolumn{4}{l}{\textbf{(a)} {Correlation between Model and Data Series} } \\  
  \midrule 
  {} & & \multicolumn{2}{c}{Model} \\
  \cmidrule{3-4}               
   {} & & Baseline & Linear Human Capital Acc. \\
  \midrule
    Aggregate output & & 0.81 & 0.82      \\
    Unemployment     & & 0.44   & 0.01       \\
  \midrule
  \\[0.25em]  %
  \multicolumn{4}{l}{\textbf{(b)} {Correlation with aggregate output} } \\
  \midrule
  {} & Data & \multicolumn{2}{c}{Model} \\
  \cmidrule{3-4}
  {} & {} & Baseline & Linear Human Capital Acc. \\
  \midrule
    Unemployment & -0.66 & -0.56 & -0.01 \\
  \\
 \midrule
  \\[0.25em]  %
  \multicolumn{4}{l}{\textbf{(c)} {Unemployment and Output Volatility} } \\
  \midrule
 {} & Data & \multicolumn{2}{c}{Model} \\
  \cmidrule{3-4}
  {} & {} & Baseline & Linear Human Capital Acc. \\
  \midrule
    Aggregate Output SD & 2.4\% & 3.7\% & 2.5\% \\
    Unemployment SD     & 1.4\% & 0.9\% & 0.3\% \\
  \midrule
  \end{tabular}
  

%% file: figures/PaperFigures/HHMS_comparisonsSmall.tex
\begin{tabular}{l*{2}{c}}
\toprule
                    &\multicolumn{1}{c}{(1)}&\multicolumn{1}{c}{(2)}\\
                    &\multicolumn{1}{c}{Net poaching}&\multicolumn{1}{c}{Net nonemployment}\\
\midrule
\multicolumn{3}{l}{\textit{Panel A: Data}}\\[0.5ex]
Delta. unempl. rate &      -1.366 &       1.530\\
                    &     (0.557)         &     (0.540)         \\

\multicolumn{3}{l}{\textit{Panel B: Baseline Model}}\\[0.5ex]
Delta. unempl. rate &      -1.599 &       0.452\\
                    &     (0.415)         &     (0.124)         \\

\multicolumn{3}{l}{\textit{Panel C: Linear Model}}\\[0.5ex]
Delta. unempl. rate &      -0.251 &       0.035\\
                    &     (0.186)         &     (0.116)         \\

\midrule
Observations (Data)        &          92         &          92         \\
Observations (Models)        &          418         &          418         \\
\bottomrule
\multicolumn{3}{p{0.8\linewidth}}{\footnotesize Robust standard errors in parentheses.}
\end{tabular}

%% file: main.bbl
\ifdefined\printappendixreferences


\else


\fi

%% file: draft_v9_FullAppendix.tex
\input{appendix_Data}

\input{appendix_Empirics}
\input{Appendix_PropositionsAndProofs}
\input{AppendixBRE_short}

\input{appendix_SolutionAndSMM}

\input{appendix_AdditionalFig}

%% file: appendix_Data.tex
\clearpage 

\section{Data construction and sample selection}\label{sect:app_data}

We rely on two main data sources provided by INPS through the VisitINPS Program: i) data on employment relationships in Italy from 1996 to 2018 (\textit{Uniemens} dataset, \citealt{inps_data}) and ii) balance sheet data for incorporated Italian firms from 1994 to 2019 (\textit{Cerved} dataset, \citealt{cerved}).

Starting from data on virtually the universe of Italian private employment relationships at the contract level, we calculated for each worker their yearly gross real earnings and for the workers with multiple contracts we selected the information associated to the highest paid contract in the year and we find the associated annualized income using information on the number of actual weeks worked. For workers that experience a job flow or are promoted after 2009 we obtain the education level from the \textit{Comunicazioni Obbligatorie} dataset, provided to INPS by the Ministry of Labor \citep{comob}, and identify graduates and non-graduates. We label as graduates in the data both workers getting undergraduate education or above, or workers getting specialized diplomas after high-school.    

We restrict our focus on workers employed under either full-time or part-time working schemes between 16 and 65 years old. Moreover, in order to compute AKM fixed effects and maintain the estimation computationally feasible, we analyze the connected set of firms and workers across firms with more than 15 employees. As in \cite{bonhomme2019} we first cluster firms by means of a weighted K-means clustering based on deciles of their yearly wage distribution. The weight for the clustering is the yearly number of employees.

From \textit{Cerved}, we have access to firm balance sheet data. From this dataset we obtain information of value added, overall labor costs and calculate quantiles of value added per employee. All monetary values (wages, value added, total compensation) are trimmed at the 1\% level and deflated by the consumer price index for Italy. 

We use two alternative classification for firm qualities: i) quantiles based on yearly value added per worker, and ii) firm-specific fixed effects estimated following \cite{bonhomme2019} and \cite{Abowd1999a}. As in \cite{bonhomme2019}, we first cluster firms by means of a weighted K-means clustering based on deciles
of their yearly wage distribution, weighting for the yearly number of
employees, and then run a two-way fixed effect regression on workers and these firm clusters.

When looking at worker flows, our preferred specification uses time-varying productivity measures because we are aiming to capture firms' relative ranking as perceived by workers at the time of their employment decisions.  From the perspective of a job seeker, a firm's current productivity---rather than an average over windows that might include past history or future values---is the relevant signal for both expected wages and learning opportunities.

To understand the representativeness our data, it is useful to report some statistics on informality. Historically, informality hovered around $10-12\%$  until 2020 but declined to $8-10\%$ in the post-pandemic period (see \textit{Labour Force Statistics (LFS) database, ILOSTAT}).  While data linking firm productivity to informality is not available for Italy, \cite{dalla2025non} provide regional estimates according to which we can observe the reliance on informality to be larger in regions with less productive firms. In particular, informality in the North is reported to be between $6\%$ and $7.5\%$, while in the informality in the South is between $10\%$ and $13.5\%$. 

We take as recession OECD recession dates, the shadowed areas in \textbf{Figure~\ref{fig:ita_u_rate_q}}, where we plot the quarterly unemployment rate in Italy from 1998-Q1 to 2018-Q4. Aside from the anomaly represented by the 2001 episode, our identification of recessions mainly comes from the great financial crisis of 2008-2009, the sovereign debt crisis of 2011 and the recession at the end of the 90s. For this reason we benchmark our recessions experiments in the results section of the paper to the measured loss in TFP during the sovereign debt crisis.

\begin{figure}[!h]
  \caption{Unemployment rate and OECD recessions\label{fig:ita_u_rate_q}}
  \includegraphics[width=.48\linewidth]{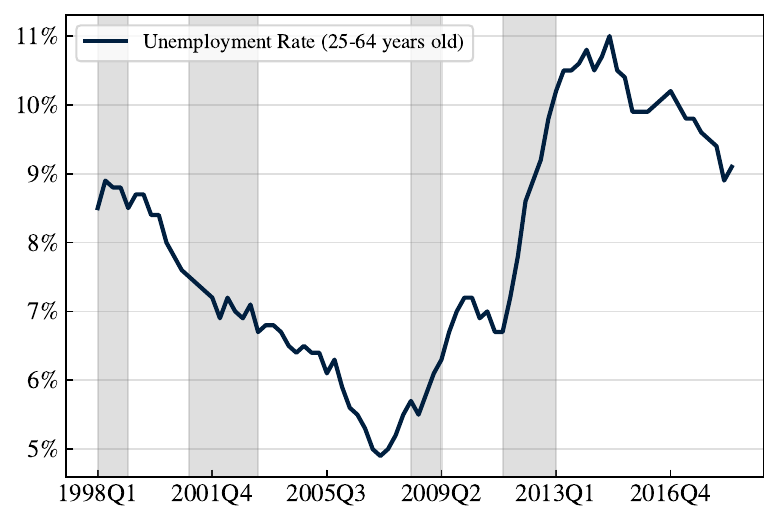}
  \floatfoot{\textbf{Note}: The figure shows the quarterly unemployment rate in Italy and the OECD recession dates (gray areas).}
\end{figure}

%% file: appendix_Empirics.tex
\section{Additional Figures and Tables, Empirics}\label{sect:empirics}

\subsection{Definition of displacement \label{sect:displ}}

We identify displacement events as a decrease in employment beyond 30\% for establishment with more than 49 employees. Displaced workers are separating workers younger than 50 years old, with at least 3 years of tenure within the firm. Displaced workers cannot be recalled, and we perform additional checks to control for false establishment closures (mergers, acquisitions, reorganizations within the same firm/group). We consider as ``false closures'' all instances when more than 20\% of employment is transferred to the same receiving firm, are excluded.

\subsection{Returns to Experience}

We discussed the role of firms for human capital accumulation in \textbf{Section \ref{sect2}}. Here, we present additional evidence that link firms to the on-the-job learning. 
What we want to test in this section is whether firms that are overall more productive are also better learning environments. To do so, we construct firm classes in two alternative ways: by quintile of firm-level value-added per employee, and by quintile of AKM fixed effects.
Then, we construct a measure of experience in each firm-class for each worker $i$. 

We can then estimate the following wage equation,
\begin{equation}\label{app_eqn:experience_profiles_abs}
\log(w_{i,t}) = \alpha_i + \alpha_{j(i,t)} + \sum_{c = 1}^5 \beta_c e^c_{i,t} + \mathbf{X}'_{i,t} \mathbf{\theta} + \varepsilon_{i,t},
\end{equation}

where $w_{i,t}$ is a monthly wage, $\alpha_i$ and $\alpha_{j(i,t)}$ are worker and firm-class fixed effects, $e_{i,t}^c$ are the number of years worker $i$ has worked in firms belonging to class $c$ up to year $t$, and $\mathbf{X}_{i,t}$ are a set of controls that include age, sector-year and contract type (temporary versus open-ended, full-time versus part-time) fixed effects. A challenge to the identification of experience effects is the presence of firm-specific skills, outside offers and bargaining, match effects, firm productivity shocks, or seniority based pay schemes. However, since these should not impact post-displacement wages, we leverage the
population-level coverage of both datasets and extend our results by presenting the same specification restricted to the sample of workers affected by firm closure and mass layoff events. 

\begin{figure}[t]
\centering 
\caption{Experience returns by firm type\label{abs_fig}}
\subcaptionbox{Full sample}{\includegraphics[width=0.5\textwidth]{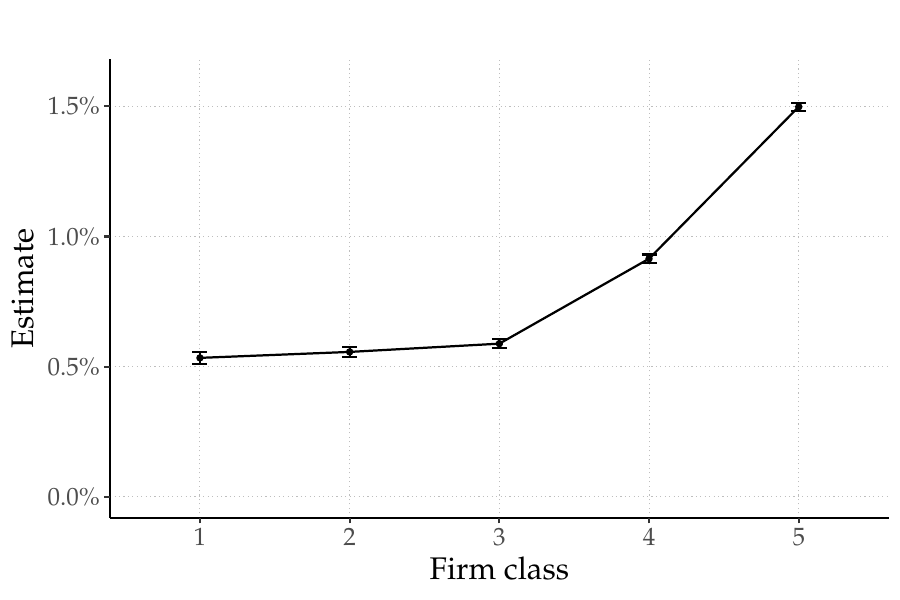}}%
\subcaptionbox{Displaced}{\includegraphics[width=0.5\textwidth]{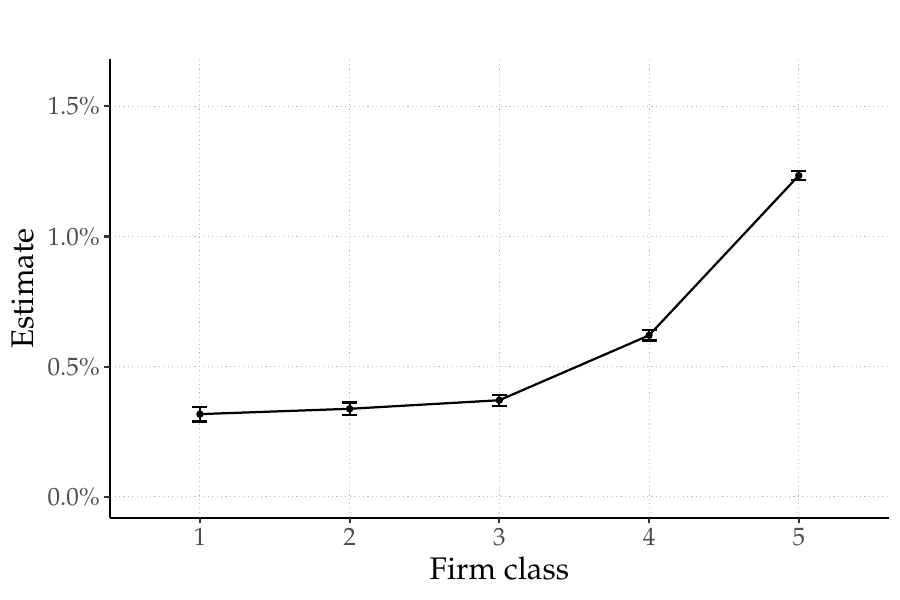}}
\floatfoot{\textbf{Note:} The figure reports the returns the returns to experience estimated from Equation~\eqref{app_eqn:experience_profiles_abs}. The full sample includes all workers employed in private firms in Italy. Firm classes are based of quintiles of value added per employee. Displaced workers are identified from mass layoffs. When estimating the returns for the sample of displaced workers we include the worker and firm-class fixed effects estimated on the full sample as regressors. Contract FEs include controls for temporary versus open-ended, full-time versus part-time, and wage clusters fixed effects. }
\end{figure}

\begin{figure}
\centering 
\caption{Experience returns by firm type, AKM FEs\label{abs_fe_fig}}
\subcaptionbox{Full sample, firm class based on AKM FE}{\includegraphics[width=0.5\textwidth]{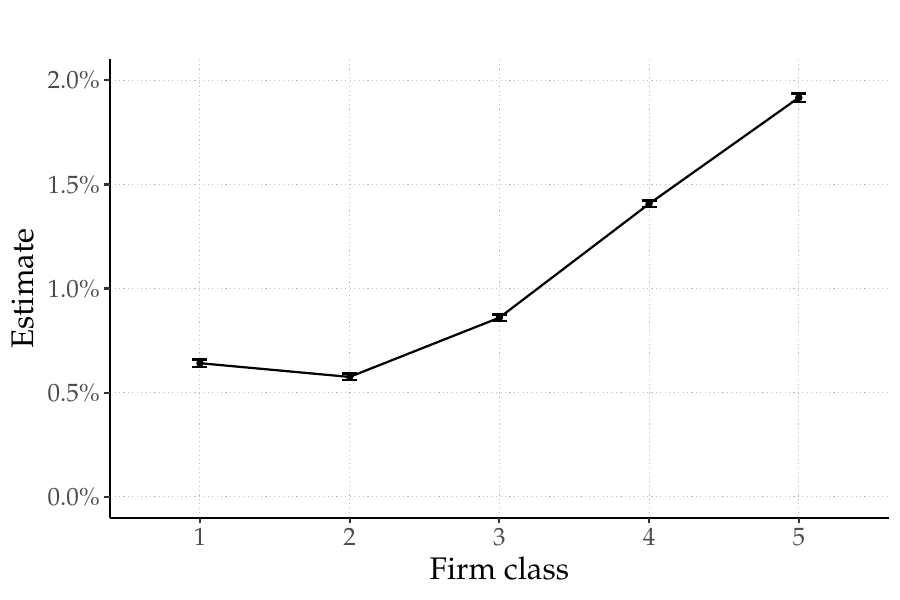}}%
\subcaptionbox{Displaced, firm class based on AKM FE}{\includegraphics[width=0.5\textwidth]{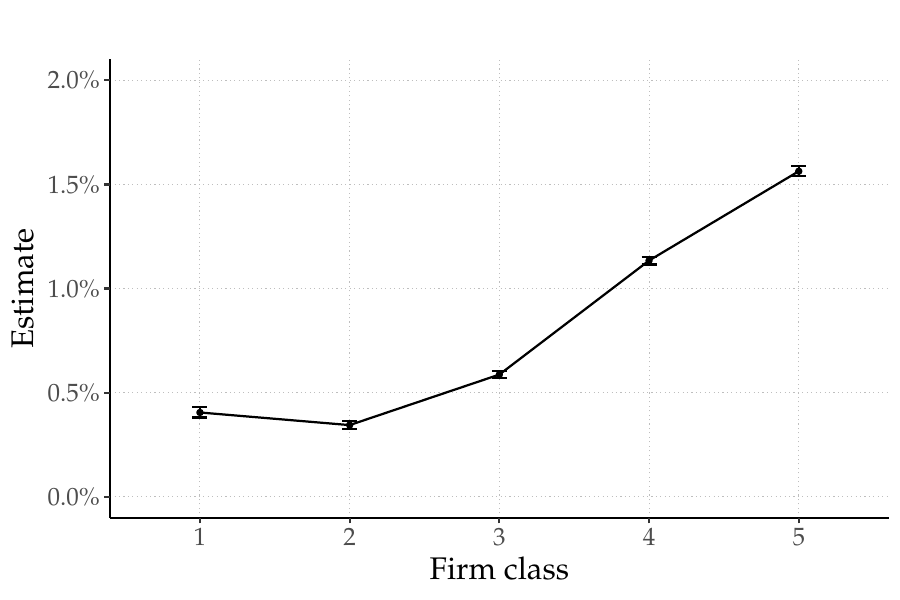}}
\floatfoot{\textbf{Note:} The figure reports the returns the returns to experience estimated from Equation~\eqref{app_eqn:experience_profiles_abs}. The full sample includes all workers employed in private firms in Italy. Firm classes are based on quintiles of wage clusters AKM FEs. Displaced workers are identified from mass layoffs. When estimating the returns for the sample of displaced workers we include the worker and firm-class fixed effects estimated on the full sample as regressors. Contract FEs include controls for temporary versus open-ended, full-time versus part-time, and wage clusters fixed effects. }
\end{figure}

\textbf{Figure \ref{abs_fig}} presents returns to heterogeneous experiences resulting from the estimation of \textbf{Equation \ref{eqn:experience_profiles}} for both the full sample and the sample of displaced workers' first post-displacement observation.

\begin{table}
\centering
\caption{Experience returns\label{tab:controlling_unemp}}
\begin{tabular}{lccc}
\toprule
& \multicolumn{2}{c}{Displaced} & Full \\
\cmidrule{2-3}
& (1) & (2) & (3) \\
\midrule 
Value Added Q1 (In) & 0.003 & 0.004 & 0.006\\
 & (0.000) & (0.000) & (0.000)\\
Value Added Q1 (Out) & 0.003 & 0.003 & 0.005\\
 & (0.000) & (0.000) & (0.000)\\
Value Added Q2 (In) & 0.004 & 0.004 & 0.006\\
 & (0.000) & (0.000) & (0.000)\\
Value Added Q2 (Out) & 0.002 & 0.002 & 0.005\\
 & (0.000) & (0.000) & (0.000)\\
Value Added Q3 (In) & 0.004 & 0.004 & 0.006\\
 & (0.000) & (0.000) & (0.000)\\
Value Added Q3 (Out) & 0.003 & 0.003 & 0.005\\
 & (0.000) & (0.000) & (0.000)\\
Value Added Q4 (In) & 0.007 & 0.007 & 0.010\\
 & (0.000) & (0.000) & (0.000)\\
Value Added Q4 (Out) & 0.006 & 0.006 & 0.009\\
 & (0.000) & (0.000) & (0.000)\\
Value Added Q5 (In) & 0.012 & 0.012 & 0.015\\
 & (0.000) & (0.000) & (0.000)\\
Value Added Q5 (Out) & 0.011 & 0.011 & 0.012\\
 & (0.000) & (0.000) & (0.000)\\
AKM FE Firm & 0.961 & 0.961 & \\
 & (0.002) & (0.002) & \\
AKM FE Worker & 1.105 & 1.105 & \\
 & (0.001) & (0.001) & \\
Unempl. duration &  & 0.000 & \\
 &  & (0.000) & \\
\midrule
\textit{Fixed effects} &  &  &  \\
Worker FE              &  &  & Y \\
Contract type          & Y & Y & Y \\
Age                    & Y & Y & Y \\
Sector $\times$ Year   & Y & Y & Y \\ 
\midrule
$R^2$ & 0.85 & 0.85 & 0.85 \\
N & 3,907,678 & 3,907,678 & 98,441,671\\
\bottomrule
\end{tabular}
\floatfoot{\textbf{Note:} The table reports the estimated experience profiles based on \textbf{Equation~\ref{eqn:experience_profiles}} and an extension in which we control for unemployment duration in the sample of the displaced. Contract FEs include controls for temporary versus open-ended, full-time versus part-time, and wage clusters fixed effects. Standard errors clustered at the worker level.}
\end{table}

We find that more productive firms increasingly contribute to the accumulation of skills, resulting in higher ex-post wages. Results are consistent across specifications, consistent with heterogeneous returns being driven by accumulation of portable skills. \textbf{Table~\ref{tab:controlling_unemp}} shows this result is robust to controlling for the duration of unemployment of displaced workers as well. Finally, since the elasticity of human capital to firm quality is different for college graduates and non-college graduates in the calibrated model, we present evidence of this feature in the data. We extend our baseline empirical specification in \textbf{Equation \ref{eqn:experience_profiles}} to allow the returns to experience to vary by firm-productivity class and education. Specifically, we estimate the experience returns for college graduates and non-college workers and recover the experience profiles across firm-productivity quintiles for each group. To do this, we run:

\begin{equation}\label{eqn:experience_profiles_educ}
\log(w_{i,t}) = \alpha_i + \alpha_{j(i,t)} + \sum_{c = 1}^5 \beta_c e^c_{i,t} + \sum_{c = 1}^5 \gamma_c e^c_{i,t} \times \mathds{1}\{ i\;\text{has College Deg.}\; \} + \mathbf{X}'_{i,t} \mathbf{\theta} + \varepsilon_{i,t},
\end{equation}

\textbf{Figure \ref{fig:educ_abs}} shows that profiles of experience returns across firm-productivity quintiles differ systematically between college and non-college workers, and the implied sensitivity of experience payoffs to firm quality is larger for college workers, in line with the calibration.

We also control for the portability of our estimated human capital gains by focusing on workers that switch sector or employer.\footnote{Absent data on occupation, using sectoral switches is the best strategy to account for occupational experience à la \cite{kambourov2009occupationalresd}.} We use 2 digits sectors for this exercise. For job stayers, our estimates can naturally capture both firm-specific and portable components of experience. For job switchers, by contrast, firm seniority in the new employer is reset to zero, so any effect of past experience that operates across productivity classes is more plausibly interpreted as reflecting portable (rather than strictly firm-specific, or sector-specific) human capital.
Concretely, we distinguish in the regression identifying (i) workers who remain with the same employer (or sector) and (ii) workers observed after a job (or sector) change. For switchers, we relate their wages in the new firm to their prior cumulative experience and to the productivity class of the current firm. We then run:

\begin{equation}\label{eqn:experience_profiles_switch}
\log(w_{i,t}) = \alpha_i + \alpha_{j(i,t)} + \sum_{c = 1}^5 \beta_c e^c_{i,t} + \sum_{c = 1}^5 \gamma_c e^c_{i,t} \times \mathds{1}\{f(i,t) \neq f(i,t-1) \} + \mathbf{X}'_{i,t} \mathbf{\theta} + \varepsilon_{i,t},
\end{equation}

As indicated by \textbf{Figure~\ref{fig:abs_switchers}}, we continue to find that workers obtain greater experience wage returns in more productive firms, even when we condition on switching employers. This pattern indicates that the higher returns to experience in high-productivity firms are not driven solely by firm-specific tenure. Instead, a substantial component of experience appears to be portable, and its payoff is higher in high-productivity firms.

\begin{figure}[t]
\centering
\caption{Returns on experience by education level}\label{fig:educ_abs}
\subcaptionbox{Full sample}{%
  \includegraphics[width=.45\linewidth]{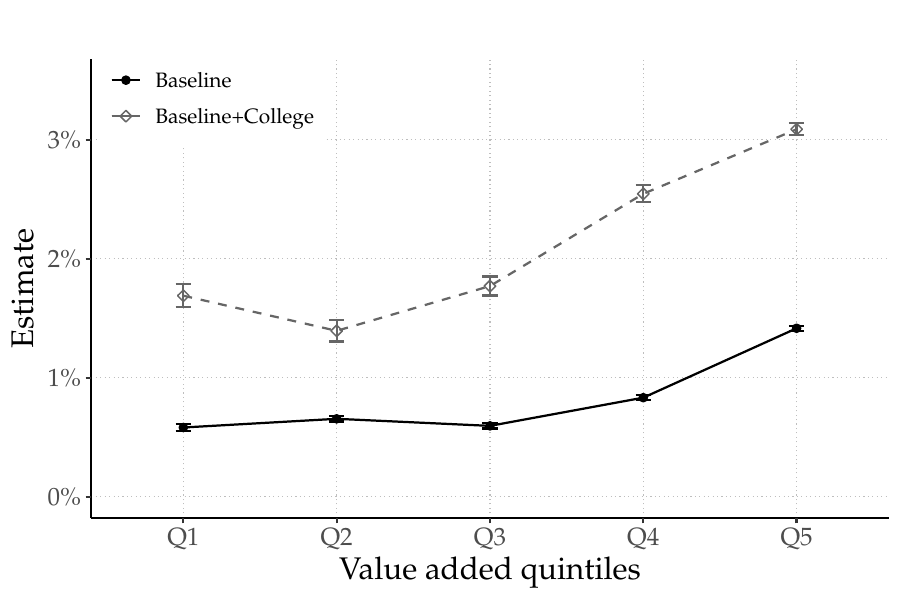}}%
\hfill
\subcaptionbox{Displaced}{%
  \includegraphics[width=.45\linewidth]{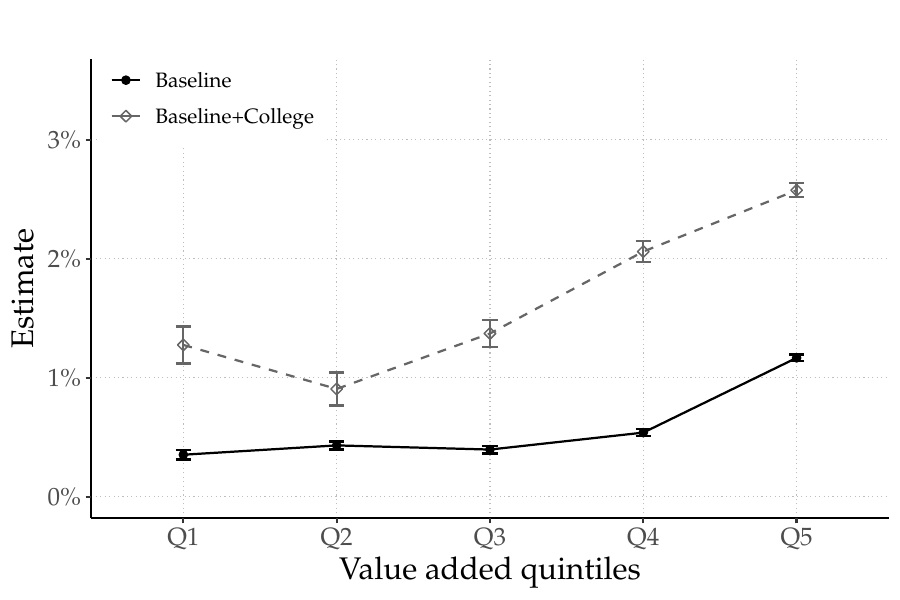}}%
\floatfoot{\textbf{Note}: The figure reports the returns to experience in firms across the productivity distribution for college and non-college graduates following Equation~\eqref{eqn:experience_profiles_educ}. Panel (a) shows the estimates on the full sample while Panel (b) reports them for the sample of displaced workers.}
\end{figure}

\begin{figure}[t]
\caption{Experience returns by sectoral and firm switchers\label{fig:abs_switchers}}
\includegraphics[width=\textwidth]{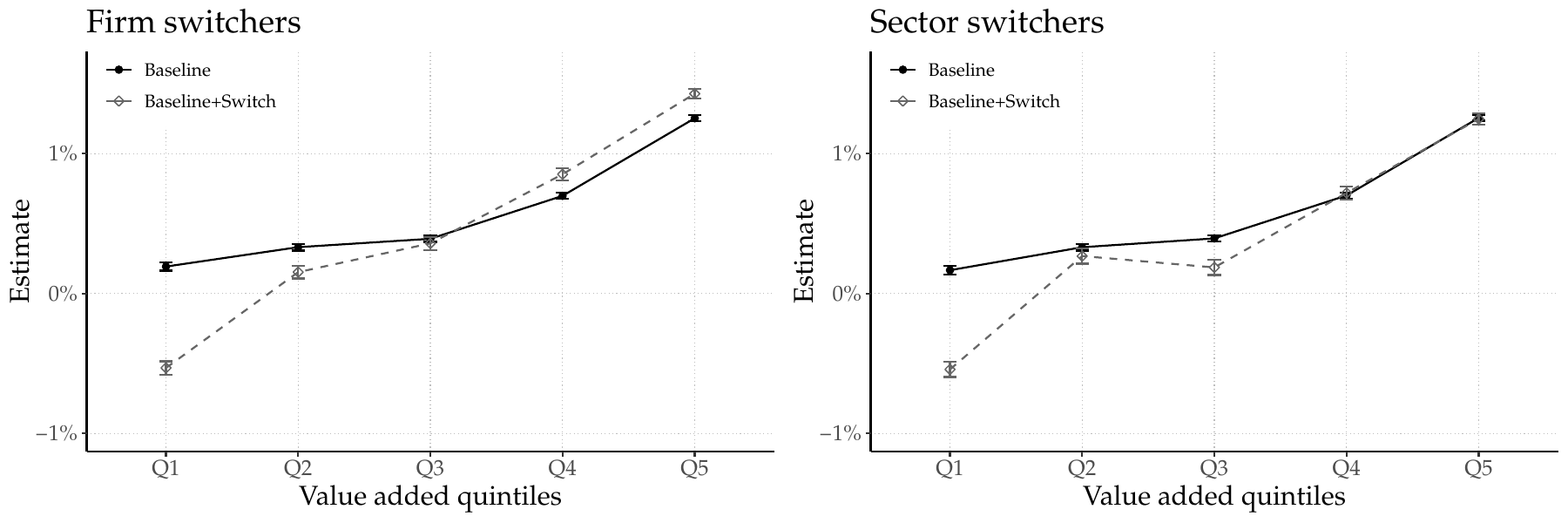}
\floatfoot{\textbf{Note}: The figure reports the estimates of the returns to experience across the firm productivity distribution for workers that switch jobs (\textit{Firm switchers}) and workers that switch sectors (\textit{Sector switchers}). For switchers, the presented effects are already the sum of the baseline and interaction coefficients.}
\end{figure}

Another implication of our modelling approach is that the accumulation of skills at a certain firm cannot go on indefinitely: as tenure in a firm of a certain productivity increases, the skill gains progressively decrease. A simple way to test this hypothesis is to introduce a quadratic experience term, to verify that the class-specific experience profile is concave, that is running:

\begin{equation}\label{eqn:experience_profiles_abs_sq}
\log(w_{i,t}) = \alpha_i + \alpha_{j(i,t)} + \sum_{c = 1}^5 \beta_c e^c_{i,t} + \sum_{c = 1}^5 \gamma_c \left( e^c_{i,t} \right)^2 + \mathbf{X}'_{i,t} \mathbf{\theta} + \varepsilon_{i,t},
\end{equation}

Results for \textbf{Equation \ref{eqn:experience_profiles_abs_sq}} are presented in \textbf{Table \ref{abs_sq_table}}. The squared term is consistently negative or very close to 0, while the linear term remains positive and increasing in firm class. This confirms the initial hypothesis of decreasing returns to experience from spending time in a given firm class, and is consistent with the model intuition that, for workers to increase earnings over the life cycle, they need to climb a `firm ladder' between jobs of different quality.

\begin{table}
\begin{center} 
\caption{Quadratic wage experience profiles by firm type}
\label{abs_sq_table}
\begin{tabular}{lcc|cc}
\toprule
 & Full & Displaced & Full & Displaced \\
 & (1) & (2) & (3) & (4) \\
\midrule
Exp. in quantile 1      & 0.006 & 0.004 & 0.006 & 0.004 \\
                        & (0.000) & (0.000) & (0.000) & (0.000) \\
Exp$^2$ in quantile 1   & -3e-5 & -9e-5 & 5e-5 & 2e-5 \\
                        & (1e-5) & (2e-5) & (<1e-5) & (1e-5) \\
Exp. in quantile 2      & 0.005 & 0.001 & 0.004 & 0.001 \\
                        & (0.000) & (0.000) & (0.000) & (0.000) \\
Exp$^2$ in quantile 2   & 7e-5 & 0.0002 & 0.0001 & 0.0001 \\
                        & (1e-5) & (1e-5) & (<1e-5) & (1e-5) \\
Exp. in quantile 3      & 0.007 & 0.003 & 0.009 & 0.005 \\
                        & (0.000) & (0.000) & (0.000) & (0.000) \\
Exp$^2$ in quantile 3   & -5e-5 & 9e-5 & -1e-5 & 7e-5 \\
                        & (<1e-5) & (1e-5) & (<1e-5) & (1e-5) \\
Exp. in quantile 4      & 0.010 & 0.007 & 0.020 & 0.015 \\
                        & (0.000) & (0.000) & (0.000) & (0.000) \\
Exp$^2$ in quantile 4   & -0.0001 & -9e-5 & -0.0003 & -0.0002 \\
                        & (<1e-5) & (1e-5) & (<1e-5) & (1e-5) \\
Exp. in quantile 5      & 0.019 & 0.016 & 0.035 & 0.027 \\
                        & (0.000) & (0.000) & (0.000) & (0.000) \\
Exp$^2$ in quantile 5   & -0.0003 & -0.0002 & -0.0009 & -0.0006 \\
                        & (<1e-5) & (1e-5) & (<1e-5) & (1e-5) \\
 &  &  &  &  \\
N & 98,447,028 & 3,907,678 & 98,447,028 & 3,907,678 \\
$R^2$ & 0.86 & 0.85 & 0.86 & 0.85 \\
\midrule
Worker FEs & \checkmark & \checkmark & \checkmark & \checkmark \\
Contract FEs & \checkmark & \checkmark & \checkmark & \checkmark \\
Sector-Year FEs & \checkmark & \checkmark & \checkmark & \checkmark \\
Quant. type & VA & VA & AKM FE & AKM FE \\
\bottomrule
\end{tabular}
\end{center}
\floatfoot{\textbf{Note:} The table reports the returns the returns to experience estimated from \textbf{Equation~\eqref{app_eqn:experience_profiles_abs}}. The full sample includes all workers employed in private firms in Italy. Displaced workers are identified from mass layoffs. When estimating the returns for the sample of displaced workers we include the worker and firm-class fixed effects estimated on the full sample as regressors. Contract FEs include controls for temporary versus open-ended, full-time versus part-time, and wage clusters fixed effects. }
\end{table}

Finally, because human capital is inherently unobservable and we cannot directly observe learning on the job, it is reassuring to document sources of wage growth that are less naturally attributed to alternative mechanisms---such as learning about one's own skills or occupational sorting. To this end, we show that within-spell wage growth is systematically related to the relative position of workers and firms in the productivity distribution, as shown in \textbf{Figure \ref{fig:japan}}. Conditional on worker and firm fixed effects, sector-by-year controls, experience, and contract characteristics, workers employed at firms that are higher-ranked relative to their own position experience faster wage growth, consistent with a gradual ``catch-up'' process within matches. This feature is central to our contribution: regardless of the primitive source of learning, firms act as differential learning environments, and access to higher-quality firms accelerates within-spell wage growth.

\begin{figure}[!h]
\centering 
\caption{Within spell wage growth and differences between firm and worker qualities\label{fig:japan}}
\subcaptionbox{Wage growth and rank differences \label{fig:japan_rank}}{\includegraphics[width=0.5\textwidth]{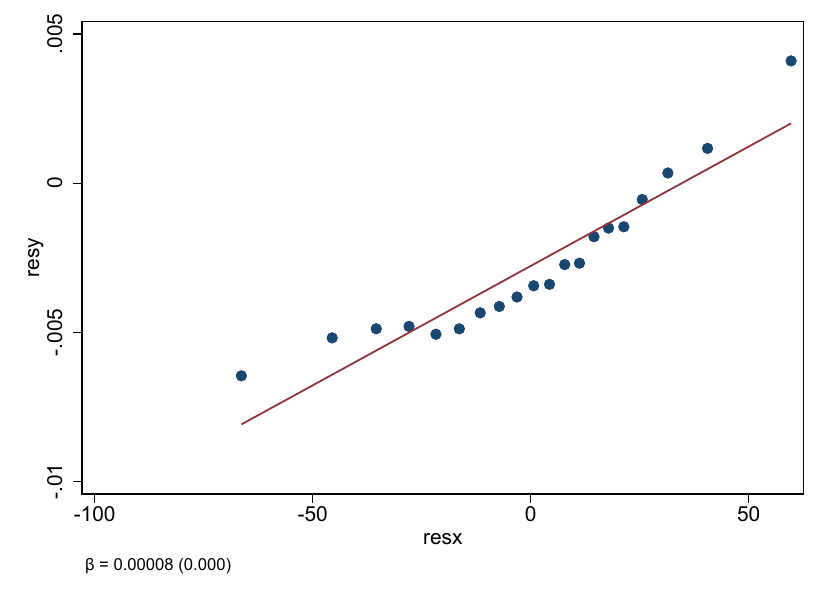}}%
\subcaptionbox{Wage growth and AKM FE differences\label{fig:japan_fes}}{\includegraphics[width=0.5\textwidth]{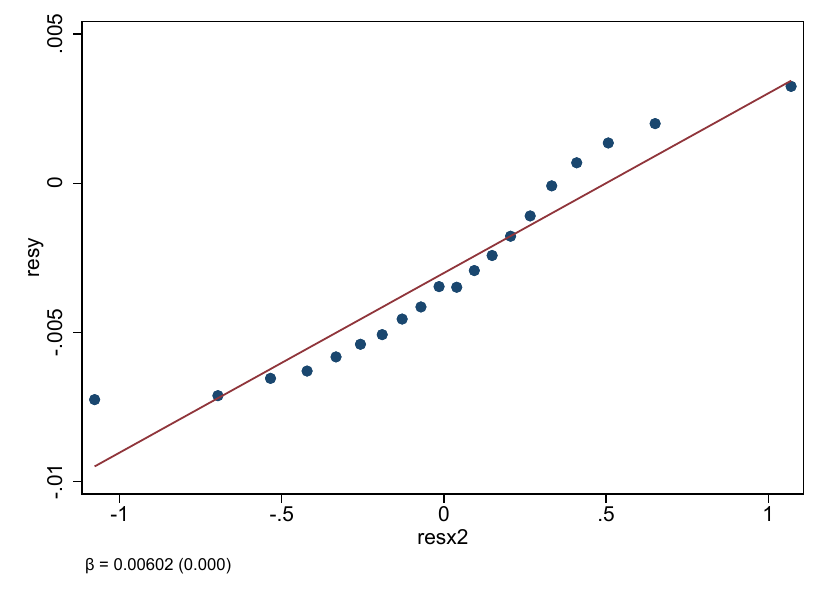}}
\floatfoot{\textbf{Note}: The figure plots the residual average wage growth within an employment spell (on the y-axis) against the difference between firm and worker qualities at the beginning of the spell (on the x-axis). Panel (a) uses the difference between the firm's rank in the value added per employee distribution and the rank of the worker in the distribution of estimated AKM fixed effects as a measure of qualities. Panel (b) instead uses the differences in the estimated AKM firm and worker effects. Both quality measures and the wage growth are residualized by sector-year, sex, local labor market, contract type, qualification and education level fixed effects. Each bin represents 5\% of the underlying data, the fitted line is based on the underlying data.}
\end{figure}

An alternative explanation for finding increasing returns to experience in higher firm classes is that better firms are also better at screening out bad workers (also once the job starts). If this were the case, then the amount spent in each firm class would be a function of match quality, which would be difficult to control for in our specification.

\textbf{Figure~\ref{fig:seprates}} below plots the average separation rates--unconditional (panel a) and conditional on match characteristics (panel b)--along the worker and firm quality distributions. The figure allows us to see directly whether high-productivity firms exhibit systematically higher separation rates for low-quality workers with respect to other workers, and are more selective than correspondingly lower productivity firms---what we would expect if more productive firms mainly operated through very strong ex-post screening of bad matches. We do not find evidence of such a pattern: separation rates do not spike for low-quality workers in high-productivity firms; if anything, separation rates tend to be lower and relatively similar in more productive firms across all worker types. Thus, the raw data do not suggest that short employment spells in high-productivity firms are concentrated among low-quality workers. If anything, lower productivity firms seem to be more selective in screening out bad matches, as relative separation rates between low- and high-productivity workers are greater in these firms.

This insight is largely confirmed once we control for worker characteristics. In the right panel of \textbf{Figure \ref{fig:seprates}} we report normalized separation rates after estimating a linear probability model of separation and controlling for contract type, qualifications, age, sex, and sector-year fixed effects. It becomes apparent that higher-productivity firms are slightly more selective in separating the lowest-quality workers relative to higher-quality ones. This greater selectivity, however, is mostly absent for other worker groups. The ratio of separation rates for lowest- to highest-quality workers is very similar across low- and high-productivity firms, and is \textit{lower} in high-productivity firms for all workers except the lowest-quality group. Overall, the most productive firms do not appear to shed low-quality workers relative to high-quality workers at markedly different rates than low-productivity firms.

\begin{figure}[!h]
\centering
\caption{Separation rates \label{fig:seprates}}
\subcaptionbox{Unconditional}{%
  \includegraphics[width=.45\linewidth]{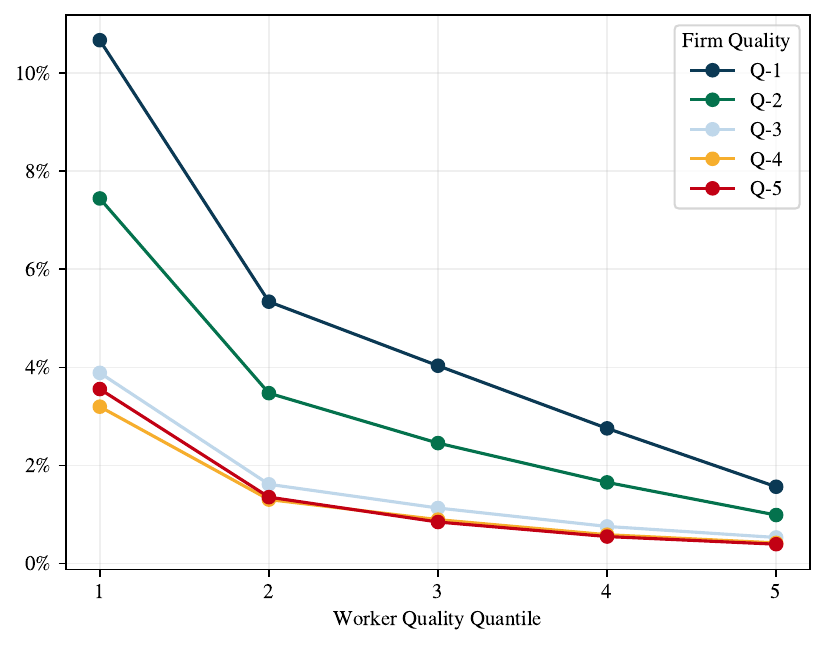}}%
\hfill
\subcaptionbox{Conditional}{%
  \includegraphics[width=.45\linewidth]{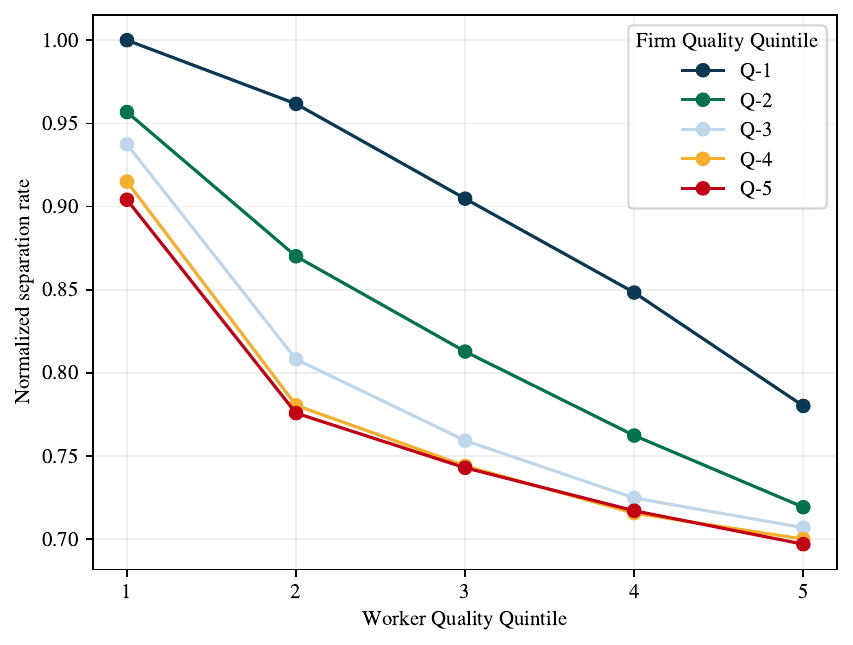}}%
\floatfoot{\textbf{Note}: The figure shows the average separations by firm and worker qualities in our sample. Panel (a) reports the unconditional average of employment--unemployment transition in our sample, by firm quality, for each type of worker. Panel (b) instead reports the conditional separation probabilities (normalized for Worker Q1 and firm Q1), residualizing for contract type, qualifications, age, sex and sector-year fixed effects the probability of separation in a linear probability model. In this panel all probabilities are relative to the separation rates of the lowest quality workers in the lowest productivity firms.}
\end{figure}

\subsection{Persistent effects of firm quality on workers' earnings}

In this section we report the coefficients for the specification in \textbf{Equation \ref{eq:eue_reg}} and \textbf{Figure \ref{fig:EUE_displ}} in \textbf{Table \ref{tab:EUE_empirical}}. In \textbf{Figure \ref{fig:EUE_empirical_akm}} and \textbf{Table \ref{tab:eue_table_fe}} we report an analogous robustness specification, in which the firm class is a quintile of estimated AKM fixed effects on firm cluster types as in \cite{bonhomme2015}. We also provide additional evidence that within-spell wage growth is systematically related to the relative position of workers and firms in the productivity distribution in \textbf{Figure~\ref{fig:japan}}.

\begin{figure}[t]
\centering
\caption{Persistent effect of past firms for workers\label{fig:EUE_empirical_akm}}
\subcaptionbox{Displaced workers \label{fig:EUE_displ_fe}}{\includegraphics[width=0.5\textwidth]{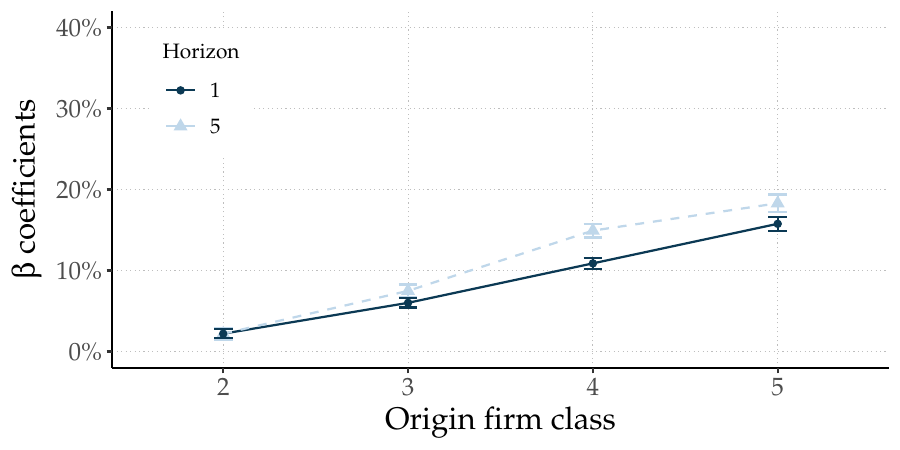}}%
\subcaptionbox{Short unemployment spells\label{fig:EUE_full_fe}}{\includegraphics[width=0.5\textwidth]{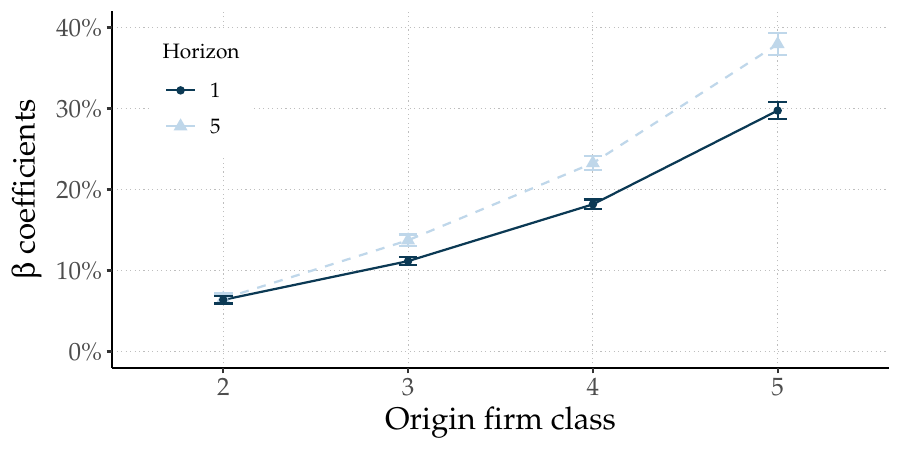}}
\vspace{-1em}
\floatfoot{\textbf{Note:} The table figure reports the main coefficient of interest for the regression in \textbf{Equation~\ref{eq:eue_reg}}. The displaced workers sample includes workers that make E-U-E transitions following mass layoffs, while the short unemployment spell sample identifies all E-U-E transitions in which workers are employed for at least 4 quarters after an unemployment spell of less than one year. Firm quality is defined as (yearly) quintiles of firm cluster AKM fixed effect.}
\end{figure}

\begin{table}
\caption{Importance of firm quality for workers' earnings after E-U-E transitions}\label{tab:EUE_empirical}
\begin{tabular}{lcccc}
\toprule
            & \multicolumn{2}{c}{Displaced workers} & \multicolumn{2}{c}{Short unemployment spells}   \\
\cmidrule(lr){2-3} \cmidrule(lr){4-5}
        & Log-Wage$_{t}$ & Log-Wage$_{t+5}$  & Log-Wage$_{t}$ & Log-Wage$_{t+5}$  \\
        & (1) & (2)  & (3) & (4)  \\
\midrule
Firm Prod. Q2   & 0.015 & 0.011 & 0.020 & 0.008\\
                & (0.002) & (0.003) & (0.002) & (0.003)\\
Firm Prod. Q3   & 0.032 & 0.028 & 0.048 & 0.040\\
                & (0.002) & (0.003) & (0.002) & (0.003)\\
Firm Prod. Q4   & 0.043 & 0.053 & 0.078 & 0.073\\
                & (0.002) & (0.003) & (0.002) & (0.003)\\
Firm Prod. Q5   & 0.075 & 0.101 & 0.124 & 0.134\\
                & (0.002) & (0.003) & (0.002) & (0.003)\\
\midrule
Controls & Yes & Yes & Yes & Yes \\
$R^2$  & 0.75 & 0.69 & 0.53 & 0.46 \\
N  & 394,821 & 312,387 & 417,958 & 211,043 \\
\bottomrule
\end{tabular}
\floatfoot{\textbf{Note}: Standard errors clustered at the worker level in parenthesis. The table reports the main coefficient of interest for the regression in \textbf{Equation~\ref{eq:eue_reg}}. The displaced workers sample includes workers that make E-U-E transitions following mass layoffs, while the short unemployment spell sample identifies all E-U-E transitions in which workers are employed for at least 4 quarters after an unemployment spell of less than one year. Firm quality is defined as (yearly) quintiles of firm value added per employee.}
\end{table}

\begin{table}
\centering
\caption{E-U-E wage regressions by past firm AKM firm cluster FE}
\label{tab:eue_table_fe}
\begin{tabular}{lcccc}
\toprule
 & \multicolumn{2}{c}{Displaced} & \multicolumn{2}{c}{Short unemployment spells} \\
 & Log-wage$_{t+1}$ & Log-wage$_{t+5}$ & Log-wage$_{t+1}$ & Log-wage$_{t+5}$ \\
 & (1) & (2) & (3) & (4) \\
\midrule
AKM FE Q2       & 0.022 & 0.022 & 0.064 & 0.065\\
                & (0.002) & (0.003) & (0.002) & (0.002)\\
AKM FE Q3       & 0.060 & 0.075 & 0.112 & 0.137\\
                & (0.002) & (0.003) & (0.002) & (0.003)\\
AKM FE Q4       & 0.109 & 0.149 & 0.182 & 0.233\\
                & (0.003) & (0.003) & (0.002) & (0.003)\\
AKM FE Q5       & 0.158 & 0.183 & 0.298 & 0.380\\
                & (0.003) & (0.004) & (0.004) & (0.005)\\
 \midrule
Controls & Yes & Yes & Yes & Yes \\
$R^2$ & 0.75 & 0.69 & 0.54 & 0.47\\
N & 394,821 & 312,387 & 417,961 & 211,044\\
\bottomrule
\end{tabular}
\floatfoot{\textbf{Note}: Standard errors clustered at the worker level in parenthesis. The table reports the main coefficient of interest for the regression in \textbf{Equation~\ref{eq:eue_reg}}. The displaced workers sample includes workers that make E-U-E transitions following mass layoffs, while the short unemployment spell sample identifies all E-U-E transitions in which workers are employed for at least 4 quarters after an unemployment spell of less than one year. Firm quality is defined as (yearly) quintiles of firm cluster AKM fixed effects.}
\end{table}

\paragraph{The role of workers' beliefs.} Clearly, the mere existence of a relationship between past wage and unemployment duration would not be at odds with our main assumptions, unless it is information \textit{on the past employer} that drives longer unemployment spells, suggesting a change of behavior that is independent of inherent worker characteristics. A first way to reassure the reader that this is not the case is to look jointly at how earnings and employer characteristics pre-separation determine unemployment duration. To this end, we run:

\begin{equation}\label{eq:coh_wage}
\text{Unemp.\;Duration}_{i} = \alpha\log{w_i} + \sum_{j=1}^5\beta_j\mathbb{I}\{\text{Origin Employer}_i \in Q_j\} + \mathbf{X}'_i \boldsymbol{\theta} \ + \varepsilon_i
\end{equation}

\textbf{Table \ref{tab:EUE_hor_wage}} shows that, while higher wage \textit{per se} increases unemployment duration, workers coming from very productive firms are the ones who are re-employed faster. Of course, both the firm characteristic and the past wage at this stage contain information on the characteristic of the worker, on their past tenure, and their beliefs, so we cannot directly interpret the results of this exercise as denying the possibility of workers' beliefs as driving these results. However, this could be concerning if we saw evidence of past job characteristics being overall associated with longer unemployment spells, which we do not find. What we do find, on the other hand, is that \textit{controlling for their past wage}, workers coming from the most productive firms find a new job \textit{faster} and are paid \textit{more}.

To further strengthen our results on this front against concerns of selection into and out of unemployment, we separate the results on E-U-E transitions according to the length of the unemployment spells. This amounts to estimating \textbf{Equation 2} in the main text \textit{separately} by cohort of re-entry in the labor market, that is unemployment duration.

Results are remarkably similar to the specifications in which cohorts are pooled and effects are measured at different horizons. Importantly, given the stability of the coefficients we can be confident that the direction of the association between the quality of the firm of origin and future wages after an unemployment spell is not driven by duration dependence effects on the pool of workers that belong to each cohort. However, we cannot take a stance on whether the effects should be monotonic also \textit{between} each cohort, as we cannot exclude that there might be other strong compositional effects due to many competing factors originating from workers' or employers' characteristics. While interesting, we believe that  modeling and measuring how these potential countervailing effects stack up against each other is beyond the scope of the present paper, but we leave it for future research.

\begin{table}[h!]
\centering
\caption{Previous wage and firm quality as predictors of unemployment duration in EUE transitions}
\label{tab:EUE_hor_wage}
\input{figures/PaperFigures/Empirics/EUE_hor_wage}
\floatfoot{\textbf{Note}: The table reports estimated coefficients based on Equation~\eqref{eq:coh_wage} on a sample of E-U-E transitions. We include sector-year, sex, contract type and qualification fixed effect.}
\end{table}

\begin{figure}[!h]
\caption{EUE by cohorts of unemployment duration}
\includegraphics[width=0.55\textwidth]{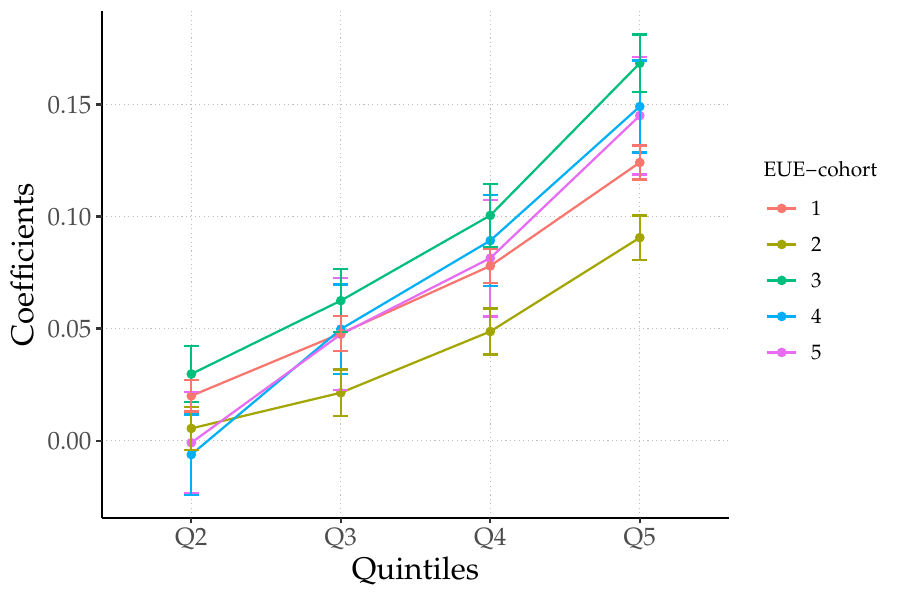}
\floatfoot{\textbf{Note}: The figure plots the associations between past firm quality, measured as quintiles of value added per employee, and wages after an Employment-Unemployment-Employment transition. In all specifications we control for the last wage pre E-U-E transition. We report the coefficients for progressive cohorts of workers, implicitly accounting for the unemployment duration in the E-U-E transition. Specifically, Cohort~1 is based on workers that have one year of non-employment before finding a new job, similarly Cohort~5 is based on workers that have five years of non-employment before finding a new job.}
\end{figure}

\subsection{Scarring Effects}

In order to estimate the effect of entering the labor market in a recession we use an age-cohort-period model in which we break the collinearity among the three set of fixed effects by proxying the cohort fixed effects with the cyclical component of real GDP (Hamilton filtered).
In particular, we estimate the following yearly model:
\begin{equation}
\log(w)_{t,c,e} = \Phi_t  +\Phi_e + \beta_e \widetilde{Y}_{c} \times \Phi_e + u_{t,c,e}, 
\end{equation}
where $\Phi_t, \Phi_e$, are dummies for calendar years and labor market experience, and $\widetilde{Y}_c$ is the cyclical realization of real GDP for cohort $c$ at time of their labor market entry. The set of coefficients $\beta_e$, therefore, estimate the effect of aggregate conditions on real wages at each year of labor market experience. 

\begin{table}
    \centering
    \caption{Effects of initial aggregate conditions along the experience profile and experience growth profile}
    \label{tab:reg_scarring}
\resizebox{0.7\textwidth}{!}{
\begin{tabular}{l cc}
\toprule
Dep.Variable: Log-Wage & Experience $\times$ Cycle & Experience \\
\midrule
Experience Dummy & & \\
0	&1.968 &	\\
	&(0.452) &	\\
1	&1.463 &	0.150 \\
	&(0.238)&	(0.011) \\
2	&1.473 &	0.246 \\
	&(0.283) &	(0.014) \\
3	&1.239 &	0.304 \\
	&(0.343) &	(0.014) \\
4	&1.250 &	0.342 \\
	&(0.349) &	(0.015) \\
5	&1.239 &	0.375 \\
	&(0.349) &	(0.015) \\
6	&1.301 &	0.399 \\
	&(0.287) &	(0.015) \\
\midrule
$R^2$	&0.89&	0.89 \\
N 	&254,000,000&	254,000,000\\
\midrule
Age FE & \checkmark & \checkmark \\
Year FE &\checkmark & \checkmark \\
Sex FE & \checkmark & \checkmark \\
LLM FE	& \checkmark & \checkmark \\
\bottomrule
\end{tabular}}
\floatfoot{\textbf{Note:} The table reports regression coefficients for the empirical estimates in the data used to construct the profiles in Figure \ref{fig:scarring_model_data}.}   
\end{table}

\subsection{Job ladder cyclicality}\label{app:job_ladder_cyc}
We document the extent of cleansing vs. sullying forces in our data. We first rank firms in our sample by productivity within sector using {\color{black} value added per employee} in rolling windows of 5 years. Then separate firms into {\color{black} 5 quintiles}, and define as high productivity the top 2 quintiles, and as low productivity the bottom 2. Consider \citeauthor{haltiwanger2022}'s  \citeyear{haltiwanger2022} measure of net job flows rates, where $H$ is the total number of hires, and $S$ is the total number of separations for firms of productivity class $j$ in time period $t$:
\begin{equation*}
    \text{NJF}_{j,t} = \sum_{i\in\left\lbrace p, n \right\rbrace}  \frac{H^{i,j}_t - S^{i,j}_t}{ E^{j}_{t-1}} \quad j = l, h
\end{equation*}
$i = \left\lbrace p, n \right\rbrace$ indicates the type of flow: poaching from other firms or from non-employment (for $H$), separation towards other firms or into non-employment (for $S$). This measure can then be used to see how workers reallocate across the productivity ladder. In particular:

\begin{equation*}
    \lambda^j_t  = \frac{H^{p,j}- S^{p,j}}{ E^{j}_{t-1}} \text{ ;}  \quad \quad   \delta^j_t  = \frac{H^{n,j}- S^{n,j}}{ E^{j}_{t-1}}  
\end{equation*}

The series $\lambda$ and $\delta$ represent net poaching and hiring from unemployment for firms of class $j$; two useful measures of cleansing vs sullying are then $\Lambda_t = \lambda^{h}_t - \lambda^{l}_t $ and $\Delta_t = \delta^{h}_t - \delta^{l}_t $. 

We can look at the cyclical properties of $\Lambda_t$ and $\Delta_t$. To do so, we use both the unemployment rate in first differences and the cyclical component of unemployment (filtering with HP filter). Additional results on this are shown in \textbf{Table~\ref{app_tab:differentialnetjob}}.

%% file: figures/PaperFigures/Empirics/EUE_hor_wage.tex
\begin{tabular}{lc}
\toprule
Variable & Unemployment Duration \\
\midrule
Log-Wage & 0.012\\
 & (0.003)\\
Value added quintile 2 & -0.092\\
 & (0.004)\\
Value added quintile 3 & -0.135\\
 & (0.004)\\
Value added quintile 4 & -0.177\\
 & (0.004)\\
Value added quintile 5 & -0.159\\
 & (0.004)\\
\midrule
Experience controls & Y\\
\midrule
Observations & 1,146,681 \\
R-squared & 0.077 \\
\bottomrule
\multicolumn{2}{p{0.8\linewidth}}{\footnotesize Standard errors in parentheses.}
\end{tabular}

%% file: Appendix_PropositionsAndProofs.tex
\section{Discussion and Proofs \label{sect:app_discussion}}
In this section we discuss the properties of the equilibrium of the model economy developed in the previous sections. All propositions and corresponding proofs related to the properties discussed in \textbf{Section~\ref{sect:discussion_body}} are reported in \textbf{Appendix~\ref{Appendix:wrk_problem}} and \textbf{\ref{Appendix:optimal_contract}}.

\subsection{Workers optimal behavior\label{Appendix:wrk_problem}}
For compactness of notation, we omit the dependence on education level, $\iota$, which is a fixed characteristic, and the idiosyncratic human capital shock, which is additive, from the proofs in Appendices. The logic of the proofs follows without loss of generality.

The following propositions characterize the properties of workers' optimal search strategies that solve the search problem in \eqref{eqn:wrk_option}, restated here for convenience:
\begin{equation} \label{app_eqn:wrk_option}
R(h,\tau,V;\Omega) = \underset{v}{\sup}\;\Big[ p(\theta(h,\tau,v;\Omega))\big[v-V]\big]\Big].
\end{equation}
\begin{lem}\label{app_cor:p_properties}
The composite function $p(\theta(h,\tau,v;\Omega))$ is strictly decreasing and strictly concave in $v$.
\end{lem}
\begin{proof}
For this proof we follow closely \cite{menzio2010}, Lemma 4.1 (ii). From the properties of the matching function we know that $p(\theta)$ is increasing and concave in $\theta$, while $q(\theta)$ is decreasing and convex.
Consider that the equilibrium definition of $\theta(\cdot)$ is
\begin{equation*}
\theta(h,\tau,v; \Omega) = q^{-1} \left( \frac{c(y)}{ J(h,\tau,v,y;\Omega)} \right), \label{app_eqn:eq_tight}
\end{equation*}
and that the first order condition for the wage and the envelope condition on $V$ of the optimal contract problem in \eqref{eqn:firm_prob} implies
$$ \frac{\partial J(h,\tau,v,y;\Omega)}{\partial v} = -\frac{1}{u'(w)}.$$
so that as $u'(\cdot) > 0$, $J(\cdot)$ is decreasing in $v$.

From the equilibrium definition of $\theta(\cdot)$ and noting that $q^{-1}(\cdot)$ is also decreasing due to the properties of the matching function we have that
\[ \frac{\partial \theta(h,\tau,v; \Omega)}{\partial v} = \left.\frac{\partial q^{-1}(\xi)}{\partial \xi}\right|_{\xi = \frac{c(y)}{ J(h,\tau,v,y;\Omega)}} \cdot \left(-\frac{\partial J(h,\tau,v,y;\Omega)}{\partial v}\right) \cdot \frac{c(y)}{ (J(h,\tau,v,y;\Omega))^2} < 0,\]
which, in turn, implies that
\[ \frac{\partial p(\theta(h,\tau,v; \Omega))}{\partial v} = \left.\frac{\partial p(\theta)}{\partial \theta} \right|_{\theta = \theta(h,\tau,v; \Omega)} \cdot \frac{\partial \theta(h,\tau,v; \Omega)}{\partial v} < 0.\]

Suppressing dependence on the states $(h,\tau,y,\Omega)$ for readability, to prove that $p(\theta(v))$ is concave, consider that $J(v)$ is concave\footnote{$J(\cdot)$ concave give the two-point lottery in the structure of the contract. See \cite{menzio2010} Lemma F.1.} and a generic function $\frac{c}{v}$ is strictly convex in $v$. This implies that with $\alpha \in [0,1]$ and $v_1, v_2 \in \mathcal{V}$, $v_1 \neq v_2$:
$$\frac{c}{J(\alpha v_1 + (1-\alpha)v_2)} \leq \frac{c}{\alpha J(v_1) + (1-\alpha)J(v_2)} < \alpha \frac{c}{J(v_1)} + (1-\alpha) \frac{c}{J(v_2)}.$$
As  $p(q^{-1}(\cdot))$ is strictly decreasing the inequality implies that
$$
\begin{array}{ll}
p\left(q^{-1}\left(\frac{c}{J(\alpha v_1 + (1-\alpha)v_2)}\right)\right) &\geq p\left(q^{-1}\left(\frac{c}{\alpha J(v_1) + (1-\alpha)J(v_2)}\right)\right) \\
 &> \alpha p\left(q^{-1}\left(\frac{c}{J(v_1)}\right)\right) + (1-\alpha) p\left(q^{-1}\left(\frac{c}{J(v_2)}\right)\right),
\end{array}
$$
and as $\theta(v) = q^{-1}(\frac{c}{J(v)})$:
$$ p(\theta(\alpha v_1+(1-\alpha)v_2)) > \alpha p(\theta(v_1)) + (1-\alpha)p(\theta(v_2))$$
so that $p(\theta(v))$ is strictly concave in $v$.
\end{proof}

In the following proposition we summarize the main results regarding the behavior of the workers and their objective functions.

\begin{prop}\label{main_prop:wrk_char}
Given the worker search problem, the following properties hold:
\begin{enumerate}[label=(\roman*)]

\item The returns to search, $p(\theta(h,\tau,v;\Omega))\big[v-V \big]$, are strictly concave with respect to promised utility, $v$.

\item The optimal search strategy
$$v^*(h,\tau,V;\Omega) \in \arg \max_{v} \left\{ p(\theta(h,\tau,v;\Omega))\big[v-
V\big] \right\}$$
is weakly increasing in $V$, Liptschitz continuous and unique.

\item For all promised utilities, the search gain $R(h,\tau,V;\Omega)$ is positive, weakly decreasing in $V$.

\item The survival probability of the match, given the optimal choice of the worker, is increasing in the value of current and future promised utilities, so $\widetilde{p}_t(h,\tau,v;\Omega)$ is increasing in $v$ and $V$.

\item Search strategies are increasing in workers' human capital, $h$.

\end{enumerate}
\end{prop}
\begin{proof}
The proofs follow closely \cite{shi2009}, Lemma 3.1 and \cite{menzio2010}, Corollary 4.4.
More formally, for each triplet $(h,\tau,\Omega)$ given at each search stage, we can re-define the search objective function as $K(v,V) = p(\theta(v))(v - V)$ and $v^*(V) \in \arg \max_{v} K(v,V)$ as the function that maximises the search returns (i.e. the optimal search strategy of the worker) and prove the following
\begin{enumerate}[label=(\roman*)]
\item To show that $K(v,V)$ is strictly concave in $v$ consider two values for $v$, $v_1,v_2 \in \mathcal{V}$ such that $v_2 > v_1$ and define $v_\alpha = \alpha v_1 + (1-\alpha)v_2$ for $\alpha \in [0,1]$.

Then by definition:
\begin{align*}
    K(v_\alpha,V) &= p(\theta(v_{\alpha}))(v_\alpha - V) \\
    &\geq [\alpha p(\theta(v_1)) + (1-\alpha)p(\theta(v_2))][\alpha (v_1 - V) + (1-\alpha) (v_2 - V)] \\
    & = \alpha K(v_1,V) + (1-\alpha)K(v_2,V) + \alpha (1-\alpha) [(p(\theta(v_1)) - p(\theta(v_2))](v_2 - v_1)  \\
    &> \alpha K(v_1,V) + (1-\alpha)K(v_2,V)
\end{align*}
where the first inequality follows from the concavity of $p(\theta(\cdot))$ (this is true if $J(\cdot)$ concave with respect to $V$) and the second inequality stems from the fact that $p(\theta(\cdot))$ is strictly decreasing hence  $\alpha (1-\alpha) [(p(\theta(v_1)) - p(\theta(v_2))](v_2 - v_1) >0$.

\item  \textbf{Weakly Increasing.} Consider a worker employed in a job that gives lifetime utility $V$. Given that $v \in \left[\underb{v},\overline{v}\right]$, and that submarkets are going to open depending on realizations of the aggregate productivity, $a$, there is only one region in the set of promised utilities where the search gain is positive. This set is $[V,v(a)]$ with $v(a)$ being the highest possible offer that a firm makes in the submarket for the worker $(h,\tau)$. Any submarket that promises higher than $v(a)$ is going to have zero tightness. Therefore, the optimal search strategy for $V \geq v(a)$ is $v^*(V) = V$, as $K(V,v(a))=K(V,V)=K(\overline{v},V)=0$ (the search gain is null given the current lifetime utility $V$).  For $V \in [V, v(a)]$, instead, as $K(v,V)$ is bounded and continuous, the solution $v^*(V)$ has to be interior and therefore respect the following first order condition
\begin{equation} \label{app_eqn:search_foc}
V = v^*(V) + \frac{p(\theta(v^*(V)))}{p'(\theta(v^*(V)) \cdot \theta'(v^*(V))}.
\end{equation}

Now consider two arbitrary values $V_1$ and $V_2$, $V_1 < V_2 < \overline{v}$ and their associated solutions $W_i = v^*(V_i)$ for $i = 1,2$. Then, $V_1$ and $V_2$ have to generate two different values for the right-hand side of \eqref{app_eqn:search_foc}. Hence, $v^*(V_1) \cap v^*(V_2) = \emptyset $ when $V_1 \neq V_2$.

This also implies that the search gain evaluated at the optimal search strategy is higher than the gain at any other arbitrary strategy so that $K(W_i,V_i) > K(W_j,V_i)$ for $i \neq j$. This implies that
\begin{align*} 0 \geq &[K(W_2,V_1) - K(W_1,V_1)] + [K(W_1,V_2) - K(W_2,V_2)] \\
&=(p(\theta(W_2)) - p(\theta(W_1)))(V_2 -
V_1),
\end{align*}
thus, $p(\theta(W_2)) \leq p(\theta(W_1))$. As $p(\theta(\cdot))$ is strictly decreasing (see Lemma~\ref{app_cor:p_properties}), then $v^*(V_1) \leq v^*(V_2)$.

\textbf{Lipschitz continuous.} We follow \cite{menzio2010} and show that, with generic points $V_1 \leq V_2$, $v^*(V_2) - v^*(V_1) \leq V_2 - V_1$. Given that in this case we know that $v^*(V_1) \leq v^*(V_2)$, then $v^*(V_2) - v^*(V_1) \geq 0$. If $v^*(V_2) - v^*(V_1) = 0$ then the claim is trivially satisfied. If $v^*(V_2) - v^*(V_1) > 0$, then let's consider a generic real number $\delta \in (0,(v^*(V_2) - v^*(V_1))/2)$. Given the definition of $v^*(\dot)$, we know that $K(v^*(V_i),V_i) \geq K(v^*(V_i)+\delta,V_i)$. From this inequality and the definition for $K(\cdot,\cdot)$ it follows that 
$p(\theta(v^*(V_1)))(v^*(V_1) - V_1) \geq p(\theta(v^*(V_1)+\delta))(v^*(V_1)+\delta - V_1)$ and similarly $p(\theta(v^*(V_2)))(v^*(V_2) - V_2) \geq p(\theta(v^*(V_2)-\delta))(v^*(V_2)\delta - V_2)$. Then
\begin{align*}
v^*(V_1) - V_1 &\geq \frac{p(\theta(v^*(V_1)+\delta))\delta}{p(\theta(v^*(V_1))) - p(\theta(v^*(V_1)+\delta))} \\
v^*(V_2) - V_2 &\leq \frac{p(\theta(v^*(V_2)-\delta))\delta}{p(\theta(v^*(V_2)-\delta)) - p(\theta(v^*(V_2)))}.
\end{align*}

Knowing that $p(\theta(\cdot))$ is a decreasing function, and that $v^*(V_1) + \delta \leq v^*(V_2) - \delta$ given the domain of $\delta$, then $p(\theta(v^*(V_1)+\delta) \geq p(\theta(v^*(V_2)-\delta)$. Also, given that $v^*(V_2) - v^*(V_1) > 0$ then $p(\theta(v^*(V_1))) - p(\theta(v^*(V_1)+\delta)) \leq p(\theta(v^*(V_2)-\delta)) - p(\theta(v^*(V_2)))$. This implies that we can rewrite the inequalities above as 
\begin{align*}
v^*(V_2) - V_2 &\leq \frac{p(\theta(v^*(V_2)-\delta))\delta}{p(\theta(v^*(V_2)-\delta)) - p(\theta(v^*(V_2)))} \\
 &\leq \frac{p(\theta(v^*(V_1)+\delta)}{{p(\theta(v^*(V_2)-\delta)) - p(\theta(v^*(V_2)))}} \\
 &\leq \frac{p(\theta(v^*(V_1)+\delta)}{p(\theta(v^*(V_1))) - p(\theta(v^*(V_1)+\delta))} \leq v^*(V_1) - V_1
\end{align*}
which implies $v^*(V_2) - v^*(V_1) \leq V_2 - V_1$.

\textbf{Unique.} Uniqueness follows directly from strict concavity shown in (i).

\item The Bellman equation for the search problem is:
$$
R(h,\tau,V;\Omega) = \underset{v}{\sup}\;\Big[ p(\theta(h,\tau,v;\Omega))\big[v-V\big]\Big]
$$
hence a simple envelope argument shows that
$$ \frac{\partial R(h,\tau,V;\Omega)}{\partial V}= - p(\theta(h,\tau,v;\Omega)) \leq 0,  $$
as the job finding probability is weakly positive for all utility promises.

As $p(\theta(\cdot)) \geq 0$, $v^*(\cdot) \in \left[\underb{v},\overline{v}\right]$ then $R(\cdot) \geq 0 $.

\item Given the optimal search strategy, $v^{*}(h,\tau,V;\Omega)$, we can define the survival probability of the match as in \textbf{Equation~\ref{eqn:p_ret}}:
$$
\widetilde{p}(h,\tau,V;\Omega) \equiv(1-\delta)(1-\lambda_{e} p(\theta(h,\tau,v^*;\Omega))).
$$
Then, given $(h,\tau,\Omega)$
$$
\frac{\partial \widetilde{p}(V) }{\partial V} = - (1 - \delta) \lambda_e \left. \frac{\partial p(\theta) }{\partial \theta } \right|_{\theta = \theta(v^*)} \left. \frac{\partial \theta(v)}{\partial v}\right|_{v = v^*(V)} \frac{\partial v^*(V) }{\partial V } > 0,
$$
because $p(\cdot)$ and $v^*(\cdot)$ are both increasing functions while $\theta(\cdot)$ is a decreasing function in promised utilities.

\item We want to show that the optimal search strategy $v^*$, is increasing in human capital, so that $\frac{\partial v^*(h)}{\partial h} > 0$. 

Knowing that the optimal search strategy has to satisfy the following first order condition: 
$$ \frac{\partial p }{\partial v} (v^* - V) = 0,$$
we can use the implicit function theorem to get 
\begin{equation}\label{app_eqn:search_in_h}
\frac{\partial v^*}{\partial h} = - \frac{(v^* - V)\frac{\partial^2 p}{\partial v \partial h} + \frac{\partial p}{\partial h}}{(v^* - V)\frac{\partial^2p}{\partial v^2} + 2 \frac{\partial p}{\partial v}}.
\end{equation}

As $(v^* - V)>0$, and $p(\cdot)$ decreasing and concave in $W$ so that $\frac{\partial^2p}{\partial W^2},\frac{\partial p}{\partial W}<0$, the denominator of  \textbf{Equation~\ref{app_eqn:search_in_h}} is negative. Therefore for the RHS \textbf{Equation~\ref{app_eqn:search_in_h}} to be positive, it has to be that $(v^* - V)\frac{\partial^2 p}{\partial W \partial h} + \frac{\partial p}{\partial h} > 0$. This, in turn, implies that 
$$ \frac{\partial^2 p}{\partial W \partial h} > \frac{\partial p}{\partial W}\frac{\partial p}{\partial h}\frac{1}{p},$$
where the left hand side is negative as $p(\cdot)$ is increasing in $h$ but decreasing in $W$.

\end{enumerate}
\end{proof}

The first statement implies that the marginal returns of searching towards better firms are decreasing. The intuition is that as workers search for work at firms granting better values, their job-finding probability decreases as better employment prospects are also subject to higher competition.

As a consequence of the strict concavity established in the first statement, workers' optimal search strategy is unique. The search strategy is also (weakly) increasing in the value of lifetime utilities granted by the current contract, which is the outside option for the worker.

The third statement follows from the fact that marginal returns to search are decreasing and the set of feasible utility promises is compact. The intuition is that employees at firms with higher utility promises have a relatively fewer chances of improving their position. Given a high outside option, the utility gain from moving is relatively lower, whereas the probability of matching with any firm does not depend on the \textit{current} utility promise per se, but on the future promise offered by the vacancy.

The fourth statement finally follows from considering the implication of the previous ones. Given that the optimal search strategy is increasing in $V$ workers' probability of leaving the firm at any time ends up depending negatively on $V$. This guarantees a longer expected duration of the match at higher current promised utility $V$, thus retention probabilities that are increasing in promised utilities $v$.

As human capital accumulation is tightly linked to the quality of the employer, workers that are able to start their working careers in good times have a greater chance of finding themselves on an higher path of human capital growth. As worker careers are limited and human capital accumulation follows a slow-moving process, business cycle effects on human capital quality fade only slowly and the quality of initial matches, both in terms of lifetime utility and firm quality, bears a long-standing effect on workers' careers.

\subsection{Characteristics of the optimal contract\label{Appendix:optimal_contract}}

The optimization in the contracting problem balances a trade-off between insurance provision and profit maximization for firms. The contract implicitly takes into account workers' search incentives and their inability to commit to stay. The following propositions characterizes workers' incentives along the business cycle from the firms' standpoint.

\begin{lem} \label{lem:jconcw}
The Pareto frontier $J(h,\tau,W,y;\Omega)$ is concave in $W$.
\end{lem}
\begin{proof}
This is a direct consequence of using a two-point lottery for $\{w_i, W_{i}^\prime\}$ as shown by \cite{menzio2010}, Lemma F.1.
\end{proof}

\begin{lem} \label{lem:y_incr}
The Pareto frontier $J(h,\tau,W,y;\Omega)$ is increasing in $y$.
\end{lem}

\begin{proof}
The intuition for this proof follows the fact that a higher $y$ firm, once the match exists, can always deliver a certain promise $V$ and have resources left over. Within a dynamic contract, future retention is already optimized as the match is formed. This means that the promise $V$ can be delivered by the greater capacity on the part of producing with respect to a close $y$ firm. In presence of human capital accumulation, the worker is compensated through greater option values in the future, which again means that, even with lower retention, the firm cashes in more profits while decreasing wages (and respecting the $V$ promise).%

One can get to the same conclusion by starting from time $T$, noticing that the function $J$ is trivially increasing in $y$ in the last period, and the stepping back. At $T-1$, given $V$, any higher $y$ function can make greater profits with the same delivery of value $V$, given the contract's optimal promise, which is a fortiori true with human capital accumulation (the option value is greater, so the firm can decrease $w$ as a response).
A more formal argument goes as follows: start from the optimal policies, as per Eq. \eqref{eqn:firm_prob} of a firm that has $y = \bar{y}$ and assume that one could exogenously increase its installed capital to $\bar{y} + \varepsilon$. We want to know whether, keeping policies constant, this would increase the flow of profits while keeping the worker indifferent. If this is true, then \textit{a fortiori} it will be true that the firm value function $J$ will be increasing in $y$. With a slight abuse of notation and for conciseness, we refer to the future $J$ in Eq. \eqref{eqn:firm_prob} as $J^\prime$.

This amounts to calculating:\footnote{Without loss of generality, we assume that $J^{\prime}$ is constant with respect to $y$. Alternatively, one may start proving the result for contracts offered to workers just one period before retirement (for which $\frac{\partial J^{\prime}}{\partial y}$ is trivially positive), and generalize the result with backward induction.}
$$
   \left.\frac{d \bar{J}}{d y} \right|_{W,\{\pi_i, w_i,W_{i}^\prime\}} = \frac{\partial f(\cdot)}{\partial y} + \beta \mathbb{E}_\Omega \left[ \left.\frac{\partial \tilde{p}(\cdot)}{\partial y}\right|_{W,\{\pi_i, w_i,W_{i}^\prime\}}  \bar{J}^{\prime} \right]
$$

This first order condition presents the trade-off discussed in words above, namely that an increase in $y$ will be instantaneously beneficial to production, but might also potentially have a longer term adverse effect on profits through decreased retention. Finding the sign of the derivative on the LHS hinges on understanding the sign of the derivative of the second element on the RHS, since $\frac{\partial f(\cdot)}{\partial y} > 0$ given the properties of the production function $f(\cdot)$. The change in $y$ would affect search objective $v_y$ through the variation in $h$ due to the human capital accumulation dynamics, even taking the current firm policies as given.

Notice that:

\begin{equation*}
   \frac{\partial \tilde{p}(\cdot)}{\partial y} = - \delta  \lambda_E \frac{\partial p(\cdot)}{\partial \theta} \frac{\partial \theta(\cdot)}{\partial y}
\end{equation*}

Remember that $\frac{\partial p(\cdot)}{\partial \theta}>0$ due to the properties of the job-finding probability function $p(\cdot)$.
Assume first the case in which $\frac{\partial \theta(\cdot)}{\partial y} \leq 0$. In this case $\frac{\partial \tilde{p}(\cdot)}{\partial y}>0$, which in turn implies, as argued, that $J$ has to be increasing in $y$.

Now consider the second case, namely $\frac{\partial \theta(\cdot)}{\partial y} > 0$. By the free entry condition, we obtain:

\begin{equation*}
    \frac{\partial \theta(\cdot)}{\partial y} =  \frac{\partial q^{-1}(c(y)/J(y))}{\partial y}= \frac{1}{q^{\prime} \left( q^{-1} \left( c(y)/J(y)  \right) \right) } \frac{c^{\prime}(y) J(y) - \frac{\partial J^(y)}{\partial y} c(y)}{J(y)^2}
\end{equation*}

The first term in the result is negative, given the properties of function $q(\cdot)$. Given our assumption on $\frac{\partial \theta(\cdot)}{\partial y}$  the second term in the result has to be negative as well. This requires: $c^{\prime}(y) J(y) - \frac{\partial J(y)}{\partial y} c(y) <0 $, or $\frac{\partial J(y)}{\partial y} > \frac{c^{\prime}(y) J(y)}{c(y)} > 0$.
\end{proof}

\begin{prop}\label{main_prop:PF_incrina}
The Pareto frontier $J(h,\tau,W,y;a,\mu)$ is increasing in the aggregate productivity shock $a$, while retention probabilities, $\widetilde{p}(h,\tau,W;a,\mu)$ decrease in aggregate productivity.
\end{prop}
\begin{proof}
We proceed by backward induction.\footnote{For compactness of notation, we omit without loss of generality the two-point lottery in the equations in the proof.} Following the logic of the proof of \textbf{Lemma \ref{lem:y_incr}}, the proposition is trivially true for workers $T$ periods old. Given that the firm increases its production while keeping the worker at least indifferent, $J$ is at least weakly increasing in $a$. However, the firm can also feasibly  increase the worker's wage by $\varepsilon$, with $ \varepsilon < \frac{\partial f(\cdot)}{\partial a}$. $J$ is thus strictly increasing in $y$. 

Consider now a worker who is $T-1$ periods old. A firm matched to a worker in submarket $\{ h, T -1,y,W \}$ will face the following Pareto frontier
\begin{align*}
    J(h, T - 1 ,y,W_{y};& a,\mu)=  \underset{w,W^\prime}{\sup}  \Big( f(h,y;a) - w \nonumber \\
    + & \mathbb{E}_\Omega\left[\widetilde{p}(h^\prime,T,W^\prime;a^\prime,\mu^\prime)(  f(h^\prime,y;a^\prime)- w^\prime )\right] \Big) \\
\end{align*}

Analogously to the proof in \textbf{Lemma \ref{lem:y_incr}}, assume that aggregate productivity increases from $\bar{a}$ to $\bar{a}+\varepsilon$. Assume that the firm keeps its policies constant once again. We aim at proving that, even in such a case, firm value increases while keeping the worker at least indifferent. If this is the case, it is \textit{a fortiori} true that $J$ increases in $a$ after reoptimizing firms' policies. 

We are now interested in the sign of:\footnote{We assume that $J^\prime = f(h^\prime,y;a^\prime)- w^\prime$ is constant with respect to $a$. It is possible to prove, by backward induction, that this assumption is without loss of generality for the sake of the proof.}
$$
    \left.\frac{\partial \bar{J}}{\partial a} \right|_{W,\pi, w,\{W^\prime\}} = \frac{\partial f(\cdot)}{\partial a} + \beta \mathbb{E}_\Omega \left[ \left.\frac{\partial \tilde{p}(\cdot)}{\partial a}\right|_{W,\pi_i, w,\{W^\prime\}}  \bar{J}^{\prime} \right]
$$

Now notice that, in equilibrium,

\begin{equation*}
\frac{\partial \widetilde{p}(\theta)}{\partial a} \propto -\frac{\partial p(\theta)}{\partial a}  \text{,} \quad
    \text{ with} \quad \frac{\partial p(\theta)}{\partial a} = \underbrace{\frac{\partial p(\theta)}{\partial \theta}}_{>0} \cdot \underbrace{\frac{\partial \theta}{\partial J(\cdot)}}_{>0} \cdot \frac{\partial J(\cdot)}{\partial a}
\end{equation*}
\noindent where the sign of the second derivative on the right hand side comes from the free entry condition and the properties of vacancy filling probability function $q(\cdot)$. Given this, it has to be that $ \frac{\partial p(\theta)}{\partial a}$ and $\frac{\partial J(\cdot)}{\partial a}$ have the same sign in equilibrium. Now, if both are strictly positive, both statements of our proposition are immediately true. Let's now assume they are both negative or zero. If this is the case, then $\frac{\partial \tilde{p}(\cdot)}{\partial a} \geq 0$. But this implies $\frac{\partial \bar{J}}{\partial a}>0$, which is a contradiction.
\end{proof}

The intuition behind this proposition relies on the observation that higher productivity realization are associated not only with better outcomes on impact but also to better future prospects, given that the productivity process is an increasing Markov chain.

A key property of the model is that it allows to characterize the workers' optimal behaviour along the business cycle. The following proposition summarizes how the search strategy changes depending on the aggregate productivity realization.

\textbf{Proposition \ref{main_prop:PF_incrina} and Corollary \ref{main_prop:searchs_incrina}} have an important implication regarding firms' vacancy posting  and workers' search decisions. The fact that at the posting stage profits $J$ are increasing in aggregate productivity implies that more entry will take place in good times, and ceteris paribus more entrepreneurs will open up vacancies across the whole firms' distribution.\footnote{In our model a better firm is a more productive firm. We do not specifically model the determinant of quality heterogeneity but we take the existence of profound differences in firm quality as a fact \citep{arellano2020a,arellano2022}.} The resulting higher tightness impacts workers' optimal search behaviour as the job finding probability increases in all submarkets. As a consequence, workers respond optimally to the productivity increase searching in submarkets that guarantee higher lifetime utility promises.

Firms utility promises depend on the structure of the optimal contract. The contract provides insurance to workers through wage paths that are downward rigid, and at the same time allows firms to profit as wages only partially adjust to productivity realizations.

The following propositions provide a clear picture of the growth path prescribed by the optimal contract for a continuing firm. First, let us define the productivity threshold that determines whether a worker-firm match does not survive.

\begin{cor}
There exists a productivity threshold $a^*(h,\tau,W,y)$ below which firms will not continue the operate.
\end{cor}
\begin{proof}
The proof follows immediately from \textbf{Proposition~\ref{main_prop:PF_incrina}} and the timing of the shock. Given the timing of the shock, exit is fully determined by the current productivity shock and incumbent firms know in advance whether they are willing to produce in the next period.

Therefore, as the Pareto frontier is strictly increasing in $a$, firms are willing to continue the contract if $ \mathbb{E}_{\Omega}[J(h^\prime,\tau+1,W^\prime,y;a^\prime,\mu^\prime)|h,\tau,W,y,a,\mu]\geq 0 $, so that the threshold that determines exit is
$$ a^*(h,\tau,W,y) : \mathbb{E}_{\Omega}[J(h^\prime,\tau+1,W^\prime,y;a^\prime,\mu^\prime)|h,\tau,W,y,a,\mu] = 0.$$
\end{proof}
The intuition of why this has to be the case is linked to the fact that the Pareto frontier is strictly increasing in $a$ and decreasing in the level of promised utilities to the worker. Hence once the aggregate state realizes a firm is able to perfectly predict whether next period it will exit the market or stay in (given the timing, the decision is based on expected profits, and is thus \textit{not} state-contingent to next period's productivity). The choice is taken \textit{before} new realizations of productivity, so it is possible that a firm makes negative profits for at most one period.

\begin{cor}
The productivity threshold $a^*(h,\tau,y,W)$ below which firm $y$ in match with worker $(h,\tau)$ and promised utility $W$ exits the market in the aggregate state $\Omega$ is increasing in $y$.
\end{cor}

\begin{proof}
Consider two firms characterized by $y_1,y_2$ with $y_1<y_2$. Consider the threshold for firm $y_1$, $a_1^* = a^*(h,\tau,,W_{y_1},y_1)$. Firm $y_1$ makes 0 profits if state $a_1^*$ materializes next period. Consider firm $y_2$ trying to mimic the current contract offered by $y_1$ to $(h,\tau)$. We know that $J$ is increasing in $y$ from \textbf{Lemma \ref{lem:y_incr}}, which implies that the firm is making a profit at $a_1^*$. This completes the proof.
\end{proof}

\begin{lem} \label{app_lem:jconcy}
The Pareto frontier $J(h,\tau,W,y;\Omega)$ is strictly concave in $y$.
\end{lem}

\begin{proof}

\textbf{No human capital accumulation.} The proof without human capital accumulation is straightforward. Assuming there is no dependency of $h^\prime$ on $y$, one can write:
$$
\frac{d J}{d y} = f_y + \widetilde{p}\frac{d J^\prime}{d y}
$$
$$
\frac{d^2 J}{d y^2} = f_{yy} + \widetilde{p}\frac{d^2 J^\prime}{d y^2}
$$
where we take $J^\prime$ to represent next period's firm value $J(\cdot)$, and the dependence of all controls on $y$ is ignored by virtue of the envelope condition. One can readily observe by induction that, given the concavity of $f(\cdot)$ in $y$, $\frac{d^2 J}{d y^2}<0$. 

\textbf{Human capital accumulation.} With human capital accumulation the concavity of the function $J$ on $y$ depends on further assumptions on the concavity of the human capital accumulation function $g(\cdot)$. As these assumptions involve equilibrium objects, we state them here and verify ex-post numerically that they are always respected in our setting. 
$$
\frac{d J}{d y} = f_y + g_y \frac{\partial \widetilde{p}J^\prime}{\partial h} + \widetilde{p}\frac{d J^\prime}{d y}
$$
$$
\frac{d^2 J}{d y^2} = f_{yy} + \underbrace{g_{yy}}_{<0} \frac{\partial \widetilde{p}J^\prime}{\partial h^\prime} + g_y^2 \frac{\partial^2 \widetilde{p}J^\prime}{\partial h^{\prime2}} + \underbrace{2g_y}_{>0}\left( \frac{\partial \widetilde{p}}{\partial h^\prime}f_y + \widetilde{p}f_{hy}  \right) + \widetilde{p}\frac{d^2 J^\prime}{d y^2}
$$

Sufficient conditions for the previous expression to be true at $T-1$ would jointly be:

\begin{equation}
    \tilde{p}_h (f - w) + f_h \tilde{p} \geq 0 
\end{equation}
\begin{equation}
    \tilde{p}_h f_y + \tilde{p}f_{h,y} \leq 0
\end{equation}
\begin{equation}
     \tilde{p}_{h,h} (f - w) + \tilde{p} f_{h,h} + 2 \tilde{p}_h f_h \leq 0 
\end{equation}

Rearranging and combining the first two would lead to:

\begin{equation}
   (f - w)  \leq   \frac{ f_y f_h }{ f_{h,y}} \Rightarrow  f   \leq   \frac{ f_y f_h }{ f_{h,y}} + w
\end{equation}

Since $w>0$, a more stringent condition would be:

\begin{equation}
 \frac{f f_{h,y}}{ f_y f_h }  \leq 1 \rightarrow ES \leq 1 
 \end{equation}

This condition is not generally too restrictive, and holds with equality in the Cobb-Douglas case.

Going back to the third condition, and combining it with the first one, one gets

\begin{equation}
   \frac{\tilde{p}}{\tilde{p}_h}\left( f_{h,h} - \frac{\tilde{p}_{h,h}}{\tilde{p}_h} f_h  \right) +  2  f_h \geq 0 
\end{equation}

A sufficient condition for that would be

\begin{equation}
    h\frac{f_{h,h}}{f_h} \leq h\frac{\tilde{p}_{h,h}}{\tilde{p}_h}  
\end{equation}

In the Cobb-Douglas case this condition amounts to

\begin{equation}
    \alpha \geq -h\frac{\tilde{p}_{h,h}}{\tilde{p}_h}  
\end{equation}

which bounds the returns on human capital. The economic interpretation is simple. For $J$ to be concave even with endogenously increasing human capital, it has to be the case that the marginal product of human capital is not too big in the production function. For this to be the case, the effect of the relative flattening of the production function and decrease in marginal product of human capital has to be dominated by the relative increase in the marginal effect of the decrease in retention.

\end{proof}

\begin{prop}\label{main_prop:EE_wages}
For each state in which the firm is willing to continue the contract, the solution to the firm problem delivers the wage Euler equation:
\begin{equation} \label{eqn:euler_app}
\frac{\partial \widetilde{p}(\Theta)}{\partial W^\prime_{i}}\frac{J^\prime(\Theta)}{\widetilde{p}(\Theta)}  =  \frac{1}{u^\prime(w_{i})} - \frac{1}{u^\prime(w)}
\end{equation}
with $\Theta \equiv (\phi(h,y),\tau+1,W^\prime_i;\Omega^\prime)$ being the definition of the relevant state and $w_{i}$ is the wage paid in the future state.
\end{prop}
\begin{proof}
Consider the firm problem in Equation \eqref{eqn:firm_prob}, restated here for convenience
\begin{align*}
J(h,&\tau,\iota,W,y; \Omega)= \nonumber \\ &\underset{\pi_i, \{\eta^\prime_i,w_i,W_{i}^\prime\}}{\sup} \sum_{i=1,2} \pi_i  \Bigg( f(y,h;a)-w_i \nonumber \\
+ & \beta \mathbb{E}_{\psi}\bigg[ (1 - \eta^\prime_i) \cdot \max\bigg\{0,\mathbb{E}_{\Omega} \bigg[\widetilde{p}(h^\prime,\tau+1,\iota,W_{i}^\prime;\Omega^\prime) J(h^\prime,\tau+1,\iota,W_{i}^\prime,y;\Omega^\prime)\bigg]\bigg\}\bigg] \Bigg) \\
\nonumber\\
s.t.\; &[\mu]:  W =  \sum_{i=1,2} \pi_i \Bigg( u(w_i)+ \beta \mathbb{E}_{\Omega,\psi} ( (1 - \eta^\prime_i)\widetilde{r}(h^\prime,\tau+1,\iota,W_{i}^\prime;\Omega^\prime)  \nonumber\\
 &+ \eta^\prime_i U(h^\prime,\tau+1,\iota;\Omega^\prime)  ) \Bigg),\\
\nonumber \\
 &[\theta]:  J(h^\prime,\tau+1,\iota,W_{i}^\prime,y;\Omega^\prime) < 0 \implies \eta^\prime = 1 \text{, otherwise } \eta^\prime = 0 ,\\
\nonumber \\
  &[\psi]: (w - w^\prime)(1- \eta^\prime) \leq   0\\
 \nonumber \\
&\sum_{i=1,2} \pi_i = 1 \text{, }
\end{align*}
For $i=1,2$, the first order conditions with respect to the wage and the promised utilities are:
\begin{align}
    &[w_i]:\;  -1 + \mu u^\prime (w) - \tilde{\psi}= 0 \\
    &[W_{i}]:\;  \tilde{p}(\cdot) \cdot \frac{\partial J(W^\prime)}{\partial W^\prime }  +  J(W^\prime) \frac{\partial \tilde{p}}{\partial W^\prime } + \mu \frac{\partial \tilde{r}}{\partial W^\prime } +  \theta \frac{\partial J(W^\prime)}{\partial W^\prime }  = 0.\label{app_eqn:V_foc}
\end{align}
where $ \tilde{\psi} =  \psi (1- \eta^\prime) $.  Note that by definition,
\begin{align*}
\widetilde{r}(h,\tau,V;\Omega) & \equiv\delta U(h,\tau;\Omega)+(1-\delta)\Big[V + \lambda_{e}\max\{0,R(h,\tau,V;\Omega)\}\Big]
\end{align*}
therefore we can use the envelope theorem as in \cite{Benveniste1979}, Theorem 1 and the definition in Equation \textbf{Equation 8} to derive an expression for the derivative of the employment value in $t+1$ as the period ahead of the following:
$$\frac{\partial \widetilde{r}(h,\tau,W;\Omega)}{\partial W} = \widetilde{p}(h,\tau,W;\Omega). $$
Similarly, using the envelope condition on the firm problem and the first order condition for the wage, we can establish that
\begin{equation} \label{app_eqn:env_firm}
 \frac{\partial J(h,\tau,y,W;\Omega))}{\partial W} = -\mu \; \therefore \; \frac{\partial J(h,\tau,W,y;\Omega))}{\partial W} = -\frac{1 + \tilde{\psi}}{u^\prime(w)}.
 \end{equation}
Moving these two expressions one period ahead, substituting them in \eqref{app_eqn:V_foc}, focusing on $\tilde{p}(\cdot) > 0$ and $\pi_i > 0$ and rearranging we have that:
$$
(1 + \tilde{\psi}^\prime) \frac{1 + \theta}{u^\prime(w_i^\prime)} - \frac{1 + \tilde{\psi}}{u^\prime(w)} = \frac{\partial \widetilde{p}(\Theta)}{\partial W_i}\frac{J(\Theta)}{\widetilde{p}(\Theta)} ,$$
with $\Theta \equiv (h^\prime,\tau+1,W;\Omega^\prime)$ and where $w^\prime$ is the wage next period in state $\Omega'$.
We now turn to the behavior of the multipliers. To do so, let's focus on a contract at age $T-1$, for which retirement is certain at age $T$ and hence $ \eta^{\prime\prime} = 1 \implies \tilde{\psi}^\prime = 0 $. If $\eta^\prime = 0$, can the firm participation constraint be binding? When this is the case, $J(\Theta) = 0$ and $\theta > 0$. But this implies the wage constraint will be binding, so $\psi > 0$. It is then easy to show that, for values of $h^\prime,\tau+1,W;\Omega^\prime$ for which the firm participation constraint is binding, $\eta^\prime = 1$ and the firm exits. The argument can then be applied backwards for any value $\tau<T$ to show that if the firm state is such that the firm continues this period, then:

$$
  \frac{1}{u^\prime(w_i^\prime)} - \frac{1}{u^\prime(w)} = \frac{\partial \widetilde{p}(\Theta)}{\partial W_i}\frac{J(\Theta)}{\widetilde{p}(\Theta)} ,$$

\end{proof}

The optimal contract links the wage growth to the realization of firms profits. The right hand side of Equation \ref{eqn:euler_app} shows that, in providing insurance to the worker, the firm links wage growth to profits and to the incentive to maximize retention, incorporated in $\frac{\partial \log \widetilde{p}}{\partial W}$, the semi-elasticity of the retention probability to the utility offer. As the production stage takes place \emph{after} exit choices are taken by the incumbent firms, the wage growth related to the continuation value of the contract is bound to be (weakly) positive, hence workers enjoy a non-decreasing wage profile under the optimal contract.

A feature that the optimal contract derived in our model shares with the literature on long-term contracts with lack of commitment on the worker side is thus the backloading of wages.\footnote{See for instance, \cite{Thomas1988,tsuyuhara2016} and \cite{balke2022}.} Workers in our model make search decisions that affect the survival probability of the match. They do not however appropriate the full future value of the current match while making these search decisions (unless the firm makes zero profits). This makes it optimal for the firm to front-load profits and back-load wages. The reason is that the firm provides insurance and income smoothing to the worker, but given its risk neutrality it prefers to front-load its profits while providing an increasing compensation path to maximize retention. The contract thus optimally balances the consumption smoothing motives (i.e. the insurance provision of the contract) with the commitment problem of the worker.

\paragraph{Special case with log-utility.} The wage Euler equation discussed in \textbf{Proposition~\ref{main_prop:EE_wages}} can be simplified to a more intuitive interpretation in the log-utility case. In case of log-utility, in fact, $u'(w_{i,\Omega})=\frac{1}{w_{i,\Omega}}$. Multiplying and dividing by wage levels and rearranging, we can express the elasticity of retention probability to offered utility as
\begin{equation}\label{fluidity_equation}
    \varepsilon_{\widetilde{p},W_y} = \underbrace{ \frac{(w_{i}-w)}{w}}_{\text{Wage growth}}  \underbrace{\frac{w}{J(\Theta)}}_{\substack{\text{Ratio of wage} \\\text{to match value}}}.
\end{equation}
with $\varepsilon_{\widetilde{p},W} \equiv \frac{\partial \widetilde{p}(\Theta)}{\partial W_{i} \widetilde{p}(\Theta)}$.

The interpretation of this result is of interest to analyses that relate labor market dynamism to wage dynamics, like \cite{engbom2020}. This is because $\varepsilon_{\widetilde{p},W_y}$, being a function of structural parameters of the matching technology, $\gamma$, search frictions $\lambda_e$, and measures of labor market tightness $\theta$, provides us with a good proxy of labor market fluidity. The right hand side of \eqref{fluidity_equation}, is composed entirely of observable quantities, as the  ratio of wages to match value is a function of factor shares in value added. The quantity can then be used to compare the dynamism of different local, regional or national labor markets. \\
The next proposition, instead, confirms our initial conjecture that in equilibrium firm qualities and utility promises are related to a one-to-one mapping.

\begin{lem}\label{app_lem:qJ_Wconc}
$\Pi(y,h,\tau,W;\Omega)$ is concave in $W$ and optimal utility promises are unique and increasing in $y$ given worker characteristics $(h,\tau)$ and aggregate state $\Omega$.
\end{lem}

\begin{proof}

Assuming the same $(h,\tau,y)$, the entrepreneur then chooses the optimal value $W^*(y):=\text{arg}\max_W \Pi(W;\Omega)$ to deliver in the contract. 
For the rest of the proof we consider as given the dependence of the functions on $(h,\tau,\Omega)$ and consider directly the composite function $q(\theta(W))$ as $q(W)$.

The optimization requires that $W^*$ satisfies the following first order condition:
\begin{align} \label{app_eqn:foc_free}
    q_W(W^*)J(y,W^*) + q(W^*)J_W(y,W^*) = 0.
\end{align}
For $W^*$ to be a unique maximum, $\Pi(\cdot)$ has to be concave hence the second-order condition must be negative. Taking into account explicitly of the composite function $q(\theta(W))$ implies the following second-order condition:
\begin{align}
& q_{\theta,\theta} (\theta_W)^2 J + q_{\theta} \theta_{W,W} J + 2 q_{\theta} \theta_W J_W + q J_{W,W} < 0 \\
& (q_{\theta,\theta} \theta_W + q_{\theta} \frac{\theta_{W,W}}{\theta_{W}}) \theta_W J + 2 q_{\theta} \theta_W J_W + q J_{W,W}<0\\
&(q_{\theta,\theta} \theta_W + q_{\theta} \frac{\theta_{W,W}}{\theta_{W}}) \underbrace{\theta_W J}_{<0} + \underbrace{2 q_{\theta} \theta_W J_W + q J_{W,W}}_{<0} <0 \label{app_eqn:soc_free}
\end{align}
where the inequalities in \textbf{Equation~\ref{app_eqn:soc_free}} follow from knowing that $J(\cdot)$ is concave and decreasing in $W$, $J_W, J_{W,W} \leq 0$, $q(\cdot)$ is increasing in $W$, and $\theta(\cdot)$ is decreasing in promises as discussed in \textbf{Lemma~\ref{app_cor:p_properties}}. Note that if $q_{\theta,\theta} \theta_W + q_{\theta} \frac{\theta_{W,W}}{\theta_{W}} > 0$ then the inequality in \ref{app_eqn:soc_free} would be respected and $\Pi(\cdot)$ would be concave. 
For this term to be positive it has to be that:
\begin{align}
&\frac{q_{\theta,\theta}}{q_{\theta}} \theta_W > - \frac{\theta_{W,W}}{\theta_W}. \label{app_eqn:soc_free_condition}
\end{align}
As $q(\cdot)$ is decreasing and convex, $q_{\theta,\theta} > 0,q_{\theta}<0$ and $\theta_W<0$, the left-hand side of \ref{app_eqn:soc_free_condition} is positive. If $\theta_{W,W} < 0$ the right-hand side would be negative and the inequality would always be respected. We know from \textbf{Lemma~\ref{app_cor:p_properties}} that $\theta_{W} = - q^{-1}_{\xi}J_W\frac{c}{J^2}$, so that we can write $\theta_{W,W} = \frac{c(cJ_W^2q^{-1}_{\xi,\xi} - J^2J_{W,W}q^{-1}_{\xi}  +2 J J^2_W q^{-1}_{\xi})}{J^4}$. As the cost function and the value of the contract are positive,  $\theta_{W,W} < 0$ if the numerator is  negative, so if 
$cJ_W^2q^{-1}_{\xi,\xi} - J^2J_{W,W}q^{-1}_{\xi}  +2 J J^2_W q^{-1}_{\xi} < 0$.

This is true if
$$ \frac{q^{-1}_{\xi,\xi}}{q^{-1}_{\xi}} > \left(\frac{J_{W,W}}{J_{W}}\frac{J}{J_W} - 2 \right) \frac{J}{c}.$$

As the right-hand side of the condition above is negative given the properties of $J$ and $q^{-1}_{\xi} < 0$, the inequality would always be respected if $q^{-1}_{\xi,\xi}<0$. We verify that the conditions holds for our matching function and parametrization.

The uniqueness of the optimum for the firm problem at object implies, given the properties of $\mathcal{V}$ and the theorem of the maximum, $W^*(y)$ is also continuous.

To show that $W_y^*(y)>0$ note that, by the implicit function theorem, the derivative of \textbf{Equation \ref{app_eqn:foc_free}} with respect to $y$ is:
\begin{align}
    (q_{W,W}J + 2q_WJ_W + qJ_{W,W})W_y + q_WJ_y + qJ_{W,y} = 0.
\end{align}

The first term in parenthesis is negative, as per the second order condition in \textbf{Equation~\ref{app_eqn:soc_free}}. The second term is positive, given that $J_y$ is positive  (\textbf{Lemma \ref{lem:y_incr}}) and $q_W$ is positive as well. From \textbf{Equation~\ref{app_eqn:env_firm}} it is possible to show that $J_{W,y} = \Gamma W_y$, with $\Gamma = \frac{u''}{(u')^2}\frac{\partial w}{\partial W}<0$ as $\frac{\partial w}{\partial W}\geq0$ and $u''<0$. Therefore, 
$$\underbrace{(q_{W,W}J + 2q_WJ_W + qJ_{W,W} + q\Gamma)}_{<0}W_y + q_WJ_y = 0.$$

This means that, in order for the equality to be respected, $W^*_y>0$.

\end{proof}

\begin{prop} \label{main_prop:mapping}
The mapping defined by the function $W: \mathcal{Y} \rightarrow \mathcal{V}$ is bijective for each worker characteristic $(h,\tau)$ and aggregate state $\Omega$.
\end{prop} 
\begin{proof}
As shown in \textbf{Lemma~\ref{app_lem:qJ_Wconc}}, the optimal promise for each entrepreneur with quality $y$ is unique, continuous and monotonically increasing. Free entry implies that 
$\Pi(W^*(y)) = c(y)$, pinning down $y$ in equilibrium.
As $W^*(y)$ is unique and increasing, for each $y$ there can be only one $W^*(y)$ so that the mapping between firm qualities and optimal promises is bijective. 
\end{proof}

\begin{cor}\label{main_prop:searchs_incrina}
The optimal search strategy of the workers and vacancy posting of the firms is increasing in aggregate productivity.
\end{cor}
\begin{proof}
We start from the free entry condition, ignoring without loss of generality $\tau,\mu$ from the notation, and ignoring the indirect dependence of $q(\cdot)$ on $\theta(\cdot)$:

\[
q(h,W;a)J(h,W,y;a) = c(y)
\]

We know from \textbf{Lemma~\ref{app_lem:qJ_Wconc}} and \textbf{Proposition~\ref{main_prop:mapping}} that the optimal $W^*$ choice is a unique, monotonically increasing function $W(y)$ given workers characteristics $(h,\tau)$ and aggregate state $\Omega$. Given \textbf{Proposition~\ref{main_prop:wrk_char}}, as the worker search policy is continuous, unique and increasing in $h$, we can also conclude that firms and workers qualities feature positive assortative matching for any aggregate state, according to a continuous function $y = \xi(h;a)$. We thus want to find conditions for which $\xi_a > 0$, which would imply that in equilibrium the ``assignment'' (search) firm quality $y$ for worker quality $h$ is higher when the state of the economy $a$ is better. We write $\xi_a$ simply as $\frac{\partial y}{\partial a}$ in the proof.

We obtain the partial derivative of the free entry condition by $y$:

\begin{equation}\label{eq:marginal_fe}
    c'(y) = q\frac{\partial J}{\partial y}
\end{equation}

We then write the total derivative of of this equation by $a$, ignoring terms featuring $W$ by envelope condition:

\begin{equation}\label{eq:totaldiff_a}
   q \left(\frac{\partial^2 J}{\partial y \partial a} + \frac{\partial^2 J}{\partial^2 y}\frac{\partial y}{\partial a}\right) + \frac{\partial J}{\partial y}\frac{\partial q}{\partial a} = c''(y) \frac{\partial y}{\partial a}  
\end{equation}

which, in turn, leads to:

\begin{equation}\label{eq:dyda}
    \frac{\partial y}{\partial a} = \frac{q\frac{\partial^2 J}{\partial y \partial a} + \frac{\partial J}{\partial y}\frac{\partial q}{\partial a}}{c''(y) - \frac{\partial^2 J}{\partial ^2 y}}
\end{equation}

The sign of the denominator is positive given the properties of the function $c(\dot)$ and \textbf{Lemma~\ref{app_lem:jconcy}}, hence it is sufficient to sign the numerator, and $\frac{\partial y}{\partial a} > 0$ amounts to

\[
q\frac{\partial^2 J}{\partial y \partial a} > -\frac{\partial J}{\partial y}\frac{\partial q}{\partial a}
\]

We check for this condition in the solution and find it to be respected in the state space given our calibration. 

\end{proof}

Firm value $J(\cdot)$ is increasing in $a$, while retention is decreasing. As firm make more profits, they can marginally increase offered utilities to workers to maximize retention. The last inequality implies that for the workers to improve their target firm quality in better times it must be true that the first marginal effect, the increase in firm profits from better matches, must dominate the second, the relative decrease in retention, when aggregate productivity changes. %

%% file: AppendixBRE_short.tex
\section{Existence of a Block Recursive Equilibrium}\label{Appendix:BRE_existence}
In order to show that a Block Recursive Equilibrium (BRE) exists in our model we need to show that the equilibrium contracts, the workers' and the entrepreneurs value and policy functions do not depend on the distribution of employed and unemployed workers.
This implies that the only element of the aggregate state that matters for a firm when making an hiring decision is the state of aggregate productivity but not the distribution of worker types (e.g. employed vs unemployed).
\begin{prop}
A Block Recursive Equilibrium as defined in Definition~\ref{defi:bre_defi} exists.
\end{prop}
\begin{proof}
We follow the approach in \cite{menzio2016,herkenhoff2019} and prove the existence of a BRE using backward induction. %

Consider the lifetime values of an unemployed and an employed worker before the production stage in the last period of households lives with $\tau = T$ :
\begin{align}
&U(h, T, \iota;\Omega) = u(b(h,T))\\
&V(h, T, \iota;\Omega) = u(w(a)),
\end{align}
their values trivially do not depend on the distribution of types as both valuations are $0$ from $T+1$ onward. Hence, $ U(h, T, \iota; \Omega) = U(h, T, \iota; a)$ and $V(h, T, \iota; \Omega) = V(h, T, \iota; a)$.

The optimal contract for agents aged $\tau = T$, instead, solves the following problem
\begin{equation*}
J(h,T,W,y; \Omega)=  \underset{w}{\sup} [f(y,h;a)-w] \quad
s.t.\; W =  u(w),
\end{equation*}
that clearly does not depend on the distribution of worker types due to the directed search protocol and where the aggregate state only affects the promised utility and the optimal wage through realization of the aggregate productivity processes. Therefore, $J(h,T,\iota,W,y; \Omega) = J(h,T,\iota,W,y;a)$.

This also implies that the equilibrium market tightness
$$ \theta(h,T,\iota,W; \Omega) = q^{-1}\left( \frac{c(y)}{J(h,T,\iota,W,y;a)} \right)$$
is independent from the distribution of worker types and it is only affected by realization of aggregate productivity, so $\theta(h,T,W;a)$.

This in turn implies that the search problem workers face at the beginning of the last period of their lives depends on the aggregate state only through aggregate productivity $a$:
\begin{equation*}
R(h,T,\iota,V;a) = \underset{v}{\sup}\;\Big[ p(\theta(h,T,\iota,v;a))\big[v-V]\big]\Big],
\end{equation*}
does not depend on the distribution of worker types.

Stepping back at $\tau = T-1$, the value functions for the unemployed and the employed agents are solutions to the following dynamic programs
\begin{align*}
&\underset{v}{\sup} \; u(b(h,T-1))+ \beta \mathbb{E}_{\Omega,\psi}\bigg(U(h^{\prime},T,\iota;a^\prime) + p(\theta(h,T,\iota,v;a^\prime))\big[v-U(h^\prime,T,\iota,;a^\prime)\big]\bigg) \\
&u(w) +\beta \mathbb{E}_{\Omega,\psi}\left(\begin{array}{c}
\delta U( h^\prime,T,\iota;a^\prime)
+ \beta(1-\delta)W
+ \\ + \beta(1-\delta)\lambda_{e} \max(0,R(h^\prime,T,\iota,,W);a^\prime)]\big]\Big]
\end{array}
\right),
\end{align*}
where both do not depend on the distribution of worker types.

The optimal contract at this step is a solution to
\begin{align*}
J_{t}(h,T-1,\iota,V,y;& a)=  \underset{\{\pi_i,w_i,W_{i}\}}{\sup} \sum_{i = 1,2} \pi_i \Big( f(y,h;a)-w_i \nonumber \\
+ & \mathbb{E}_{\Omega,\psi}\left[\widetilde{p}(h^\prime,T,W_{i,\Omega^\prime};a^\prime)(J(h^\prime,T,y,W_{i};a^\prime)\right] \Big) \\
\nonumber\\
s.t.\; &V = \sum_{i = 1,2} \pi_i \left( u(w_i)+ \mathbb{E}_{\Omega,\psi} \widetilde{r}(h^\prime,T,W_{i};a^\prime)\right), \; h^\prime = \phi(h,y,\iota,\psi) \\
&\mathbb{E}_{\Omega,\psi} \sum_{i = 1,2} \pi_i  \left( \mathbb{E}_{\Omega,\psi} J(h^\prime,T,\iota,W_{i},y;a^\prime) \right)\geq0\,\,\text{and}\,\,t\leq T
\end{align*}
which does not depend on types distribution.

Therefore, also the equilibrium tightness and the search gain at $T-1$ are independent from types' distributions, as
\begin{align*}
    &\theta(h,T-1,\iota,W;a) = q^{-1}\left( \frac{c(y)}{J(h,T-1,\iota,
    W,y;a)} \right) \\
    &R(h,T-1,\iota,V;a) = \underset{v}{\sup}\;\Big[ p(\theta(h,T-1,\iota,v;a))\big[v-V]\big]\Big].
\end{align*}
Stepping back from $\tau = T-1,...,1$ and repeating the arguments above completes the proof.
\end{proof}

%% file: appendix_SolutionAndSMM.tex
\section{Model solution and estimation}\label{sect:app_computation}
\defcitealias{2020SciPy-NMeth}{Virtanen et al., (2020)}

\paragraph{Model solution.}
For each education level, we solve the model by backward induction on a regular grid of human capital, wage multiplier levels, firm quality and aggregate shocks. In particular, starting from the last period of workers lives we compute the terminal wage in each state, then solve the search problems for both the employed and the unemployed, which allows us to compute the market tightness in each submarket as well as the value of the contract and of unemployment. With these objects in hand we proceed backwards until workers labor market entry.

In the model simulations, we populate each cohort of agents with 42 individuals, 30 non-graduates and 12 graduates, consistent with the share of graduates in the Italian economy. Graduates enter the labor market in a staggered manner every quarter for the first three years of their life. We then simulate the model for 600 periods (we take 181 periods as burn in) and construct a panel with 419 periods and $7,560$ individuals per period. In the model, we consider one period as one quarter in the data.\footnote{We use a grid of $20 \times 15 \times 25 \times 7$ for each education level and we consider 180 quarters of workers life. After numba accelerations, solving and simulating the model takes approximately 3 minutes with Python 3.12/Fortran95 on a Linux HPC cluster with 64 cores and 120G of RAM.} 

\paragraph{Estimation.}
To pin down the 22 internal parameters in the model, $\boldsymbol{\theta}$, we minimize the Euclidean distance in the scaled arc-percent deviation $e(\boldsymbol{\theta}) = \frac{m(\boldsymbol{\theta}) - \mathbf{d}}{0.5(|m(\boldsymbol{\theta})|+|\mathbf{d}|)+0.25}$ between model moments, $m(\boldsymbol{\theta})$, and their empirical counterparts, $\mathbf{d}$.\footnote{This specific formulation of the error metric helps to deal with potential issues that could arise from the large variation in the scales of the moments \citep{guvenen_what_2021}.} More formally, we pick the vector of parameters 
$$ \boldsymbol{\theta}^* = \arg\min_{\boldsymbol{\theta}} \left\{|| e(\boldsymbol{\theta})' W e(\boldsymbol{\theta})||) \right\}.$$
 
In $\mathbf{d}$ and $m(\boldsymbol{\theta})$  we stack: the profiles of E-E and E-U transitions by age (2 age groups) and education level; earnings growth relative to the initial earnings by education level and experience (2 experience groups); the unemployment rate by education level; the correlations of labor market flows with the cyclical component of aggregate output by education level; average sorting, measured as the average correlation between AKM fixed effects in the data and as the correlation of worker and firm qualities in the model; the ratio of the third and second quintiles of value added per employee; and the share of workers unemployed by more than 12 months by education level.

Jointly, these moments deliver 22 restrictions for 22 parameters. In model simulations, we construct each moment weighting by age from Italian population weights. 
We solve for $ \boldsymbol{\theta}^*$ numerically adopting two complementary strategies. In the first, we use a two step procedure and minimize the distance between model-generated and data-generated moments using a global solver on the largest feasible domain. For this, we rely on the differential evolution algorithm contained in the SciPy library developed by \citetalias{2020SciPy-NMeth}. In the second, we solve and simulate the model over a sparse grid of the parameter space.\footnote{The advantage of selecting the parameters with this procedure is that the grid is built so that the function domain is optimally covered with the least amount of possible points compared to other forms of approximation (e.g. equi-spaced or random grids). The sparse grid is built using the``Tasmanian" libraries \citep{stoyanov2015tasmanian}.} In both cases, we refine the estimation using the best combination of parameters as the starting point of a local minimization routine, for this we rely on the Nelder-Mead algorithm in the SciPy library.

We also use the series generated by the differential evolution--1,500 points from a Sobol sequence spanning the parameter space--to examine the role of each moment in identifying our parameter values. We then project each parameter $k$ on the simulated moments using a LASSO regression with 0.8 penalty,
\begin{equation}\label{eq:lasso_id}
\hat{\beta} = \arg\min_{\beta} \left\{ \sum(\theta_{i} - M(\theta_i)'\beta)^2 + \lambda \sum |\beta_j| \right\},
\end{equation}
and then normalize the non-zero coefficients that we report in Figure~\ref{fig:app_IDfigLasso}. 

We also complement this global approach with a local perturbation of our parametrization. Specifically, we perturb each parameter by $\pm 2.5\%$ and re-simulate the model to derive the associated elasticities for each target moment. We then compute the average of these elasticities and collect them in an average elasticity matrix and verify that this Jacobian is invertible, confirming local identification. To provide some intuition, in \textbf{Figure~\ref{fig:IDelas}} we report the heatmaps for two normalizations of this matrix: Panel~(a) shows the normalization by column maxima, highlighting for each parameter the moment that most discipline it, while Panel~(b) reports the normalization by row maxima, identifying for each moment the parameter to which it is most informative. 

Unemployment durations, business cycle correlations of transition rates, Earnings growth by education, and firm productivity distribution measures emerge as the dominant sources of identification. Transition rates (EE, EU) exhibit somewhat weaker relationships with parameters than stock-based moments like unemployment rates, reflecting the long-run consistency between stocks and flows: while individual flow moments may show attenuated elasticities, they discipline the model through their joint determination of equilibrium stocks. Targeting both flows and stocks ensures internal consistency and sharpens identification. Reassuringly, the LASSO-based and local elasticity analyses yield consistent findings---parameters showing strong local elasticities are recovered using similar moments in the global LASSO exercise.

\begin{figure}
\caption{Normalized LASSO coefficient from regressing parameters on moments}\label{fig:app_IDfigLasso}
\includegraphics[width=0.8\textwidth]{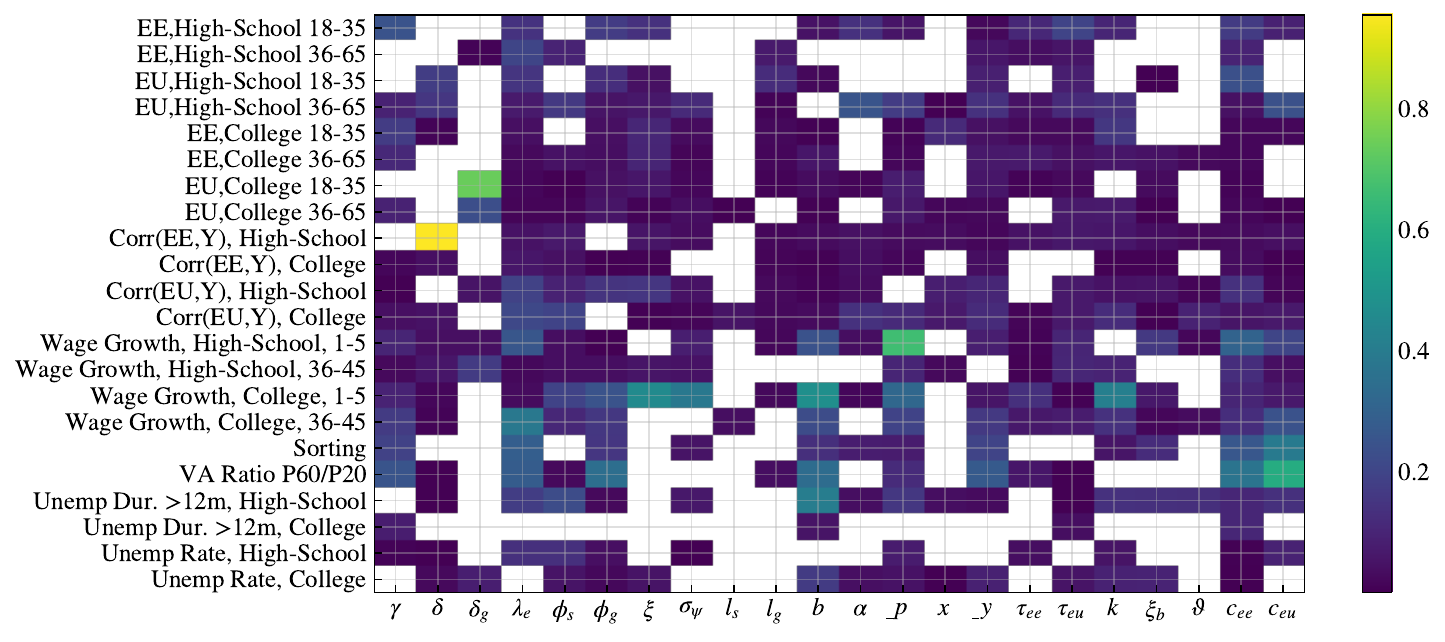}
\floatfoot{\textbf{Note:} The figure presents the heatmap of non-zero LASSO coefficients obtained from Equation~\ref{eq:lasso_id}. Each regression of parameters on moments uses a regularization hyper-parameter of 
$\lambda=0.8$. The estimation sample consists of 1,500 parameter points generated via a Sobol sequence and excluding simulations that generate errors in the construction of the target moments.}
\end{figure}

\begin{figure}
\caption{Average elasticity matrices}\label{fig:IDelas}
\subcaptionbox{Parameters to Moments\label{fig:IDfigCol}}{\includegraphics[width=0.5\textwidth]{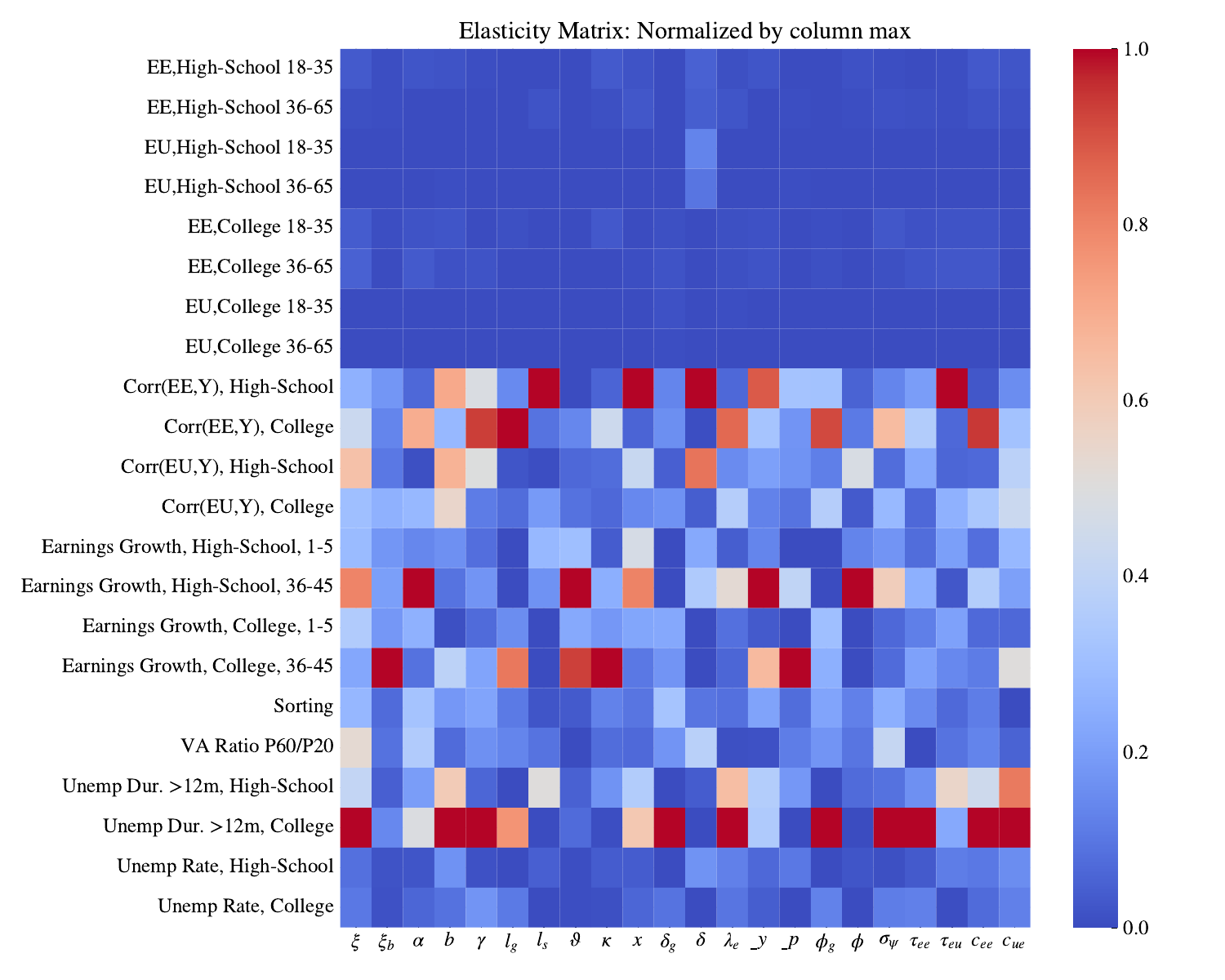}}%
\subcaptionbox{Moments to Parameters\label{fig:IDfigRow}}{\includegraphics[width=0.5\textwidth]{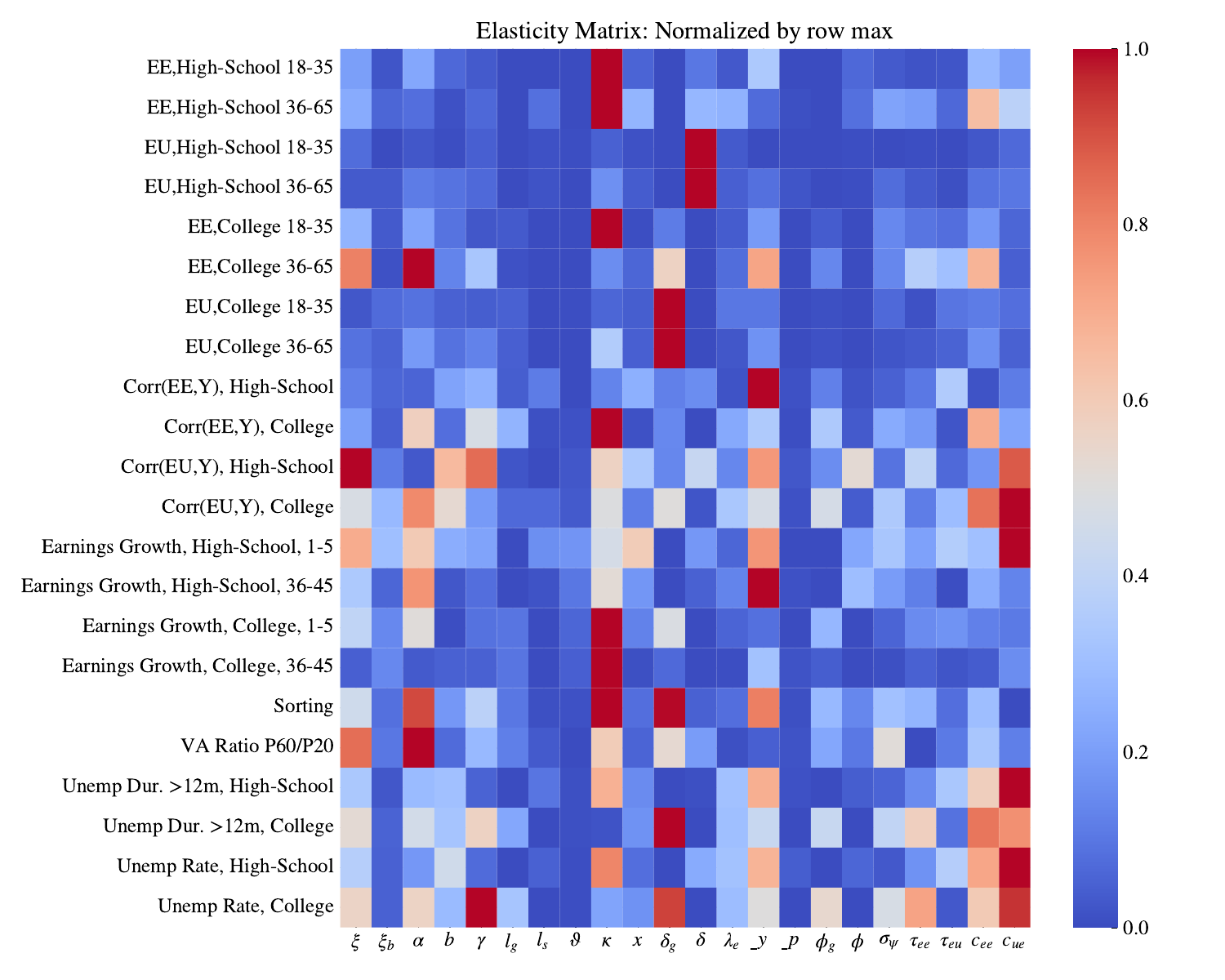}}%
\floatfoot{\textbf{Note:} Each panel reports the average local elasticity of model moments to parameter perturbations of $\pm 2.5\%$. Panel (a) normalizes elasticities by column maxima and highlights, for each parameter, the moments that most strongly discipline it. Panel (b) normalizes elasticities by row maxima and highlights, for each moment, the parameter to which it is most informative.}
\end{figure}

\begin{table}[h]
\caption{Target moments\label{app_tab:targets}}
\resizebox{0.7\textwidth}{!}{\input{figures/PaperFigures/MomentsTable.txt}}
\end{table}

%% file: figures/PaperFigures/MomentsTable.txt
\begin{tabular}{lcc}
\toprule
\textbf{Moment} & \textbf{Data} & \textbf{Model} \\
\midrule
\multicolumn{3}{l}{\textbf{A. E-E Transition Rates, High School}} \\
\midrule
18-35 & 3.7\% & 4.6\% \\
36-65 & 2.2\% & 3.3\% \\
\midrule
\multicolumn{3}{l}{\textbf{B. E-E Transition Rates, College}} \\
\midrule
18-35 & 3.4\% & 3.8\% \\
36-65 & 1.7\% & 1.8\% \\
\midrule
\multicolumn{3}{l}{\textbf{C. E-U Transition Rates, High School}} \\
\midrule
18-35 & 6.0\% & 4.6\% \\
36-65 & 3.2\% & 3.5\% \\
\midrule
\multicolumn{3}{l}{\textbf{D. E-U Transition Rates, College}} \\
\midrule
18-35 & 4.1\% & 2.8\% \\
36-65 & 1.2\% & 1.6\% \\
\midrule
\multicolumn{3}{l}{\textbf{E. Earnings Growth rel. to entry, High School}} \\
\midrule
1-5 & 0.32 & 0.32 \\
36-45 & 1.13 & 1.14 \\
\midrule
\multicolumn{3}{l}{\textbf{F. Earnings Growth rel. to entry, College}} \\
\midrule
1-5 & 0.36 & 0.29 \\
36-45 & 2.48 & 1.85 \\
\midrule
\multicolumn{3}{l}{\textbf{G. Flows Cyclicality (GDP Correlation)}} \\
\midrule
E-E, High School & 0.76 & 0.53 \\
E-E, College & 0.55 & 0.46 \\
E-U, High School & -0.23 & -0.15 \\
E-U, College & 0.03 & 0.01 \\
\midrule
\multicolumn{3}{l}{\textbf{H. Other Moments}} \\
\midrule
Unemployment Rate, High School & 9.4\% & 8.2\% \\
Unemployment Rate, College & 6.2\% & 4.7\% \\
Sorting & 0.40 & 0.42 \\
VA Ratio (P60/P20) & 1.20 & 1.52 \\
Share of Unemployed $>$12 months, High School & 63.6\% & 63.3\% \\
Share of Unemployed $>$12 months, College & 45.8\% & 58.0\% \\
\bottomrule
\end{tabular}

%% file: appendix_AdditionalFig.tex
\section{Additional Figures and Tables, Model}\label{sect:app_additional_tables_andfig}

\subsection{Inequality dynamics} \cite{heathcote2020} show that recessions have a persistent effect on inequality, affecting the earnings of workers in the left tail of the income distribution the most. \textbf{Figure~\ref{fig:HPV_fig}} illustrates the dynamics of inequality around business cycles by displaying the pattern of losses across the  earnings distribution. Recessions hit the poorest workers the hardest, and worsening job prospects push some of them out of the labor force. 

A prediction of the model that departs from existing literature is that the persistence of earnings losses varies across the distribution: while workers with less human capital display more volatility in earnings (mostly due to displacement), the impact on workers with high human capital is dampened but quite persistent. The presence of limited commitment to job matches by firms creates a subset of workers who regularly undergo shorter spells of employment. Their layoffs exhibit strong counter-cyclical behavior (a pattern highlighted in \citealt{jarosch2023}). 

Due to match-specific human capital accumulation, the opportunity for workers to recover from low levels of human capital depends on climbing the job ladder. However, falling off the lowest rungs of the job ladder is more likely than from the highest rungs. In the data, losing a job implies that on average workers will land a worse job after the unemployment spell. Displacement knocks workers down to lower rungs of the firm productivity ladder and makes transitions to higher-productivity firms less likely, effectively flattening the job ladder. We estimate an event-study specification à la  \cite{jacobson1993} where the outcome is the quality of the employer, measured as the log of value added per worker of the worker's current firm. We compare displaced workers to similar non-displaced workers, and trace the coefficients on event time from three years before displacement to seven years after displacement and show results in \textbf{Figure \ref{fig:jls}}. The model generates an analogous flattening of the job ladder, with a realistic cross-sectional distribution of the cross-section of separations.

Consequently, the model endogenously generates a dual economy with workers whose earnings cyclicality depends on the extensive margin of employment (the lowest deciles of income distribution), and higher income workers whose earnings cyclicality depend on the quality of their employment over time \citep[see][]{doniger2023}.

\begin{figure}
    \centering 
    \caption{Inequality dynamics  \label{fig:HPV_fig}}
    \includegraphics[width=0.55\textwidth]{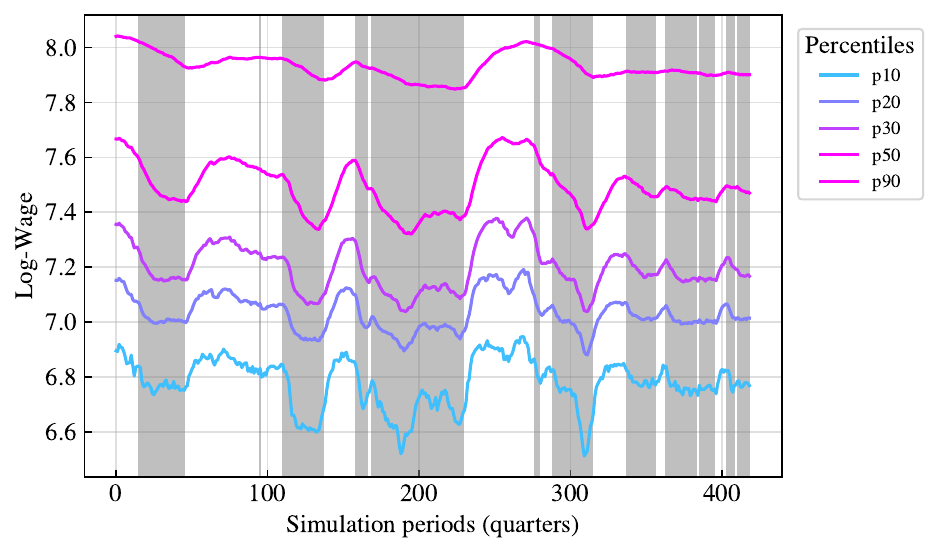}
    \vspace{-1.0em}
    \floatfoot{\textbf{Note}:The figure reports the dynamics of income for different
    percentiles of the income distribution.}
\end{figure}

\begin{figure}[!h]
\captionsetup{type=figure}
\caption{Firm quality after displacement}\label{fig:jls}
\centering
\includegraphics[width=0.48\textwidth]{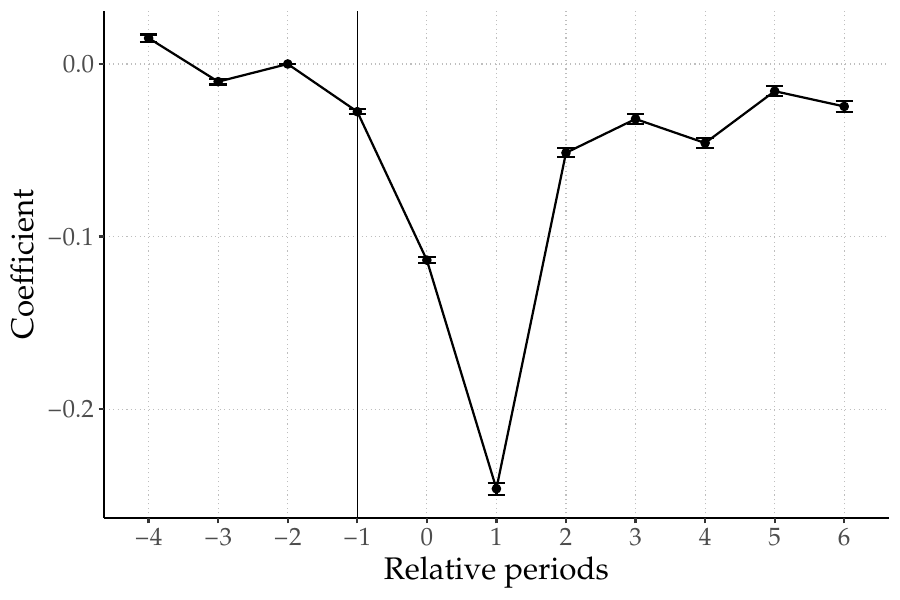}
\floatfoot{\textbf{Note}: The figure plots the change in the average quality of employers -- measured as the log of value added per employee -- \textit{after} a displacement event \citep{jacobson1993}.}
\end{figure}

\subsection{State dependence}
We check for state dependence by running the following regression across different model simulations 
\begin{equation*}
\resizebox{0.95\textwidth}{!}{$
Y_{\text{Post}} = \alpha + \sum_{j=1}^{4} \beta_{j}\underbrace{\mathbb{E}[FQ^{j}_\text{Pre}]}_{\text{Firm Quality}} + \gamma_j \overbrace{\mathbb{E}[HC^{j}_\text{Pre}]}^{\text{Human Capital}} + \delta_j\underbrace{\mathbb{E}[W^{j}_\text{Pre}]}_{\text{Wages}} + \eta_j\overbrace{ \mathbb{E}[LS^{j}_\text{Pre}] }^{\text{Labor Share}} + \theta \underbrace{\mathbb{C}(FQ_\text{Pre},HC_\text{Pre})}_{\text{Sorting}} + \varepsilon,
$} 
\end{equation*}
in which $\mathbb{E}[X^{j}_\text{Pre}]$, denotes the $j^{th}$ moments of an endogenous, cross-sectional distribution before the shock, and $Y_{\text{Post}}$ are model outcomes after the shock. Results are displayed in \textbf{Table \ref{app_tab:state_dependence}}.

\begin{table}[h!]
\caption{State Dependence\label{app_tab:state_dependence}}

\subcaptionbox{Mean}{
\resizebox{0.43\textwidth}{!}{
\input{figures/PaperFigures/StateDepTable_raw_avg.tex}
}
}%
\subcaptionbox{Variance}{
\resizebox{0.43\textwidth}{!}{
\input{figures/PaperFigures/StateDepTable_raw_var.tex}

}
}

\vspace{0.5em}
\subcaptionbox{Skewness}{
\resizebox{0.43\textwidth}{!}{
\input{figures/PaperFigures/StateDepTable_raw_ske.tex}

}
}%
\subcaptionbox{Kurtosis}{
\resizebox{0.43\textwidth}{!}{
\input{figures/PaperFigures/StateDepTable_raw_kur.tex}

}
}
\floatfoot{\textbf{Note:} Standard-errors in parenthesis. 
\noindent The table reports the coefficients of regressing model outcomes after the shock on moments of relevant endogenous variables before the shock (one year average). The results are scaled to report the effect of changing the regressors by 1pp. Sorting is computed every quarter as the correlation between firm and worker quality and reported in the same table of the means even if 
technically is not a moment of a distribution. Model outcomes are: i) short-term cumulative output response (1 year after the shock); ii) the medium-term cumulative output response (3 years after the shock); and iii) the persistence of the shock (number of quarters before the output IRF is back at zero). For ease of exposition, the results are grouped by moments but the coefficients are computed including all moments in the same regression. The simulations are those underlying \textbf{Figure~\ref{fig:recession_experiment}}.}
\end{table}

\subsection{Additional validation and results}
We report additional evidence on the ability of the baseline model to replicate relevant data features and we augment the evidence from Section~\ref{validation} with two additional models: one in which we prohibit endogenous separations and one in which we remove the downward wage rigidity. In both cases the models fail to replicate all the key features of the data that our baseline model manages to capture.

\begin{figure}
\centering
\includegraphics[width=0.5\linewidth]{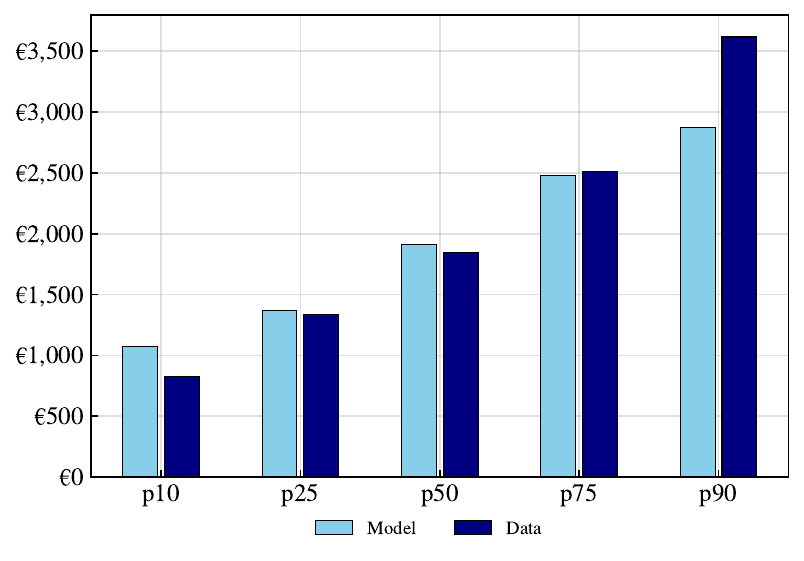}
\caption{Wage distribution\label{fig:wage_dist}}
\floatfoot{\textbf{Note:} The figure reports the income distribution in our sample and in model simulations.}
\end{figure}

\begin{figure}[h]
     \centering
     \caption{GDP: Model and Data}\label{fig:gdp_data_model}
     \includegraphics[width=0.5\textwidth]{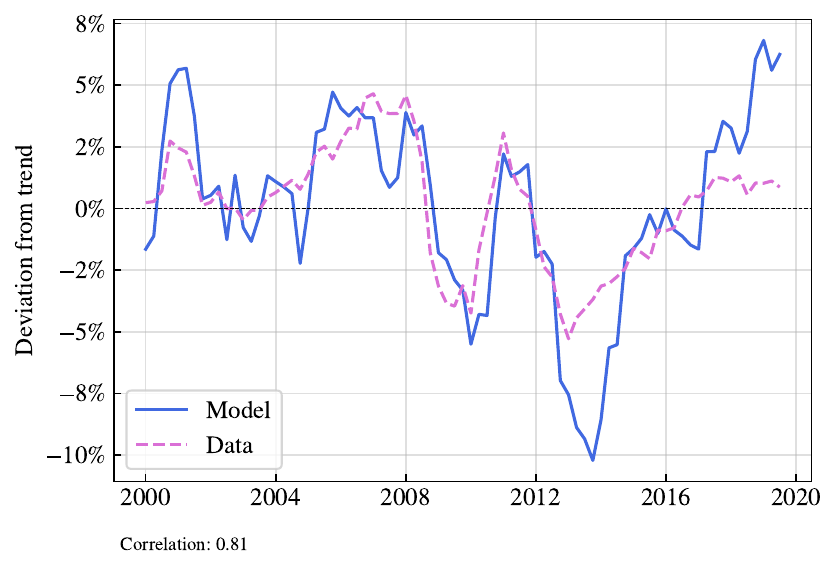}    
     \vspace{-1.0em}
     \floatfoot{\textbf{Note:} The figure plots the cyclical components of real GDP for Italy and for a model simulation in which the TFP process is matched to the Italian TFP realizations from 2000 to 2019, both series are detrended using an Hamilton filter (4 lags and 8 leads) and their correlation is robust to the choice of the filter.}
\end{figure}

\begin{figure}[h]
 \centering 
 \caption{Scarring Effect of Recessions}\label{fig:scarring_model_data}
 \subcaptionbox{\footnotesize Data}{\includegraphics[width=0.3\textwidth]{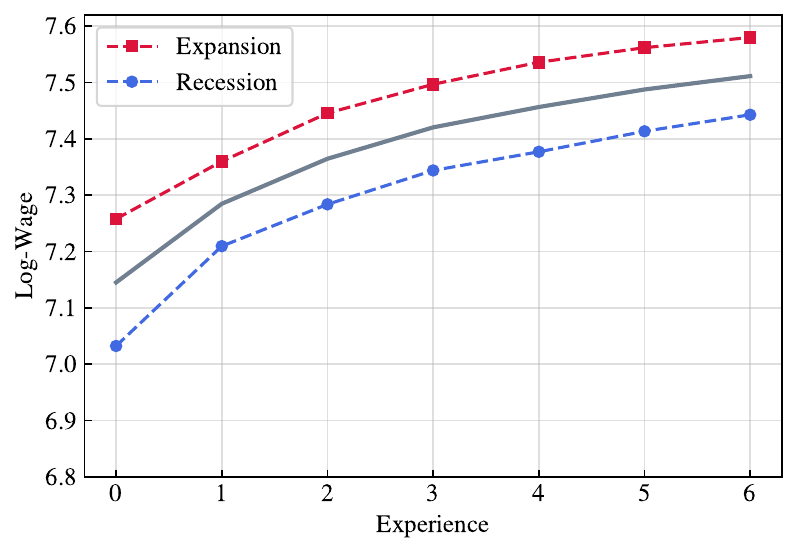}}
 \subcaptionbox{\footnotesize Model: Baseline}{\includegraphics[width=0.3\textwidth]{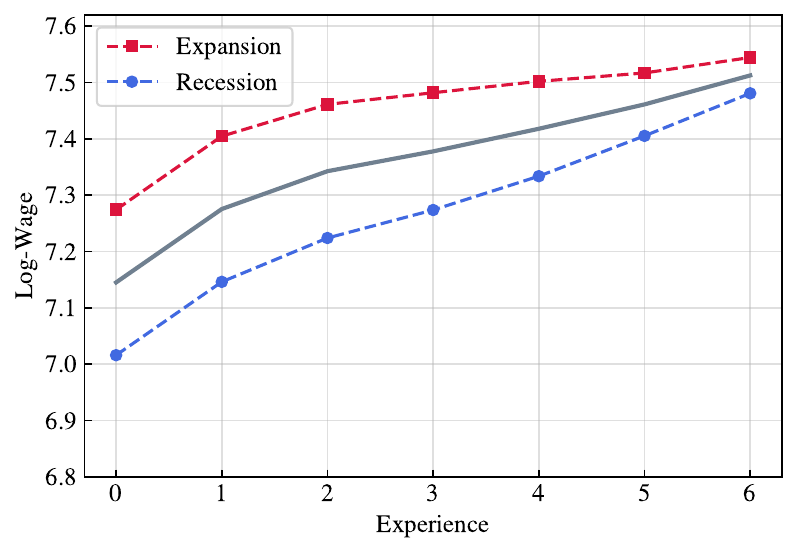}}
 \subcaptionbox{\footnotesize Model: Linear HC acc.}{\includegraphics[width=0.3\textwidth]{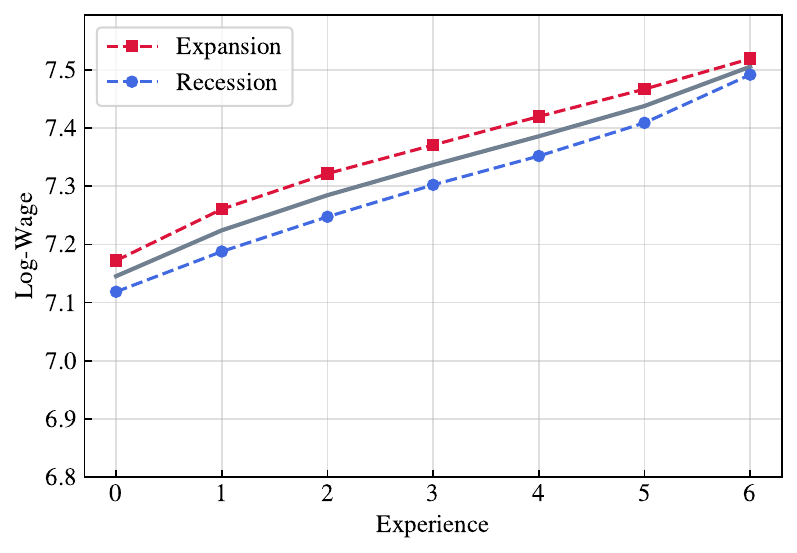}}
 
 \floatfoot{\textbf{Note}: The figure plots the wage profiles estimated on the data and on model simulations for cohorts of workers entering the labor market. The counterfactual profiles for expansions (recessions) are obtained considering a positive (negative) two standard deviation realization of cyclical GDP.}
\end{figure}

\begin{figure}[!h]
    \centering
    \caption{Shares of E--U transitions across firms' productivity distribution \label{fig:separationsShares}}
    \includegraphics[width=\textwidth]{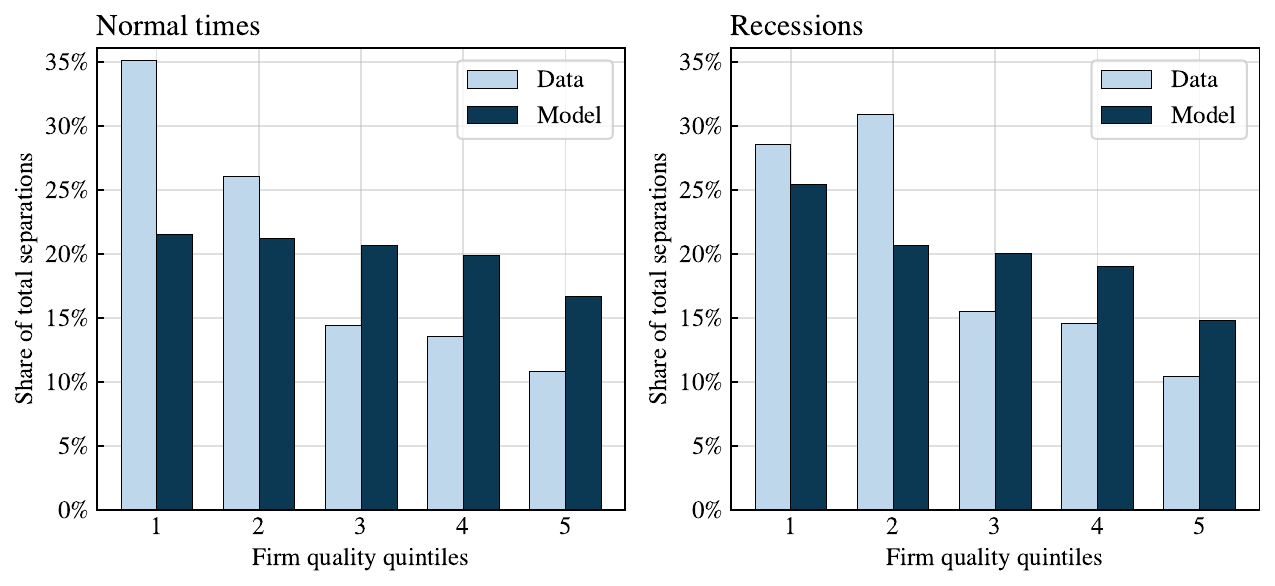}
\vspace{-2em}    
\floatfoot{\textbf{Note:} The figure reports share of separations in each quintile of the firms' productivity distributions. In the data, recessions are OECD recession periods, while in the model they are periods with at least three quarters of negative GDP growth.}
\end{figure}

\begin{figure}%
    \caption{Cross-Sectional Features, Model and Data}\label{app_fig:bigvalunemp}
    \subcaptionbox{Unemployment rates\label{app_fig:unemployment}}{\includegraphics[width=0.49\textwidth]{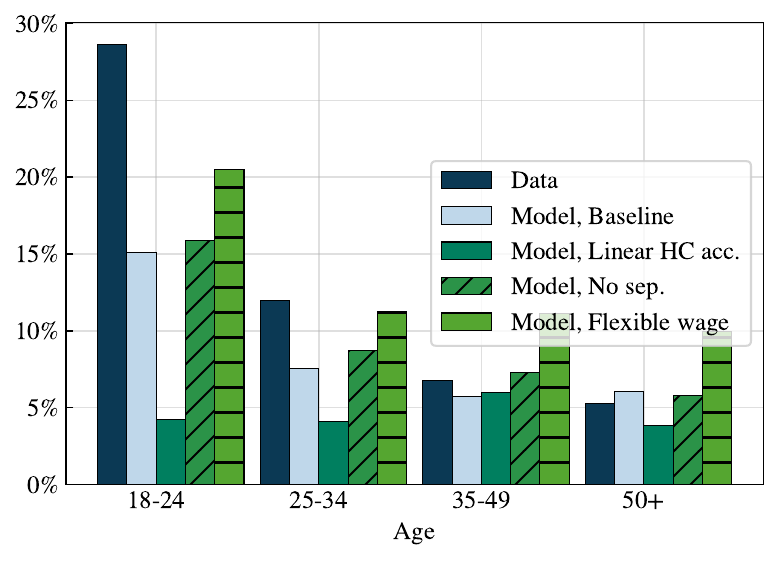}}
    \subcaptionbox{Unemployment-to-Employment rates\label{app_fig:ue_model_data}}
    {\includegraphics[width=0.49\textwidth]{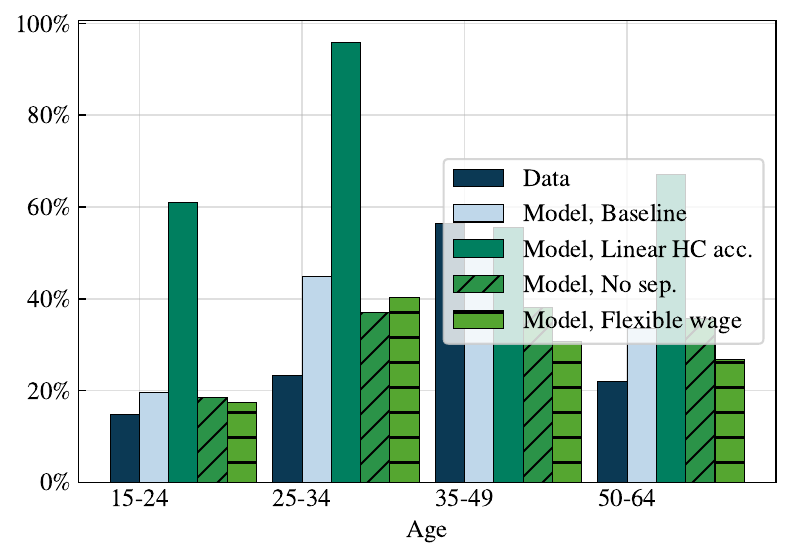}}
    \vspace{-1em}
    \floatfoot{\textbf{Note:}  Panel~(\textbf{a}) plots the unemployment rate by age groups in model simulations and in the data. 
    \textit{Sources:} unemployment rates are taken from the Italian National Statistical Agency (ISTAT). Panel~(\textbf{b}) reports the average UE rates by worker age.}
\end{figure}

\begin{figure}
\caption{Experience returns, Model and Data\label{app_fig:abs_returns}}
\includegraphics[height=0.4\textheight]{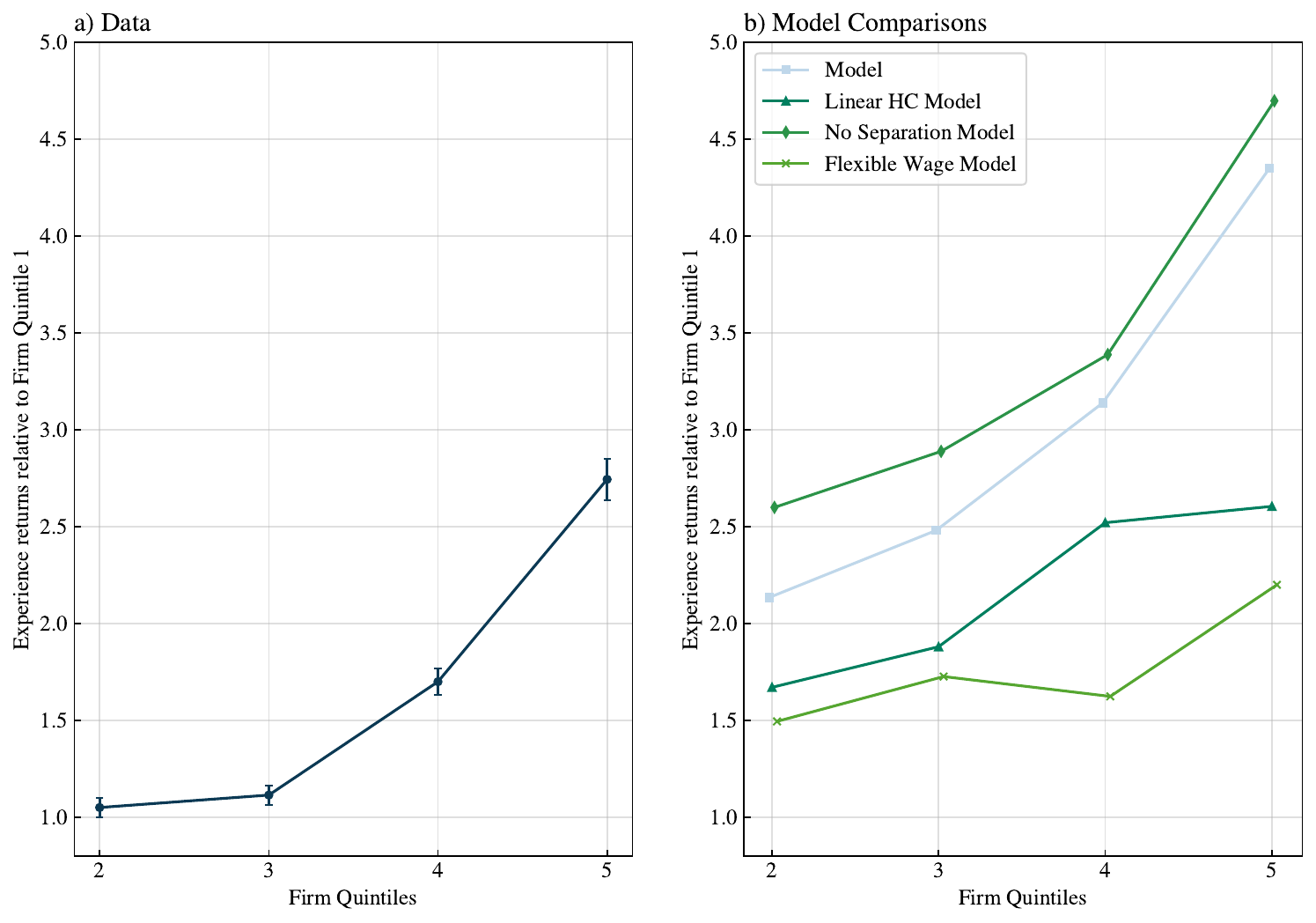}
\floatfoot{\textbf{Note}: The figure plots the returns to experience from working one year in a firm belonging to a given class. Panel a) reports the average returns to experiences from \textbf{Equation~\ref{app_eqn:experience_profiles_abs}}. Panel b) reports the coefficients from running the same regression on model simulations. We report the coefficients for the baseline model, a model linear human capital accumulation, one with flexible wages, and one with no separations.}
\end{figure}

\begin{figure}
\caption{Persistent effect of previous employers on post-unemployment wages\label{app_fig:eue_comparisons}}
\includegraphics[width=0.7\textwidth]{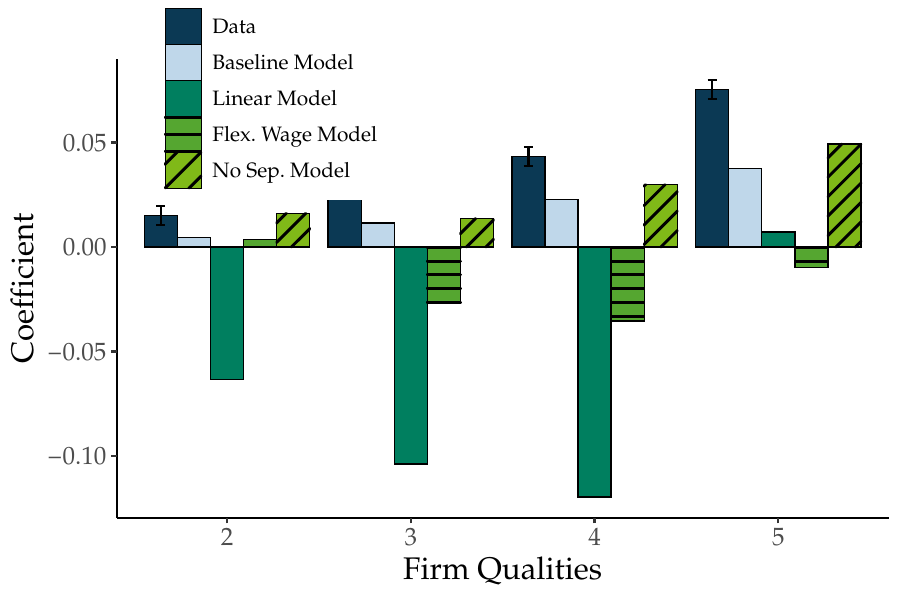}
\floatfoot{\textbf{Note}: The figure plots the persistency of previous employers on post-unemployment wages for a sample of workers experiencing an E-U-E transition. The coefficients are estimated following \textbf{Equation~\ref{eq:eue_reg}}. We compare the empirical estimates with a sample of E-U-E transitions in model simulations and we report the estimated coefficients for the baseline model, a model linear human capital accumulation, one with flexible wages, and one with no separations.}
\end{figure}

\begin{table}[t]
\centering
\caption{Co-movements at business cycle frequency \label{app_tab:validation_table}}
\resizebox{0.95\textwidth}{!}{\input{figures/PaperFigures/ValidationTable_AllModels.txt}}
\floatfoot{\textbf{Note:} \footnotesize{We report the correlation between model simulated series and their data counterparts (Panel~A), the correlations with aggregate output in the model and in the data (Panel~B) and the standard deviations of output and unemployment (Panel~C). Simulations are obtained by feeding the model with a TFP series  that matches the Italian TFP from 2000Q1 to 2019Q4. All series have been  detrended.  Panel~A and~B: Hamilton filter (4 lags and 8 leads); Panel~C: Hodrick-Prescott filter (smoothing equal to $10^5$, as in \cite{Shimer2005}). We report the correlations for our baseline model and for two alternative ones. One in which human capital accumulation does not depend on firm quality and is given by $h_{\iota,t+1} = \phi_{\iota} + h_{\iota, t} + \epsilon_{\iota,t}$ (\textit{Linear Human Capital}), and a second in which wages can decrease in order to satisfy the firm participation constraint (\textit{Flexible Wage}). All the alternative models have been re-estimated as described in \textbf{Section \ref{calib_est}}.}}
\end{table}

\begin{table}[t]
\centering
\caption{Cyclical job ladder collapse \label{app_tab:differentialnetjob}}
\resizebox{0.7\textwidth}{!}{\input{figures/PaperFigures/HHMS_comparisons.tex}}
\vspace{-1.0em}
\floatfoot{\textbf{Note:} The table reports the cyclicality of the difference in Net poaching and Net nonemployment flows as in  \cite{haltiwanger_cyclical_2025} and \cite{haltiwanger_cyclical_2018}. Each cell presents results from a separate regression. The dependent variable is the differential worker flow rate between high- and low-productivity firms. The independent variables in each regression include a cyclical indicator as well as a linear and quadratic time trend and a constant, which are not reported. In the data, the sample is quarterly from 1996 to 2018. Robust standard errors in parentheses.}
\end{table}

\begin{figure}
    \centering
    \caption{Decomposing output's losses across the age distribution\label{fig:decompByAge_irf}}
    \subcaptionbox{Baseline model\label{fig:decompByAge_irf_output}}{
    \includegraphics[width=0.8\textwidth]{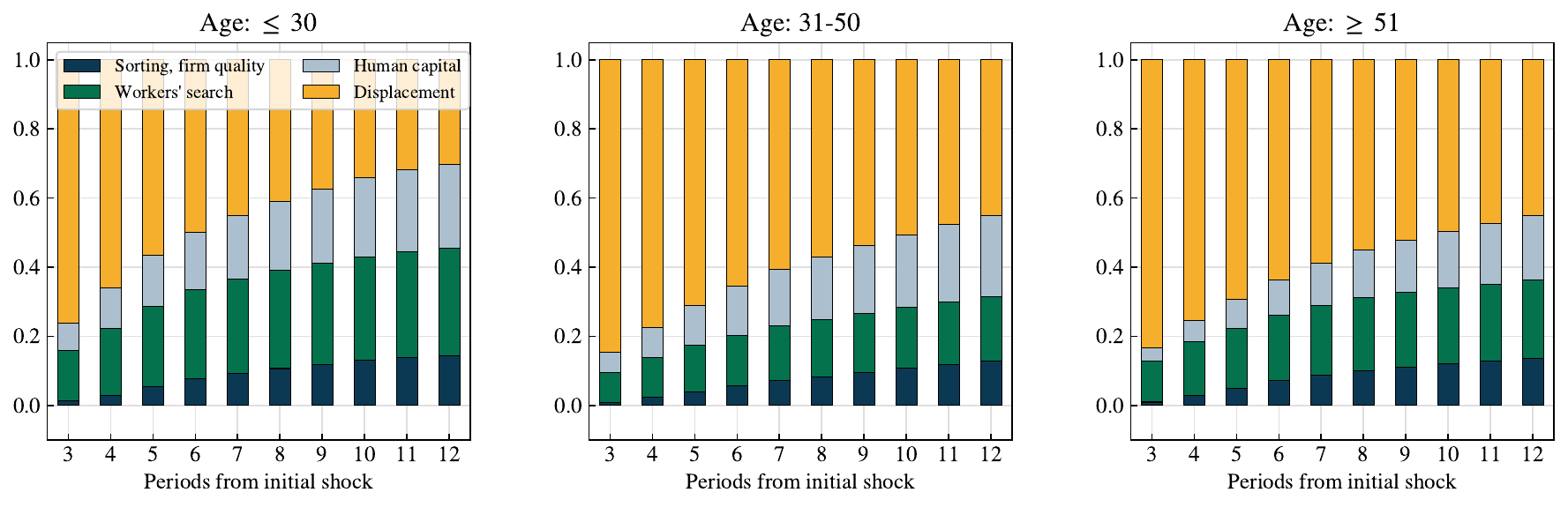}}
    
    \subcaptionbox{Linear human capital model\label{fig:decompByAge_irf_linear}}{\includegraphics[width=0.8\textwidth]{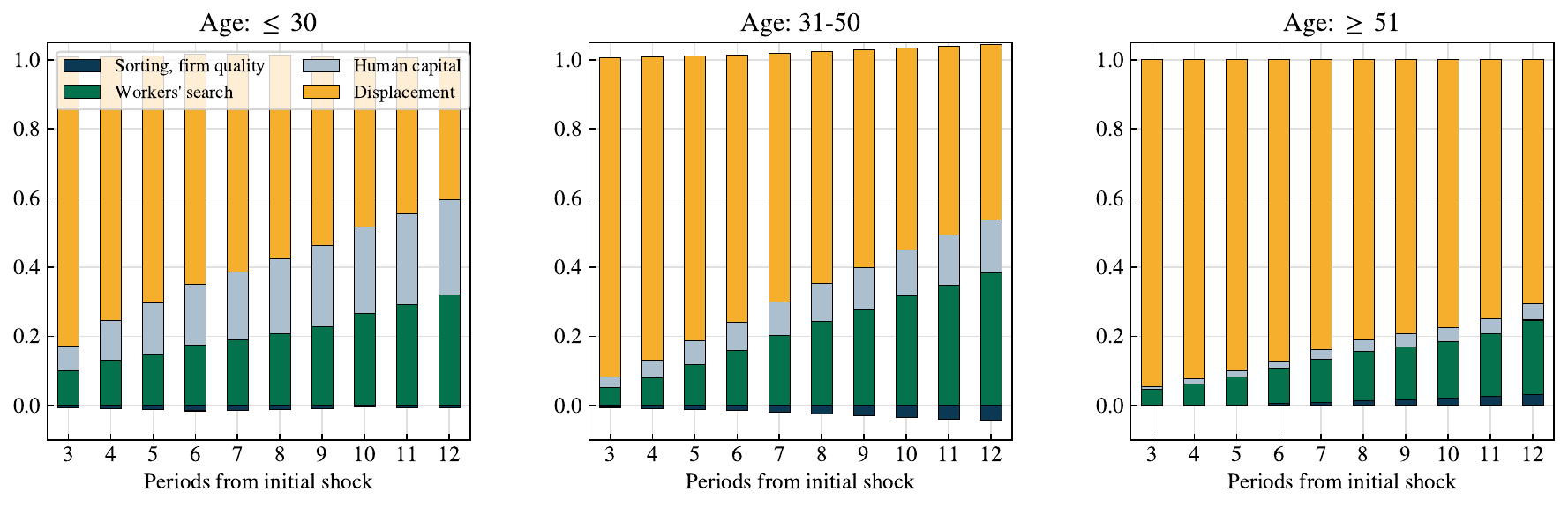}}
    \vspace{-1.0em}
    \floatfoot{\textbf{Note:} For each age group, Panel~(\textbf{a}) shows the relative importance of each transmission channel compared to the baseline recession for the cumulative impulse response of GDP across different age groups in the two years after the onset of the shock. Panel~(\textbf{b}) plots the same decomposition for the model with linear human capital accumulation.}
\end{figure}

%% file: figures/PaperFigures/StateDepTable_raw_avg.tex
\begin{tabular}{lccc}
\toprule
{} & \multicolumn{2}{c}{Output Response} & Persistence \\
\cmidrule{2-3}
& Short-term & Medium-term & {} \\
\midrule
Firm Quality    & -0.18      & -0.30     & 1.30         \\
                & (0.40)     & (0.91)    & (1.70)       \\
Human Capital   & -0.39      & -0.88     & 1.21         \\
                & (3.37)     & (7.73)    & (14.42)      \\
Wage            & -0.05      & -0.26     & -0.81        \\
                & (0.51)     & (1.18)    & (2.19)       \\
Labor Share     & -2.11      & -3.52     & -2.12        \\
                & (1.15)     & (2.65)    & (4.93)       \\
Sorting         & -1.14      & -2.56     & -7.03        \\
                & (1.71)     & (3.93)    & (7.32)       \\
\midrule
$R^2$ & 0.8 & 0.8 & 0.3 \\
N & 100 & 100 & 100 \\
\bottomrule
\end{tabular}

%% file: figures/PaperFigures/StateDepTable_raw_var.tex
\begin{tabular}{lccc}
\toprule
{} & \multicolumn{2}{c}{Output Response} & Persistence \\
\cmidrule{2-3}
& Short-term & Medium-term & {} \\
\midrule
Firm Quality    & -0.07      & -0.17     & 0.90         \\
                & (0.20)     & (0.46)    & (0.86)       \\
Human Capital   & 0.41       & 0.55      & -2.30        \\
                & (3.09)     & (7.09)    & (13.22)      \\
Wage            & 0.71       & 2.05*     & -0.62        \\
                & (0.52)     & (1.20)    & (2.23)       \\
Labor Share     & -1.43      & -7.75     & 9.23         \\
                & (6.42)     & (14.73)   & (27.46)      \\
                &      &       &   \\
                &      &      &       \\
\midrule
$R^2$ & 0.8 & 0.8 & 0.3 \\
N & 100 & 100 & 100 \\
\bottomrule
\end{tabular}

%% file: figures/PaperFigures/StateDepTable_raw_ske.tex
\begin{tabular}{lccc}
\toprule
{} & \multicolumn{2}{c}{Output Response} & Persistence \\
\cmidrule{2-3}
& Short-term & Medium-term & {} \\
\midrule
Firm Quality    & -0.09      & -0.06     & 0.64         \\
                & (0.34)     & (0.78)    & (1.45)       \\
Human Capital   & -0.18      & -0.44     & -2.87        \\
                & (2.13)     & (4.88)    & (9.10)       \\
Wage            & -0.10      & -0.26     & -0.03        \\
                & (0.13)     & (0.30)    & (0.56)       \\
Labor Share     & 0.34       & 1.03      & 3.34         \\
                & (0.73)     & (1.68)    & (3.13)       \\
                &      &       &   \\
                &      &      &       \\
\midrule
$R^2$ & 0.8 & 0.8 & 0.3 \\
N & 100 & 100 & 100 \\
\bottomrule
\end{tabular}

%% file: figures/PaperFigures/StateDepTable_raw_kur.tex
\begin{tabular}{lccc}
\toprule
{} & \multicolumn{2}{c}{Output Response} & Persistence \\
\cmidrule{2-3}
& Short-term & Medium-term & {} \\
\midrule
Firm Quality    & 0.01       & -0.02     & -0.13        \\
                & (0.13)     & (0.29)    & (0.55)       \\
Human Capital   & -0.26      & -0.51     & 1.19         \\
                & (0.47)     & (1.08)    & (2.01)       \\
Wage            & 0.08       & 0.21      & 0.38         \\
                & (0.07)     & (0.17)    & (0.31)       \\
Labor Share     & -0.20      & -0.55     & -1.04        \\
                & (0.28)     & (0.65)    & (1.20)       \\
                &      &       &   \\
                &      &      &       \\
\midrule
$R^2$ & 0.8 & 0.8 & 0.3 \\
N & 100 & 100 & 100 \\
\bottomrule
\end{tabular}

%% file: figures/PaperFigures/ValidationTable_AllModels.txt
\begin{tabular}{lccccc}
  \multicolumn{6}{l}{\textbf{(a)} {Correlation between Model and Data Series} } \\ 
  \midrule 
  {} & & \multicolumn{4}{c}{Model} \\
  \cmidrule{3-6}                 
   {} & & Baseline & Linear Human Capital Acc. & No Separations & Flexible Wages \\
  \midrule
    Aggregate output & & 0.81 & 0.82     & 0.81  & 0.81     \\
    Unemployment     & & 0.44   & 0.01       & 0.47    & 0.37       \\
  \midrule
  \\[0.25em]
  \multicolumn{6}{l}{\textbf{(b)} {Correlation with aggregate output} } \\
  \midrule
  {} & Data & \multicolumn{4}{c}{Model} \\
  \cmidrule{3-6}
  {} & {} & Baseline & Linear Human Capital Acc. & No Separations & Flexible Wage \\
  \midrule
    Unemployment & -0.66 & -0.56 & -0.01 & -0.65 & -0.45 \\
  \midrule
  \\[0.25em]
  \multicolumn{6}{l}{\textbf{(c)} {Unemployment and Output Volatility (in \%)} } \\
 {} & Data & \multicolumn{4}{c}{Model} \\
  \cmidrule{3-6}
  {} & {} & Baseline & Linear Human Capital Acc. & No Separations & Flexible Wage \\
  \midrule
    Aggregate Output SD & 2.4\% & 3.7\% & 2.5\% & 4.1\% & 2.8\% \\
    Unemployment SD     & 1.4\% & 0.9\% & 0.3\% & 1.3\% & 0.8\% \\
  \midrule
  \end{tabular}
  

%% file: figures/PaperFigures/HHMS_comparisons.tex
\begin{tabular}{l*{2}{c}}
\toprule
                    &\multicolumn{1}{c}{(1)}&\multicolumn{1}{c}{(2)}\\
                    &\multicolumn{1}{c}{Net poaching}&\multicolumn{1}{c}{Net nonemployment}\\
\midrule
\multicolumn{3}{l}{\textit{Panel A: Data}}\\[0.5ex]
Delta. unempl. rate &      -1.366 &       1.530\\
                    &     (0.557)         &     (0.540)         \\

\multicolumn{3}{l}{\textit{Panel B: Baseline Model}}\\[0.5ex]
Delta. unempl. rate &      -1.599 &       0.452\\
                    &     (0.415)         &     (0.124)         \\

\multicolumn{3}{l}{\textit{Panel C: Linear Model}}\\[0.5ex]
Delta. unempl. rate &      -0.251 &       0.035\\
                    &     (0.186)         &     (0.116)         \\
                    
\multicolumn{3}{l}{\textit{Panel D: Flex. Wage Model}}\\[0.5ex]
Delta. unempl. rate &      -0.226 &       0.059\\
                    &     (0.267)         &     (0.097)         \\
\multicolumn{3}{l}{\textit{Panel E: No Separations Model}}\\[0.5ex]
Delta. unempl. rate &      -2.039  &       0.072\\
                    &     (0.310)         &     (0.080)         \\
\midrule
Observations (Data)        &          92         &          92         \\
Observations (Models)        &          418         &          418         \\
\bottomrule
\end{tabular}